\documentclass[superscriptaddress, reprint, aps, prx]{revtex4-2}

\pdfoutput=1

\usepackage{amsmath}
\usepackage{amssymb}
\usepackage{amsthm}
\usepackage{mathtools}
\usepackage{natbib}
\usepackage{url}
\usepackage{textcase}
\usepackage{bm}
\usepackage{float}
\usepackage[T1]{fontenc}
\usepackage{graphicx}
\usepackage{dcolumn}
\usepackage{hyperref}
\usepackage{xcolor}
\usepackage{soul}

\graphicspath{{./images/}}
\hypersetup{
    colorlinks,
    linkcolor={blue!90!black},
    citecolor={blue!90!black},
    urlcolor={blue!90!black}
}

\newtheorem{conjecture}{Conjecture}
\newtheorem{corollary}{Corollary}
\newtheorem{lemma}{Lemma}

\newtheorem{result}{Result}

\theoremstyle{definition}
\newtheorem{definition}{Definition}
\newtheorem{step}{Step}[subsection]

\makeatletter
\newtheorem*{rep@theorem}{\rep@title}
\newcommand{\newreptheorem}[2]{%
    \newenvironment{rep#1}[1]{%
        \def\rep@title{#2 \ref*{##1}}%
        \begin{rep@theorem}}%
    {\end{rep@theorem}}}
\makeatother

\newreptheorem{conjecture}{Conjecture}

\usepackage{thmtools}
\usepackage{thm-restate}

\newcommand{\ket}[1]{| #1 \rangle}
\newcommand{\bra}[1]{\langle #1 |}
\newcommand{\braket}[2]{\langle #1 | #2 \rangle}
\newcommand{\ketbra}[2]{| #1 \rangle \langle #2 |}
\newcommand{\ev}[1]{\langle #1 \rangle}
\newcommand{\fv}[3]{\langle #1 | #2 | #3 \rangle}

\DeclareMathOperator{\fec}{vec}
\DeclareMathOperator{\rank}{rank}
\DeclareMathOperator{\spn}{span}
\DeclareMathOperator{\tr}{tr}
\DeclareMathOperator{\Wg}{Wg}

\let\oldl\left
\let\oldr\right
\renewcommand{\left}{\mathopen{}\mathclose\bgroup\oldl}
\renewcommand{\right}{\aftergroup\egroup\oldr}

\newcommand{\xone}{\xi_{1 \mathrm{D}}}
\newcommand{\xtwo}{\xi_{2 \mathrm{D}}}

\begin{document}

\title{Typical Correlation Length of Sequentially Generated Tensor Network States}

\author{Daniel Haag}

\email{daniel.haag@tum.de}

\affiliation{Max-Planck-Institut f{\"{u}}r Quantenoptik, Hans-Kopfermann-Str.~1, 85748 Garching, Germany}

\affiliation{Physik-Department, Technische Universit{\"{a}}t M{\"{u}}nchen, James-Franck-Str.~1, 85748 Garching, Germany}

\affiliation{Munich Center for Quantum Science and Technology (MCQST), Schellingstr.~4, 80799 M{\"{u}}nchen, Germany}

\author{Flavio Baccari}

\affiliation{Max-Planck-Institut f{\"{u}}r Quantenoptik, Hans-Kopfermann-Str.~1, 85748 Garching, Germany}

\affiliation{Munich Center for Quantum Science and Technology (MCQST), Schellingstr.~4, 80799 M{\"{u}}nchen, Germany}

\author{Georgios Styliaris}

\affiliation{Max-Planck-Institut f{\"{u}}r Quantenoptik, Hans-Kopfermann-Str.~1, 85748 Garching, Germany}

\affiliation{Munich Center for Quantum Science and Technology (MCQST), Schellingstr.~4, 80799 M{\"{u}}nchen, Germany}

\date{\today}

\begin{abstract}
    The complexity of quantum many-body systems is manifested in the vast diversity of their correlations, making it challenging to distinguish the generic from the atypical features. This can be addressed by analyzing correlations through ensembles of random states, chosen to faithfully embody the relevant physical properties. Here, we focus on spins with local interactions, whose correlations are extremely well captured by tensor network states. Adopting an operational perspective, we define ensembles of random tensor network states in one and two spatial dimensions that admit a sequential generation. As such, they directly correspond to outputs of quantum circuits with a sequential architecture and random gates. In one spatial dimension, the ensemble explores the entire family of matrix product states, while in two spatial dimensions, it corresponds to random isometric tensor network states. We extract the scaling behavior of the average correlations between two subsystems as a function of their distance. Using elementary concentration results, we then deduce the typical case for measures of correlation such as the von Neumann mutual information and a measure arising from the Hilbert-Schmidt norm. We find for all considered cases that the typical behavior is an exponential decay (for both one and two spatial dimensions). We observe the consistent emergence of a correlation length that depends only on the underlying spatial dimension and not the considered measure. Remarkably, increasing the bond dimension leads to a higher correlation length in one spatial dimension but has the opposite effect in two spatial dimensions.
\end{abstract}

\maketitle

\onecolumngrid

\begin{figure}[t]
	\centering
	\includegraphics{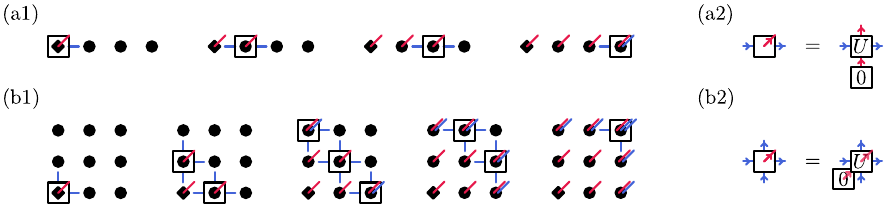}
	\caption{Sequential generation of MPS and isoTNS. Each circle represents a site of the finalized state and boxes represent the isometries of the sequential generation. (a1) Sequential generation of an MPS with physical dimension \( d \) and bond dimension \( D \). The diamond indicates the origin of the process. (a2) Each isometry arises from a unitary matrix of \( U (d D) \), with input and output as shown. The (blue) ancillary system is initialized at the first step of the sequential generation, transferred along the process, and accumulated at the final step. (b1) Sequential generation of an isoTNS with physical dimension \( d \) and bond dimension \( D \). In addition to indicating the origin of the process, the diamond also indicates the orthogonality center of the isoTNS. (b2) Each isometry arises from a unitary matrix of \( U \left ( d D^2 \right ) \). Ancillary systems are initialized and eventually accumulated at the boundary of the isoTNS at different steps of the sequential generation.}
	\label{fig:sequential_generation}
\end{figure}

\twocolumngrid

\section{INTRODUCTION}

The behavior of correlations in quantum many-body systems is an inherently difficult problem to characterize. Specifying a generic \( n \)-particle state requires exponentially many parameters, a fact which reflects the enormous variety of correlations possible in the quantum realm. Nonetheless, significant insights can be gained about the nature of correlations by utilizing random ensembles of states. A celebrated result along this direction shows that random states sampled uniformly from the full Hilbert space of an \( n \)-particle system typically exhibit strong correlations, as manifested by a volume-law behavior of the entanglement entropy~\cite{Page1993, Foong1994, SanchezRuiz1995, Hayden2006, Bengtsson2006}. However, there is by now clear evidence that the set of physically relevant states constitutes an exponentially small subset of the full Hilbert space of an \( n \)-particle system~\cite{Poulin2011}, bringing into question the relevance and utility of conclusions obtained under the assumption of uniform sampling from the full Hilbert space.

For quantum spin systems with local interactions, tensor network states have been exceedingly successful at capturing relevant properties~\cite{Cirac2021}. They exhibit an area law for the entanglement entropy by construction and are, therefore, good candidates to represent many physically relevant many-body states. Their pre-eminent one-dimensional representatives, matrix product states (MPS), have been shown to represent faithfully ground states of gapped local Hamiltonians~\cite{Verstraete2006, Hastings2007, Arad2013} and have given rise to the complete classification of topological phases of matter in one dimension~\cite{Chen2011, Schuch2011}. MPS have been generalized to their counterparts in two (or more) spatial dimensions, projected entangled pair states (PEPS). While only a weaker link between local Hamiltonians and PEPS has been proven rigorously, two-dimensional PEPS are known to efficiently represent a wide class of strongly correlated states~\cite{Cirac2021, PerezGarcia2008}, including states with power-law~\cite{Verstraete2006a} and topological correlations~\cite{Levin2005, Gu2009, Buerschaper2009}.

The importance of defining ensembles of random tensor network states for the purpose of exploring typical properties of physically relevant states has been recognized more than a decade ago~\cite{Garnerone2010}. MPS ensembles have been utilized to gain insights into, among other things, the typicality of expectation values of local observables~\cite{Garnerone2010, Garnerone2010a}, equilibration under Hamiltonian time evolution~\cite{Haferkamp2021}, the entropy of subsystems~\cite{Collins2013}, nonstabilizerness~\cite{Chen2022}, and, most relevant for this work, the behavior of correlations~\cite{Movassagh2017, Movassagh2022, Lancien2021, Bensa2023, Svetlichnyy2022}. In particular, correlation functions of random inhomogeneous MPS (that is, MPS whose local tensors can be different) were shown to exhibit almost surely an exponential decay~\cite{Movassagh2021, Movassagh2022}. A qualitatively similar behavior was observed also for correlation functions of translation-invariant MPS and PEPS with random Gaussian entries~\cite{Lancien2021}. Instead of incorporating the randomness directly at the level of states, one can also consider random local Hamiltonians and examine their ground states. The typical behavior of correlations for this case was found to depend on the nature of randomness, allowing for both long and short-range correlated states~\cite{Movassagh2017, Lemm2019, Jauslin2022}.

Here, we approach the problem of typical correlations in random MPS and PEPS from a more operational point of view. We introduce families of inhomogeneous random tensor network states that arise from a sequential generation in a quantum computer. Such ensembles are, by definition, directly connected to the study of quantum circuits with a sequential architecture and random gates, where each unitary gate is independently chosen randomly from the uniform (Haar) measure. In the one-dimensional case, every MPS admits such a preparation~\cite{Schoen2005}, where the bond dimension dictates the number of overlapping qudits between any two successive gates. In the two-dimensional setting, our ensemble can be understood as being uniform over the space of so-called isometric tensor network states (isoTNS)~\cite{Zaletel2020}, which are PEPS with given bond dimension. In this case, the resulting family of random circuits is composed of two-dimensional circuits with local overlapping gates, each resembling a tile acting on a neighborhood of qudits~\cite{Wei2022}. Although isoTNS are only a subfamily of PEPS, they are known to contain a rich variety of strongly correlated states such as topological models~\cite{Soejima2020}.

For the above ensembles, we study the scaling behavior of the average correlations between two subsystems as a function of their distance. We then utilize this average behavior of correlations to deduce the typical case via concentration inequalities. Instead of using well-known correlation functions, we perform the analysis using a measure of correlation arising from the Hilbert-Schmidt norm. Although, in a generic many-body setting, such a measure might have undesirable properties, we show that it is particularly suited in the context of tensor network states because it bounds the trace distance as well as all connected correlation functions. For MPS, we also consider the R\'{e}nyi-\( \alpha \) mutual information. Given a technical conjecture, we compute the average correlations for all integer values of \( \alpha \geq 1 \). We then use those results to retrieve the von Neumann mutual information~\cite{Wilde2013} via analytic continuation.

We confirm analytically the common intuition that inhomogeneous MPS typically exhibit exponentially decaying correlations. We show that a single common correlation length \( \xone \) persists among different measures of correlation. We obtain a similar quantitative conclusion for two-dimensional isoTNS, where we observe the emergence of a different correlation length \( \xtwo \) that is also consistent among different measures of correlation. \( \xone \) and \( \xtwo \) have a surprisingly weak dependence on the underlying bond dimension \( D \) of the tensor network, and both MPS and isoTNS remain short-range correlated for all values of \( D \). However, \( \xone \) and \( \xtwo \) have exactly opposite behaviors when the bond dimension is varied; \( \xone \) monotonically increases, while \( \xtwo \) monotonically decreases. Our findings also establish that exponentially decaying correlations are typical for the family of (inhomogeneous) isoTNS and consequently for the random states produced by the corresponding quantum circuit architecture.

While an exponential decay of correlations is not an unexpected result for MPS, the weak dependence of the average correlation length on the bond dimension is particularly surprising. Intuitively, one might expect that a higher bond dimension allows for a more effective spreading of correlations. Our results show that that is not the case for most MPS. The starkly different behavior of \( \xone \) and \( \xtwo \) is also remarkable. While increasing the bond dimension affects the typical spreading of correlations only modestly in one dimension, it is detrimental to the average correlation length in two dimensions.

The paper is structured as follows. In Sec.~\ref{sec:preliminaries}, we introduce our families of sequentially generated tensor network states and the main technical tools required to compute their average properties. In Sec.~\ref{sec:summary}, we summarize our results for both MPS and isoTNS. In Secs.~\ref{sec:mps_results} and \ref{sec:iso_results}, we, respectively, discuss our findings for MPS and isoTNS in detail. Lastly, we devote Sec.~\ref{sec:conclusion} to final observations and potential followups to our work.

\section{PRELIMINARIES} \label{sec:preliminaries}

In this section, we introduce the main technical concepts that will be needed throughout the paper. In Sec.~\ref{sec:sequential_generation}, we review the relevant families of sequentially generated tensor network states in one and two dimensions. Sec.~\ref{sec:measures} is devoted to the measures of correlation we are interested to estimate. In Sec.~\ref{sec:twirl}, we explain how to compute averages with respect to the Haar measure. Lastly, Sec.~\ref{sec:graphical} introduces the graphical notation we will use to present and prove our results.
 
\subsection{Tensor network states} \label{sec:sequential_generation}

In one dimension, the pre-eminent tensor network structure are matrix product states (MPS)~\cite{Cirac2021}. An \( n \)-particle MPS with open boundary conditions and local (physical) dimension \( d \) is given by \begin{align} \label{eq:mps_normal}
	\ket{\psi} = \sum_{i_1, \dots, i_n} \fv{L}{A_{i_1}^{(1)} \cdots A_{i_n}^{(n)}}{R} \ket{i_1 \cdots i_n},
\end{align} where \( A^{(j)} \in \mathbb{C}^{D \times D} \), \( \ket{L} \in \mathbb{C}^D \) is the left boundary condition, and \( \ket{R} \in \mathbb{C}^D \) is the right boundary condition. \( D \) is called the bond dimension of the MPS. A commonly used graphical notation for Eq.~\eqref{eq:mps_normal} is \begin{align}
	\ket{\psi} = \raisebox{-7pt}{\includegraphics{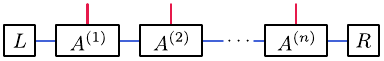}},
\end{align} where vertical (red) legs represent physical space indices \( \left ( \mathbb{C}^d \right ) \), and horizontal (blue) legs represent bond space indices \( \left ( \mathbb{C}^D \right ) \).

From its definition, it might not be evident how an MPS can be generated because each tensor \( A \) does not necessarily correspond to a physical process. However, the representation of an MPS in terms of tensors is not unique. This can be resolved by imposing a convenient canonical form~\cite{Schollwock2011}. Any MPS in such a canonical form can be seen as a state generated sequentially by applying unitary matrices \( U^{(1)}, \dots, U^{(n)} \in U (d D) \) to a product state initialized in \( \ket{0}^{\otimes n} \) for the physical space and \( \ket{0} \) for the bond space~\cite{Schoen2005}. The resulting state is given by \begin{align} \label{eq:mps_unitaries}
	\ket{\psi} = \raisebox{-32pt}{\includegraphics{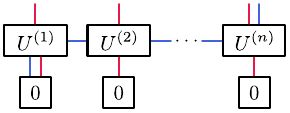}}.
\end{align} Note that the final site has dimension \( d D \), while all other sites have dimension \( d \). As we will see later, the final site will not play a significant role in our analysis, making its different dimension not an issue. In Fig.~\ref{fig:sequential_generation}~(a1), we sketch an equivalent representation of sequential generation in terms of isometries instead of unitary matrices. 

The family of MPS is thus equivalent to states resulting from quantum circuits that have a sequential architecture and act on input product states. The architecture is a consequence of the connectivity of the MPS network [see Fig.~\ref{fig:sequential_generation}~(a1)]. In this picture, larger bond dimensions translate to wider gates, each acting on \( 1 + \lceil \log_d (D) \rceil \)qudits. For example, for \( D = d^2 \), one has \begin{align}
	\ket{\psi} & = \raisebox{-32pt}{\includegraphics{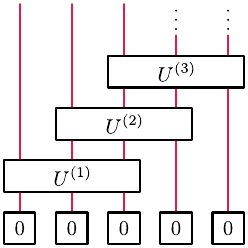}} \\
	& = \raisebox{-32pt}{\includegraphics{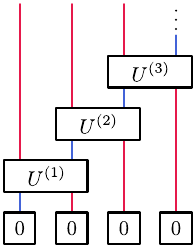}}.
\end{align} Naturally, using this correspondence, all of our results can be translated to the language of quantum circuits with the described architecture.

Projected entangled-pair states (PEPS) are the generalization of MPS to two (or more) dimensions~\cite{Cirac2021}. Because no simple generalization of the sequential generation of MPS to arbitrary PEPS is known, we restrict ourselves to the rich family of two-dimensional isometric tensor network states (isoTNS), which were first defined in Ref.~\cite{Zaletel2020} (see also Ref.~\cite{Haghshenas2019}).

Much like MPS, isoTNS can be generated sequentially by applying unitary matrices to a product state initialized in \( \ket{0}^{\otimes m n} \) for the physical space~\cite{Wei2022}, where \( m \) denotes the number of rows and \( n \) the number of columns of the underlying rectangular lattice. We will use the sequential generation sketched in Fig.~\ref{fig:sequential_generation}~(b1), which is a generalization of the one proposed in Ref.~\cite{Wei2022}. Each box corresponds to an isometry that arises from a unitary matrix \( U^{(i, j)} \in U \left ( d D^2 \right ) \). In particular, isometries in the bulk can be drawn as \begin{align}
    \raisebox{-24pt}{\includegraphics{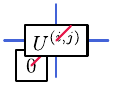}},
\end{align} as indicated in Fig.~\ref{fig:sequential_generation}~(b2). The diamond in Fig.~\ref{fig:sequential_generation}~(b1) indicates the so-called orthogonality center of the isoTNS. Its row and column constitute the orthogonality hypersurface, which can be treated like an MPS. That is, if an operator is supported only on the orthogonality hypersurface, its expectation value with respect to the isoTNS reduces to that of the underlying MPS~\cite{Zaletel2020}. Although isoTNS of a given bond dimension form by definition only a subclass of PEPS, they are known to contain states with a rich structure of correlations, such as topological models~\cite{Soejima2020}. On top, their properties make isoTNS a suitable candidate for studying correlations analytically, which is otherwise a generally challenging task in more than one dimension.

IsoTNS correspond to quantum circuits on a two-dimensional grid with local overlapping gates, which now resemble tiles. Increasing the bond dimension translates to larger tile sizes and overlaps, as in the MPS case. The corresponding architecture is dictated by the connectivity of the isoTNS network [see Fig.~\ref{fig:sequential_generation}~(b1)], and it is tedious (although straightforward), which is why we refer the reader to Ref.~\cite{Wei2022} for details.

\subsection{Quantifying correlations} \label{sec:measures}

Correlations express that knowledge about one subsystem can convey information about another. They are quantified by different measures that frequently arise from an information-theoretic perspective and are based on operationally motivated tasks. A prime example is the von Neumann mutual information~\cite{Wilde2013} \begin{align}
    I (A : B) = S (\varrho_A) + S (\varrho_B) - S (\varrho_{A B}),
\end{align} where \begin{align}
	S (\varrho) = - \tr [\varrho \log (\varrho)]
\end{align} is the von Neumann entropy. \( A \) and \( B \) are two disjoint subsystems of a larger system, and \( \varrho_A \) and \( \varrho_B \) denote the marginals of \( \varrho_{A B} \). The von Neumann mutual information captures the total (classical and quantum) amount of correlations between \( A \) and \( B \) as it is equal to the minimum rate of randomness required to asymptotically turn \( \varrho_{A B} \) into a product state~\cite{Groisman2005}. It is also a nonnegative quantity and nonincreasing under local operations~\cite{Wilde2013}, both of which are desirable properties for a measure of correlation. The latter means that a quantum channel~\cite{Wilde2013} acting on \( A \) or \( B \) alone (for example, by discarding part of a subsystem) cannot increase \( I (A : B) \). Unfortunately, the analytical treatment of the von Neumann mutual information is impractical because computing the logarithm of \( \varrho \) generally requires the knowledge of its full spectrum.

To overcome this issue, an alternative is to consider a particular R\'{e}nyi-\( \alpha \) generalization of the mutual information \begin{align} \label{eq:renyi}
	I_\alpha (A : B) = S_\alpha (\varrho_A) + S_\alpha (\varrho_B) - S_\alpha (\varrho_{A B}),
\end{align} where \begin{align}
	S_\alpha (\varrho) = \frac{1}{1 - \alpha} \log [\tr (\varrho^\alpha)]
\end{align} is the R\'{e}nyi-\( \alpha \) entropy. As is apparent from the definition, for integer values of \( \alpha \), its evaluation is considerably simpler. The R\'{e}nyi-\( \alpha \) mutual information has been investigated in the context of conformal field theories~\cite{Asplund2014, Agon2016, Chen2019}, free fermions~\cite{Bernigau2015}, and quantum dynamics~\cite{Murciano2022, McGinley2022}. We will use later that the limit of \( \alpha \to 1 \) of \( I_\alpha (A : B) \) recovers the von Neumann mutual information. The mentioned positive aspects notwithstanding, unlike the von Neumann mutual information, Eq.~\eqref{eq:renyi} does not arise from a (generalized) divergence~\cite{Khatri2020}, and the R\'{e}nyi-\( \alpha \) mutual information can be negative~\cite{Kormos2017, Scalet2021} and increasing under local operations. It is thus hard to interpret it as a proper measure of correlation in general. Nevertheless, for certain families of initial states (see, for example, Ref.~\cite{McGinley2022}), monotonicity and nonnegativity can be restored. Henceforth, we will mostly focus on the case of \( \alpha = 2 \), but we will also consider an analytic continuation on positive integer values of \( \alpha \). As we will show, in the present context of tensor network states, the case of \( \alpha = 2 \) appropriately captures the decay of correlations at large distances between subsystems \( A \) and \( B \) with little effort.

In addition to the previous quantities, we would also like to probe the trace distance \begin{align}
    T (A : B) = \frac{1}{2} \left \| \varrho_{A B} - \varrho_A \otimes \varrho_B \right \|_1,
\end{align} where \( \| \cdot \|_p \) denotes the Schatten \( p \)-norm~\cite{Watrous2018}. For an operator \( X \), the Schatten \( p \)-norm is given by \begin{align}
    \| X \|_p = \tr \left [ \left ( X^\dagger X \right )^{p / 2} \right ]^{1 / p}.
\end{align}  \( T (A : B) \) has a well-known operational interpretation as it quantifies the optimal distinguishability between \( \varrho_{A B} \) and the product of its marginals \( \varrho_A \otimes \varrho_B \) by a two-element generalized global measurement~\cite{Nielsen2012}. Moreover, the trace distance upper bounds the (properly normalized) connected correlation function~\cite{Khatri2020}: \begin{align}
    T (A : B) \geq 2 \frac{\left | \ev{M_A \otimes M_B} - \ev{M_A} \ev{M_B} \right |}{\left \| M_A \right \|_\infty \left \| M_B \right \|_\infty}
\end{align} Although the bound can be tight, the two quantities are different whenever product measurements are ineffective in distinguishing \( \rho_{A B} \) from \( \rho_A \otimes \rho_B \), a fact used in quantum data hiding~\cite{DiVincenzo2002}.

As one expects from its operational interpretation, the trace distance satisfies the monotonicity property under local operations~\cite{Nielsen2012}. However, \( T (A : B) \) is usually hard to compute exactly. We will now argue that investigating \begin{align}
	N (A : B) & = \left \| \varrho_{A B} - \varrho_A \otimes \varrho_B \right \|_2^2
\end{align} meaningfully probes \( T (A : B) \) for tensor network states with constant (that is, size-independent) bond dimension, all while being much simpler to treat.

In general, for mixed many-body states, the two measures can have vast disagreement because it holds~\cite{Watrous2018} that \begin{align} \label{eq:schatten_bound}
    \| X \|_2 \leq \| X \|_1 \leq \sqrt{\rank (X)} \| X \|_2.
\end{align} Both bounds are tight, and the upper bound is saturated for \( X \propto I \). As such, for arbitrary mixed states of an exponentially large Hilbert space, the factor \( \rank (X) \) can render the upper bound useless. Crucially, in this work, we investigate (random) tensor network states with fixed bond dimension \( D \). Let \( \partial R \) denote the boundary of a system \( R \) and \( | \partial R | \) its size (number of sites). The ranks of \( \varrho_A \) and \( \varrho_B \) are, respectively, upper bounded by \( D^{| \partial A |} \) and \( D^{| \partial B |} \), and that of \( \varrho_{A B} \) is upper bounded by \( D^{| \partial A | + | \partial B |} \)~\cite{Cirac2021}. Thus, \begin{align}
	\rank \left ( \varrho_{A B} - \varrho_A \otimes \varrho_B \right ) \leq 2 D^{| \partial A | + | \partial B |},
\end{align} yielding the bound \begin{align} \label{eq:iso_norm_1}
   \frac{1}{2} \sqrt{N (A : B)} \leq T (A : B) \leq \sqrt{\frac{D^{| \partial A | + | \partial B |}}{2}} \sqrt{N (A : B)}.
\end{align}

For MPS (one dimension) with connected subsystems \( A \) and \( B \), the bound reads \begin{align} \label{eq:mps_norm_1}
   \frac{1}{2} \sqrt{N (A : B)} \leq T (A : B) \leq \frac{D^2}{\sqrt{2}} \sqrt{N (A : B)}.
\end{align} That is, the bound is independent of the sizes of \( A \) and \( B \), unlike in the generic case of Eq.~\eqref{eq:schatten_bound}. This suggests that, for reasonably small bond dimension, \( N (A : B) \) is a reliable probe of correlations [as quantified by \( T (A : B ) \)] between subsystems \( A \) and \( B \). We will expand on this point later.

\subsection{\texorpdfstring{\( k \)}{k}-fold twirl}
\label{sec:twirl}

Let \( X \) be an operator acting on \( \left ( \mathbb{C}^q \right )^{\otimes k} \). The \( k \)-fold twirl of \( X \) with respect to the Haar measure on the unitary group \( U (q) \) is defined~\cite{Collins2006, Roberts2017, Brandao2021} as \begin{align} \label{eq:twirl_before}
	\mathcal{T}_U^{(k)} (X) = \int \mathrm{d} U \, U^{\otimes k} X \left ( U^\dagger \right )^{\otimes k}.
\end{align} One can employ the Schur-Weyl duality for unitary groups to show~\cite{Collins2003, Collins2006} that \begin{align} \label{eq:twirl_after}
	\mathcal{T}_U^{(k)} (X) = \sum_{\sigma, \tau \in S_k} \Wg \left ( \sigma \tau^{- 1}, q \right ) P_\sigma^{(q)} \tr \left [ X \left ( P_\tau^{(q)} \right )^T \right ],
\end{align} where \begin{align}
	P_\pi^{(q)} : v_1 \otimes \cdots \otimes v_k \mapsto v_{\sigma^{- 1} (1)} \otimes \cdots \otimes v_{\sigma^{- 1} (k)}
\end{align} is the representation of \( \pi \in S_k \) on \( \left ( \mathbb{C}^q \right )^{\otimes k} \), where \( S_k \) is the symmetric group. \( \Wg \left ( \sigma \tau^{- 1}, q \right ) = \left ( G^{- 1} \right )_{\sigma \tau} \)~\footnote{Although the Weingarten matrix \( W = G^{- 1} \) exits only if \( k \leq q \)~\cite{Collins2021}, the Weingarten function can easily be extended to \( k > q \)~\cite{Collins2006}.} is the Weingarten function, where \( G \in \mathbb{R}^{k! \times k!} \) denotes the Gram matrix whose entries are given by \begin{align} \label{eq:gram}
	G_{\sigma \tau} = \tr \left [ P_\sigma^{(q)} \left ( P_\tau^{(q)} \right )^T \right ] = q^{\# \left ( \sigma \tau^{- 1} \right )}.
\end{align} Here, \( \# (\pi) \) counts the number of cycles in the decomposition of \( \pi \in S_k \) into disjoint cycles. Thus, \( \Wg (\pi, q) \) depends only on the conjugacy class of \( \pi \)~\cite{Collins2003}. In Appendix~\ref{app:twirl}, we show how to obtain Eq.~\eqref{eq:twirl_after} from Eq.~\eqref{eq:twirl_before} by using a result of Ref.~\cite{Collins2003}.

\subsection{Graphical notation} \label{sec:graphical}

In this section, we introduce the graphical notation used throughout this paper. To keep the images compact, we employ the operator-vector correspondence. Let \( \{ \ket{i} \} \) denote the standard basis of \( \mathbb{C}^q \). Then, the operator-vector correspondence~\cite{Watrous2018} is defined by \begin{align}
	\fec (\ketbra{i}{j}) = \ket{i} \otimes \ket{j}
\end{align} and extended linearly to the vector space at large.

Because we consider the standard (product) basis to be fixed, we do not distinguish between tensors (as multidimensional arrays) and their basis-independent counterparts (such as vectors and operators). Let \( X \) be an operator acting on \( \left ( \mathbb{C}^q \right )^{\otimes k} \). Using the operator-vector correspondence, we denote it by \begin{align}
	\raisebox{-7pt}{\includegraphics{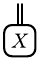}} = \fec (X).
\end{align} Note that the orientation of the legs does not have any meaning in our images. That is, \begin{align}
	\vcenter{\hbox{\includegraphics{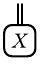}}} = \vcenter{\hbox{\includegraphics{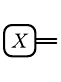}}} = \vcenter{\hbox{\includegraphics{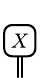}}} = \vcenter{\hbox{\includegraphics{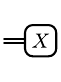}}}.
\end{align} When we need the transpose of an operator, we will explicitly use \begin{align}
	\raisebox{-7pt}{\includegraphics{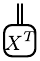}} = \fec \left ( X^T \right ).
\end{align} As such, when we contract two operators \( X \) and \( Y \), we mean the trace of their product: \begin{align}
	\vcenter{\hbox{\includegraphics{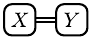}}} = \tr (X Y)
\end{align}

Let us state the two most prominent operators we will encounter. We will see \begin{align}
	\raisebox{-7pt}{\includegraphics{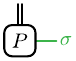}} = \fec \left ( P_{\sigma}^{(q)} \right ),
\end{align} where the horizontal (green) leg is permutation valued, and \begin{align}
	\raisebox{-7pt}{\includegraphics{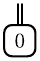}} = \fec \left ( \ketbra{0}{0}^{\otimes k} \right ).
\end{align} Their relevant contractions are \begin{align} \label{eq:permutation_contraction}
    \vcenter{\hbox{\includegraphics{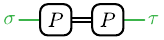}}} = \tr \left ( P_\sigma^{(q)} P_\tau^{(q)} \right ) = q^{\# (\sigma \tau)}
\end{align} and \begin{align} \label{eq:zero_contraction}
	\vcenter{\hbox{\includegraphics{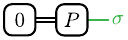}}} = \tr \left ( \ketbra{0}{0}^{\otimes k} P_\sigma^{(q)} \right ) = 1.
\end{align}

Moving forward, we will not explicitly write the operator \( \fec \), as it shall be clear from the context.

With the definition of the Weingarten matrix, \begin{align}
	\vcenter{\hbox{\includegraphics{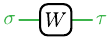}}} = \Wg \left ( \sigma \tau^{- 1}, q \right ),
\end{align} we can then write the \( k \)-fold twirl [see Eq.~\eqref{eq:twirl_after}] as \begin{align} \label{eq:twirl_graphical}
	\mathcal{T}_U^{(k)} (X) = \raisebox{-32pt}{\includegraphics{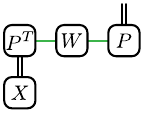}},
\end{align} where the contraction of two green legs corresponds to a summation over the permutations of \( S_k \).

\section{SUMMARY OF RESULTS} \label{sec:summary}

\begin{figure}
	\centering
	\includegraphics{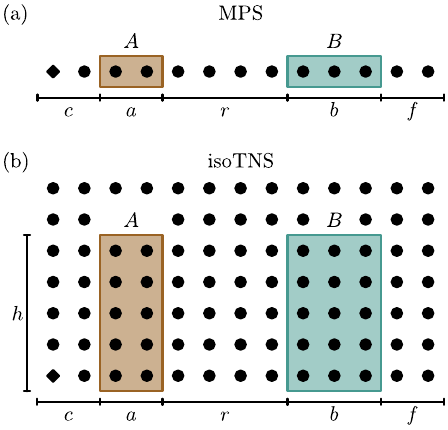}
	\caption{We investigate average correlations between two subsystems \( A \) and \( B \) as a function of their (horizontal) distance \( r \). \( A \) and \( B \) respectively stretch across \( a \) and \( b \) consecutive (horizontal) sites. (a) The diamond indicates the origin of the sequential generation of the MPS. (b) In addition to indicating the origin of the sequential generation, the diamond also indicates the orthogonality center of the isoTNS. For now, we restrict ourselves to \( A \) and \( B \) that touch the orthogonality hypersurface and stretch across \( h \) consecutive vertical sites.}
	\label{fig:setup}
\end{figure}

In this paper, we analyze the average behavior of correlations in random tensor network states. Through the average, we obtain conclusions about the generic case. Our work focuses on the disordered case, that is, the case where each local tensor is independent. Our setting can be equivalently understood as an investigation of correlations in states resulting from quantum circuits with a sequential architecture and random gates.

In one dimension, generic MPS are known to exhibit exponentially decaying correlations in the translation-invariant case~\cite{PerezGarcia2007}. This is due to the fact that injectivity is a generic property~\cite{Cirac2021}. On the other hand, injectivity alone is not enough to guarantee an exponential decay of correlations for an inhomogeneous sequence of tensors. Nevertheless, the exponential decay of correlations is widely expected to persist without translational invariance, but it has never been rigorously studied so far in this setting.

In two (or more) dimensions, the landscape of correlations is much richer. For instance, already in two dimensions, certain PEPS corresponding to thermal states of classical models are known to exhibit power-law correlations~\cite{Verstraete2006a}. Moreover, prominent topological states, such as quantum double models~\cite{Kitaev2003} (which include the toric code) and string-net models~\cite{Gu2009, Buerschaper2009}, admit a description in terms of PEPS. On the other hand, for translation-invariant PEPS whose tensors' entries are drawn from a Gaussian measure, it is known that correlations typically decay exponentially~\cite{Lancien2021}.

Computing correlations in higher-dimensional systems usually poses a significant challenge because they can be mediated through different paths connecting the two subsystems of interest. Here, we restrict our analysis to two-dimensional isoTNS. This rich class of tensor network states is relevant in both the analytical and the numerical context~\cite{MacCormack2021, Slattery2021, Hauru2021, Tepaske2021, Haghshenas2022}, all while admitting a simple physical interpretation through sequential generation [see Fig.~\ref{fig:sequential_generation}~(b1)]. Moreover, its mathematical properties make the analytical study of correlations in two dimensions tractable.

For isoTNS, it is expected that correlations between two subsystems decay exponentially if they are both on the orthogonality hypersurface because the calculation reduces to the contraction of an MPS~\cite{Zaletel2020}. Nonetheless, isoTNS can represent a rich variety of topological models, as all string-net models admit an exact and explicit description in terms of isoTNS~\cite{Soejima2020} (on the appropriate underlying lattice). This motivates us to study the typical behavior of correlations in isoTNS, particularly between subsystems that extend beyond the orthogonality hypersurface.

To investigate the decay of correlations in our two families of random tensor network states, one must specify the ensembles to draw from. Here we adopt an operational perspective and relate our measures of randomness directly to the sequential generation process. Because that is defined with respect to isometries, one can incorporate randomness at the level of the underlying unitary matrices. A natural choice is to draw each unitary matrix from the Haar measure on the appropriate unitary group. This approach was introduced for MPS in Ref.~\cite{Garnerone2010} (see also Ref.~\cite{Garnerone2010a}), and it can directly be applied to higher-dimensional tensor network states that admit a sequential generation, such as isoTNS, yielding normalized states by construction. Although one can sample random translation-invariant states with this method, we investigate the disordered case by drawing each unitary matrix independently from the Haar measure.

Because we are interested in the decay of correlations, we focus on computing average correlations between two subsystems \( A \) and \( B \) as a function of their distance \( r \). For random MPS and isoTNS, we consider subsystems \( A \) and \( B \) as sketched in Fig.~\ref{fig:setup}~(a) and Fig.~\ref{fig:setup}~(b), respectively. In both cases, \( A \) and \( B \) stretch across \( a \) and \( b \) consecutive (horizontal) sites. In Fig.~\ref{fig:setup}~(b), \( A \) and \( B \) touch the orthogonality hypersurface and stretch across \( h \) vertical sites. We will relax this condition later.

For all of the measures of correlation we study, we find that the average with respect to the considered ensemble of states decays exponentially. We formalize this type of behavior in Definition~\ref{def:definition}.

\begin{definition} \label{def:definition}
    Let \( \mathcal{C} (A : B) \) denote a measure of correlation. We say that the average of \( \mathcal{C} (A : B) \) with respect to a given ensemble of random states decays exponentially if \begin{align}
        \mathbb{E} \mathcal{C} (A : B) = K \exp \left ( - \frac{r}{\xi} \right ) + O \left [ \exp \left ( - \frac{r}{\chi} \right ) \right ],
    \end{align} where \( K \) is constant with respect to \( r \), and \( \xi > \chi \) is the average correlation length for \( \mathcal{C} (A : B) \).
\end{definition}

Remarkably, we find that a single average correlation length persists throughout the different families of measures of correlation and that it depends only on the underlying spatial dimension. We later pinpoint the origin of this behavior to the invariance of a spectral gap of a family of transfer matrices. \begin{subequations} \label{eq:lenghts}
    For MPS, \begin{align} \label{eq:mps_length}
        \xi =  \left [ \log \left ( \frac{d^2 D^2 - 1}{d D^2 - d} \right ) \right ]^{- 1} \equiv \xone,
    \end{align} and for isoTNS, \begin{align} \label{eq:iso_length}
        \xi = \left [ \log \left ( \frac{d^2 D^4 - 1}{d D^3 - d D} \right ) \right ]^{- 1} \equiv \xtwo.
    \end{align}
\end{subequations}

Note that the average correlation length for MPS coincides with that for isoTNS for \( d \to d D \). This seemingly small modification will prove to change the qualitative behavior substantially.

Before moving to the detailed presentation of our methods and results, we briefly comment on the considered measures of correlation and the implications of our findings, first for MPS and then for isoTNS.

\subsection{Results in one dimension (MPS)} \label{sec:mps_summary}

In one dimension, we compute the averages of the R\'{e}nyi-\( 2 \) mutual information \( I_2 (A : B) \), the \( 2 \)-norm expression \( N (A : B) \), and the von Neumann mutual information \( I (A : B) \) (see Sec.~\ref{sec:measures} for the definitions of the measures of correlation), with subsystems \( A \) and \( B \) as sketched in Fig~\ref{fig:setup}~(a).

We find that the averages decay exponentially as specified in Definition~\ref{def:definition} with the same correlation length \( \xone \) (see Results~\ref{res:mps_renyi}, \ref{res:mps_norm}, and \ref{res:mps_neumann}). The derivation for \( I (A : B) \) relies on a technical conjecture (see Conjecture~\ref{con:lambda_2}), which we will discuss in detail later. In addition, we show that the same conjecture is enough to assert that \( \xone \) is also the average correlation length for \( I_\alpha (A : B) \) for any integer value of \( \alpha \geq 1 \) (see Corollary~\ref{cor:mps_alpha}). In short, the same average correlation length \( \xone \) persists across different measures of correlation.

Interestingly, \( \xone \) depends only weakly on the bond dimension \( D \). In particular, for \( d, D \geq 2 \), \begin{align} \label{eq:mps_monotonicity_1}
    \xone = \left \{ \log \left [ \frac{d}{\zeta_{1 \mathrm{D}} (d, D)} \right ] \right \}^{- 1}
\end{align} with \begin{align} \label{eq:mps_monotonicity_2}
    \frac{4}{5} \leq \zeta_{1 \mathrm{D}} (d, D) < 1
\end{align} is monotonically increasing with \( D \).

Because we are concerned with random tensor network states, \( \xone \) is obtained after averaging over realizations. It is then natural to ask if exponentially decaying correlations are typical and, if so, what is the typical correlation length for an individual realization. This motivates the investigation of the concentration of the distribution around its average. To that end, we will show that it is exponentially unlikely in \( r \) that \( N (A : B) \) and \( I (A : B) \) decay slower than with \( \xone \) (see Corollaries~\ref{cor:mps_norm_concentration} and \ref{cor:mps_neumann_concentration}). Our result for \( N (A : B) \) allows us to deduce that the average of the trace distance \( T (A : B) \) decays at least exponentially with correlation length \( \xi \leq 2 \xone \), and it leads to a concentration result for \( T (A : B) \) (see Corollary~\ref{cor:mps_trace_concentration}).

As stated before, MPS are known to exhibit exponentially decaying correlations~\cite{PerezGarcia2007}. However, our finding that the average correlation length is almost independent of the bond dimension is novel and implies that long-range correlated states are very atypical members of the ensemble.

\subsection{Results in two dimensions (isoTNS)} \label{sec:iso_summary}

In two dimensions, we compute the averages of the R\'{e}nyi-\( 2 \) mutual information \( I_2 (A : B) \) and the \( 2 \)-norm expression \( N (A : B) \), where subsystems \( A \) and \( B \) are sketched in Fig~\ref{fig:setup}~(b).

As in one dimension, we find that the averages decay exponentially as specified in Definition~\ref{def:definition} with the same average correlation length \( \xtwo \) (see Results~\ref{res:iso_renyi} and \ref{res:iso_norm}).

The correlation length \( \xtwo \) displays a qualitatively different dependence on the bond dimension \( D \). In particular, for \( d, D \geq 2 \), \begin{align}
    \xtwo = \left \{ \log \left [ \frac{d}{\zeta_{2 \mathrm{D}} (d, D)} \right ] \right \}^{-1}
\end{align} with \begin{align}
    0 < \zeta_{2 \mathrm{D}} (d, D) \leq \frac{8}{21}
\end{align} is monotonically decreasing with \( D \), in contrast to its one-dimensional counterpart. As such, the largest correlation length is achieved for \( D = 2 \).

For \( N (A : B) \), we can extend the applicability of our results to any size and shape of subsystems \( A \) and \( B \); we find the decay to be at least exponential with correlation length \( \xi = \xtwo \) (see Corollary~\ref{cor:iso_norm}). We furthermore prove a concentration result for \( N (A : B) \) expressing that it is highly unlikely that \( N (A : B) \) decays slower than with \( \xtwo \) (see Corollary~\ref{cor:iso_norm_concentration}). This also allows us to draw a similar conclusion about the behavior of \( T (A : B) \) (see Corollary~\ref{cor:iso_trace_concentration}).

While it is expected for isoTNS that correlations decay exponentially if they are both on the orthogonality hypersurface~\cite{Zaletel2020}, our finding that the exponential decay persists for subsystems in the bulk is novel. Moreover, the different behavior of \( \xtwo \) as compared to \( \xone \) is particularly surprising; for typical isoTNS, increasing the bond dimension reduces the average correlation length.

\subsection{Further discussion} \label{sec:discussion}

Here, we discuss three aspects of our results on a more qualitative level: the absence of large average correlation lengths, the opposite dependence of said correlation lengths on the bond dimension in one and two dimensions, and the effect of blocking sites together.

From the analytical expressions for \( \xone \) and \( \xtwo \) [see Eq.~\eqref{eq:lenghts}], it follows that, for a fixed physical dimension \( d \), both average correlation lengths are upper bounded by a constant, irrespective of the bond dimension \( D \). In one dimension, the sequential-generation scheme corresponds to actions of unitary matrices on two neighboring \( d \)-dimensional sites together with a \( D \)-dimensional system. The unitary matrices act as an information carrier, conveying information from one physical qudit to another through interaction. Our results shows that, when these unitary matrices are random and independent from each other, the resulting correlations are, on average, weak.

The average correlation lengths for isoTNS can be obtained from the one for MPS by replacing \( d \to d D \). This simple fact, however, leads to \( \xtwo \) being a decreasing function of \( D \), in contrast to \( \xone \). To explain this feature, we use the fact (which we will prove later) that \( \xtwo \) is independent of the vertical extent \( h \) of subsystems \( A \) and \( B \). As a result, we can gain sufficient insight from the limiting case of a \( 1 \times n \) isoTNS, which corresponds to a certain MPS graphically given by \begin{align}
    \psi = \raisebox{-32pt}{\includegraphics{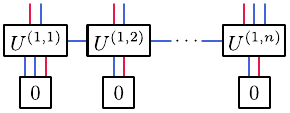}}.
\end{align} Importantly, this MPS has local dimension \( d D \), explaining why the average correlation lengths coincide for \( d \to d D \). An increase of the bond dimension of the isoTNS thus corresponds to an increase of the local dimension of the corresponding MPS, which dominates the behavior of \( \xone \).

Finally, we comment on the effect of blocking sites together. It holds that \begin{align}
    \xone = \frac{1}{\log (d)} + O \left ( \frac{1}{D^2} \right ),
\end{align} implying that \( \xone \approx 1 / \log (d) \) for \( D \gg 1 \). This suggests the following scale invariance property, valid in the context of this approximation. Consider two families \( \ket{\psi} \) and \( \ket{\widetilde \psi} \) of random MPS with \( \widetilde{d} = d^q \) and \( \widetilde{D} = D \). Graphically, for \( q = 2 \), \begin{align}
	\ket{\psi} = \raisebox{-32pt}{\includegraphics{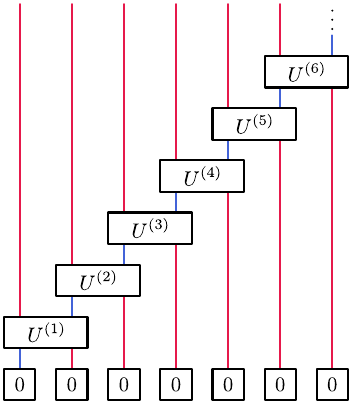}}
\end{align} and \begin{align}
    \ket{\widetilde{\psi}} = \raisebox{-32pt}{\includegraphics{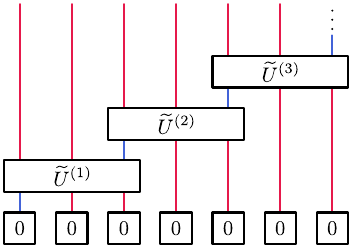}}.
\end{align} Our simple observation is that, although it holds that \( \widetilde{\xi} = \xi / q \), the effective distances (that is, the number of unitary matrices separating any two physical sites) are also rescaled as \( \widetilde{r} = r / q  \). Consequently, the circuit architectures generating \( \ket{\psi} \) and \( \ket{\widetilde \psi} \) have scale-invariant (average) correlations within the approximation \( \xone \approx 1 / \log (d) \).

\section{CORRELATIONS IN ONE DIMENSION} \label{sec:mps_results}

In this section, we state and discuss the results for random MPS summarized in Sec.~\ref{sec:mps_summary} in more detail. Before doing that, we develop the tools behind our proofs in Secs.~\ref{sec:mps_transfer} and \ref{sec:mps_correlations}. In Sec.~\ref{sec:mps_renyi}, we compute the average of \( I_2 (A : B) \), and in Sec.~\ref{sec:mps_norm}, we investigate the decays of \( N (A : B) \) and \( T (A : B) \). Finally, we discuss the behavior of \( I (A : B) \) in Sec.~\ref{sec:mps_neumann}.

When computing the averages of measures of correlation for random MPS, we will exploit a simplification with respect to the scenario depicted in Fig.~\ref{fig:setup}~(a). Instead of allowing for an arbitrary number of sites before subsystem \( A \), we prove our statements in the limit of \( c \to \infty \). As we show in Appendix~\ref{app:mps_neumann_c}, this does not constitute a limitation because the \( c \) initial sites do not affect the decay of correlations and, therefore, neither the average correlation length \( \xone \). Furthermore, we will see that the \( f \) sites after subsystem \( B \) do not play a role in the computation of average correlations, as it is expected for any sequentially generated state.

\subsection{Transfer matrices} \label{sec:mps_transfer}

The key challenge for computing the average of each measure of correlation will be evaluating multiple expressions of the form \begin{align} \label{eq:mps_caffe}
	\tr \left ( P \mathbb{E} \ketbra{\psi}{\psi}^{\otimes k} \right ),
\end{align} where \begin{align} \label{eq:mps_p}
    P & = \left ( P_e^{(d)} \right )^{\otimes c} \otimes \left ( P_\alpha^{(d)} \right )^{\otimes a} \otimes \left ( P_e^{(d)} \right )^{\otimes r} \nonumber \\
    & \hphantom{{} = {}} \otimes \left ( P_\beta^{(d)} \right )^{\otimes b} \otimes \left ( P_e^{(d)} \right )^{\otimes (f - 1)} \otimes P_e^{(d D)}.
\end{align} The permutation \( \alpha \in S_k \) acts on the sites comprising subsystem \( A \), while \( \beta \in S_k \) acts on the sites comprising \( B \). The exact forms of \( \alpha \) and \( \beta \) as well as the number of required replicas \( k \) depend on the considered measure of correlation and will be specified later. It shall also become clear why sites belonging to neither \( A \) nor \( B \) are acted upon by the trivial permutation \( e \in S_k \). In the following, we show that Eq.~\eqref{eq:mps_caffe}, for random MPS, reduces to multiplying matrices \( T_\rho \in \mathbb{R}^{k! \times k!} \) with \( \rho \in S_k \) whose definition will be natural. Because their role is analogous to the known transfer matrices mediating correlations, we will also adopt this term here.

Before introducing the transfer matrices, we must analyze \( \mathbb{E} \ketbra{\psi}{\psi}^{\otimes k} \). To that end, let us define \( V^{(j)} = U^{(j)} \otimes \overline{U^{(j)}} \). Then, by Eq.~\eqref{eq:mps_unitaries}, \begin{align}
	\ketbra{\psi}{\psi} = \raisebox{-32pt}{\includegraphics{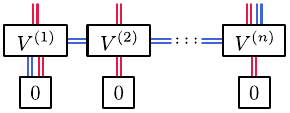}}
\end{align} and \begin{align}
	& \ketbra{\psi}{\psi}^{\otimes k} \nonumber \\
	& \quad = \raisebox{-39pt}{\includegraphics{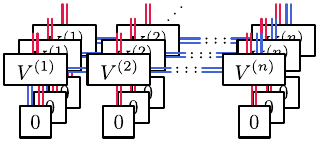}}.
\end{align}

By computing the \( k \)-fold twirl [see Eq.~\eqref{eq:twirl_graphical}], we obtain the building block \begin{align}
	\raisebox{-7pt}{\includegraphics{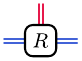}} & = \int \mathrm{d} U^{(j)} \, \raisebox{-39pt}{\includegraphics{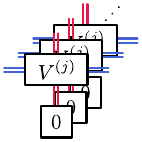}} \label{eq:mps_r} \\
	& = \raisebox{-7pt}{\includegraphics{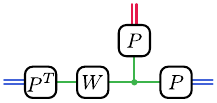}},
\end{align} where the (green) dot represents a Kronecker \( \delta \) on three permutation indices. Note that we have not drawn the contraction of a permutation matrix with \( \ketbra{0}{0}^{\otimes k} \) because it is trivial by Eq.~\eqref{eq:zero_contraction}. The average of a random MPS is then given by \begin{align}
	\mathbb{E} \ketbra{\psi}{\psi}^{\otimes k} = \raisebox{-7pt}{\includegraphics{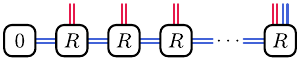}}.
\end{align}

We could, in principle, work with the building block above. However, it is not convenient to have dangling bond (blue) legs whose dimension grows with \( D \). By cutting permutation-valued (green) legs instead, we obtain a building block with fixed dimension for fixed \( k \). With that building block, the average of a random MPS is given by \begin{align}
	\mathbb{E} \ketbra{\psi}{\psi}^{\otimes k} = \raisebox{-7pt}{\includegraphics{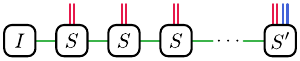}}.
\end{align}

The entries of the initial vector \( \bra{I_k} \in \mathbb{R}^{k!} \) are given by \begin{align}
	\raisebox{-7pt}{\includegraphics{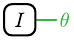}} = \raisebox{-7pt}{\includegraphics{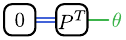}} = 1,
\end{align} where we have used Eq.~\eqref{eq:zero_contraction}. The tensors in the bulk are given via \begin{align}
	\raisebox{-7pt}{\includegraphics{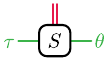}} & = \raisebox{-7pt}{\includegraphics{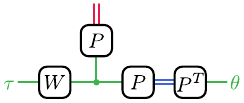}},
\end{align} and the final tensor is given via \begin{align}
	\raisebox{-7pt}{\includegraphics{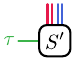}} & = \raisebox{-7pt}{\includegraphics{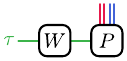}}.
\end{align}

Computing an expression of the form of Eq.~\eqref{eq:mps_caffe} corresponds to contracting each tensor \( S \) with \( P_\rho^{(d)} \), which leads us to the promised definition of a transfer matrix \( T_\rho \in \mathbb{R}^{k! \times k!} \). Using Eq.~\eqref{eq:permutation_contraction}, its entries are given by \begin{align}
	& \vcenter{\hbox{\includegraphics{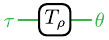}}} \nonumber \\
	& \qquad = \raisebox{-7pt}{\includegraphics{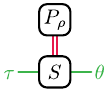}} = \raisebox{-7pt}{\includegraphics{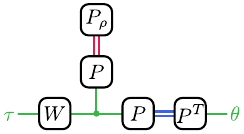}} \label{eq:gelato} \\
	& \qquad = \sum_{\sigma \in S_k} \Wg \left ( \sigma \tau^{- 1}, d D \right ) d^{\# \left ( \sigma \rho \right )} D^{\# \left ( \sigma \theta^{- 1} \right )}. \label{eq:mps_fame}
\end{align} We define \( T_\rho \) with respect to the basis defined by the map \( s_i \mapsto e_i \), where \( s_i \) is the \( i \)th element of \( S_k = \left \{ s_1, \dots, s_{k!} \right \} \), and \( \{ e_i \} \) is the standard basis of \( \mathbb{R}^{k!} \). In Appendix~\ref{app:conjecture}, we find that \( T_\rho \) is block triangular if the elements of \( S_k \) are ordered in a certain way.

As alluded to in Eq.~\eqref{eq:mps_p}, the final tensor \( S' \) will be contracted with the trivial permutation \( e \in S_k \) in our computations. The final vector \( \ket{F_k} \in \mathbb{R}^{k!} \) is thus defined via \begin{align}
	\raisebox{-7pt}{\includegraphics{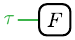}} & = \raisebox{-7pt}{\includegraphics{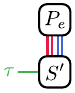}} = \raisebox{-7pt}{\includegraphics{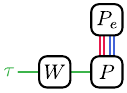}} \\
	& = \sum_{\sigma \in S_k} \Wg \left ( \sigma \tau^{- 1}, d D \right ) (d D)^{\# \left ( \sigma e \right )} = \delta_{e \tau} \, ,
\end{align} where we have used the definition of the Weingarten function.

Using the definitions of \( T_e \) and \( \ket{F_k} \), it is easy to confirm that \( T_e \ket{F_k} = \ket{F_k} \). Graphically, this implies the simplification \begin{align} \label{eq:mps_sequential_generation}
	\vcenter{\hbox{\includegraphics{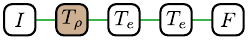}}} = \vcenter{\hbox{\includegraphics{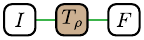}}}.
\end{align} From this, it also follows that \( \mathbb{E} \ketbra{\psi}{\psi}^{\otimes k} \) is properly normalized: \begin{align}
	\tr \left ( \mathbb{E} \ketbra{\psi}{\psi}^{\otimes k} \right ) = \fv{I_k}{T_e^{n - 1}}{F_k} = \braket{I_k}{F_k} = 1
\end{align}

With what we have laid out above, we can write Eq.~\eqref{eq:mps_caffe} in terms of transfer matrices: \begin{align} \label{eq:mps_biscottino}
    \tr \left ( P \mathbb{E} \ketbra{\psi}{\psi}^{\otimes k} \right ) = \fv{I_k}{T_e^c T_\alpha^a T_e^r T_\beta^b}{F_k}
\end{align}

We provide a simple Mathematica package~\cite{GitHub} that defines \( T_\rho \) with \( \rho \in S_k \) for \( k \in \{ 1, \dots, 20 \} \) according to Eq.~\eqref{eq:mps_fame}. The package relies on the package provided by the authors of Ref.~\cite{Fukuda2019} for evaluating the Weingarten function.

\subsection{Estimating the decay of correlations} \label{sec:mps_correlations}

The decay of the average of each measure of correlation is necessarily determined by the \( r \) sites separating subsystems \( A \) and \( B \). As we will see in the following sections, this will, for each measure, translate to a simple statement in terms of the just-defined transfer matrices. In particular, we will find that the decay of correlations is determined by \( T_e^r \) with \( e \in S_k \). Taking this as a fact for now, we connect the decay of correlations with the spectrum of \( T_e \).

The spectrum of \( T_e \) depends on \( k \) because \( k \) determines its dimension and entries. Still, we can make general statements about \( T_e \) for any \( k \geq 2 \). In particular, we will prove the following statements in Appendices~\ref{app:x_and_y} and \ref{app:lambda_1}.

\begin{restatable}{proposition}{nonnegative} \label{prop:nonnegative}
	The eigenvalues of \( T_e \) with \( e \in S_k \) are nonnegative for any \( k \geq 2 \).
\end{restatable}

\begin{restatable}{proposition}{diagonalizable} \label{prop:diagonalizable}
	\( T_e \) with \( e \in S_k \) is diagonalizable for any \( k \geq 2 \).
\end{restatable}

\begin{restatable}{proposition}{lambdaone} \label{prop:lambda_1}
	Let \( \lambda_1 > \lambda_2 > \cdots \geq 0 \) denote the distinct eigenvalues of \( T_e \) with \( e \in S_k \). Then, \( \lambda_1 = 1 \), and it is nondegenerate for any \( k \geq 2 \) if \( d \geq 2 \).
\end{restatable}

Given the statements above, we can expand \( T_e^r \) as \begin{align} \label{eq:expansion}
    T_e^r = \ketbra{R_1}{L_1} + \lambda_2^r \sum_{\mu = 1}^{w_2} \ketbra{R_2^{(\mu)}}{L_2^{(\mu)}} + O \left ( \lambda_3^r \right ),
\end{align} where \( \ket{R_i^{(\mu)}} \) denotes the \( \mu \)th right eigenvector corresponding to \( \lambda_i \), \( \bra{L_i^{(\mu)}} \) denotes the \( \mu \)th left eigenvector corresponding to \( \lambda_i \), and \( w_i \) denotes the degeneracy of \( \lambda_i \).

The asymptotic decay of correlations is thus determined by the subleading eigenvalue \( \lambda_2 \) of \( T_e \), and the average correlation length is given by \begin{align}
    \xi = - \frac{1}{\log \left ( \left | \lambda_2 \right | \right )}.
\end{align} The argument behind this is similar to one known from the analysis of correlations in translation-invariant MPS~\cite{Fannes1992, PerezGarcia2007}, where the decay is determined by the subleading eigenvalue of the relevant transfer matrix.

\subsection{R\'{e}nyi-\texorpdfstring{\( 2 \)}{2} mutual information} \label{sec:mps_renyi}

We start our analysis of correlations in random MPS with the simplest case, namely the computation of the average of the R\'{e}nyi-\( 2 \) mutual information \begin{align}
	I_2 (A : B) & = \log \left [ \tr \left ( \varrho_{A B}^2 \right ) \right ] \nonumber \\
	& \hphantom{{} = {}} - \log \left [ \tr \left ( \varrho_A^2 \right ) \right ] \nonumber \\
	& \hphantom{{} = {}} - \log \left [ \tr \left ( \varrho_B^2 \right ) \right ].
\end{align} The analytical treatment turns out to be comparatively simple if one assumes that \( \mathbb{E} \log (X) = \log (\mathbb{E} X) \), as is frequently done in this context~\cite{vonKeyserlingk2018, Zhou2019, Bertini2020}. The assumption can furthermore be justified with the fact that the purities of \( \varrho_A \), \( \varrho_B \), and \( \varrho_{A B} \) concentrate around their averages, as we discuss towards the end of this section. That said, we will not make the assumption further below when we study the von Neumann mutual information. The analysis there will require transfer matrices \( T_\rho \) with \( \rho \in S_k \) for all \( k \geq 2 \), while \( \rho \in S_2 \) will suffice here because only the averages of purities are needed. Our first result is summarized below.

\begin{result} \label{res:mps_renyi}
    The average of \( I_2 (A : B) \) with respect to the random MPS ensemble and subsystems \( A \) and \( B \) as sketched in Fig.~\ref{fig:setup}~(a) decays exponentially as specified in Definition~\ref{def:definition} with the average correlation length \( \xone \) defined in Eq.~\eqref{eq:mps_length}.
\end{result}

\begin{proof}
    We split the proof into four steps. The exact same structure will also appear in the proofs for the other measures of correlation. Thus, this proof serves as the simplest example and a point of reference for later proofs.
    
    \begin{step} \label{step:mps_renyi_1}
        We rewrite \( \mathbb{E} I_2 (A : B) \) in terms of expressions of the form of Eq.~\eqref{eq:mps_caffe}. To that end, we make the assumption that \( \mathbb{E} \log (X) = \log (\mathbb{E} X) \). Then, \begin{align}
            \mathbb{E} I_2 (A : B) & = \log \left [ \mathbb{E} \tr \left ( \varrho_{A B}^2 \right ) \right ] \nonumber \\
        	& \hphantom{{} = {}} - \log \left [ \mathbb{E} \tr \left ( \varrho_A^2 \right ) \right ] \nonumber \\
        	& \hphantom{{} = {}} - \log \left [ \mathbb{E} \tr \left ( \varrho_B^2 \right ) \right ].
        \end{align} \( \mathbb{E} \tr \left ( \varrho_A^2 \right ) \), \( \mathbb{E} \tr \left ( \varrho_B^2 \right ) \), and \( \mathbb{E} \tr \left ( \varrho_{A B}^2 \right ) \) can already be written in the desired form. For example, \begin{align} \label{eq:mps_renyi_2}
            \mathbb{E} \tr \left ( \varrho_{A B}^2 \right ) = \tr \left ( P \mathbb{E} \ketbra{\psi}{\psi}^{\otimes 2} \right )
        \end{align} with \begin{align} \label{eq:mps_renyi_3}
            P & = \left ( P_e^{(d)} \right )^{\otimes c} \otimes \left ( P_{(1 2)}^{(d)} \right )^{\otimes a} \otimes \left ( P_e^{(d)} \right )^{\otimes r} \nonumber \\
            & \hphantom{{} = {}} \otimes \left ( P_{(1 2)}^{(d)} \right )^{\otimes b} \otimes \left ( P_e^{(d)} \right )^{\otimes (f - 1)} \otimes P_e^{(d D)}.
        \end{align}
    \end{step}
    
    \begin{step}
        We express \( \mathbb{E} I_2 (A : B) \) in terms of the transfer matrices defined in Sec.~\ref{sec:mps_transfer}. Given the previous step, it is easy to confirm that \begin{align}
            \mathbb{E} I_2 (A : B) & = \log \left ( \fv{I_2}{T_e^c T_{(1 2)}^a T_e^r T_{(1 2)}^b}{F_2} \right ) \nonumber \\
		    & \hphantom{{} = {}} - \log \left ( \fv{I_2}{T_e^c T_{(1 2)}^a}{F_2} \right ) \nonumber \\
		    & \hphantom{{} = {}} - \log \left ( \fv{I_2}{T_e^{c + a + r} T_{(1 2)}^b}{F_2} \right ). \label{eq:mps_renyi_1}
        \end{align}
    \end{step}
    
    \begin{step}
        We expand \( \mathbb{E} I_2 (A : B) \) in terms of the spectrum of \( T_e \) with \( e \in S_2 \) [see Eq.~\eqref{eq:expansion}]. Using the relevant expressions for the Weingarten function, it is evident~\footnote{Note that the entries of \( T_e \) with \( e \in S_2 \) have previously appeared in Ref.~\cite{Haferkamp2021}.} that \begin{align}
    		T_e = \begin{pmatrix}
    			1 & \displaystyle \frac{d^2 D - D}{d^2 D^2 - 1} \\[1em]
    			0 & \displaystyle \frac{d D^2 - d}{d^2 D^2 - 1}
    		\end{pmatrix}
    	\end{align} is diagonalizable with \begin{align}
    		\lambda_1 = 1 \qquad \text{and} \qquad \lambda_2 = \frac{d D^2 - d}{d^2 D^2 - 1}.
    	\end{align} Expanding \( T_e^c \) and taking the limit of \( c \to \infty \) yields \begin{align}
    	    \mathbb{E} I_2 (A : B) & = \log \left ( \fv{L_1}{T_{(1 2)}^a T_e^r T_{(1 2)}^b}{F_2} \right ) \nonumber \\
		    & \hphantom{{} = {}} - \log \left ( \fv{L_1}{T_{(1 2)}^a}{F_2} \right ) \nonumber \\
		    & \hphantom{{} = {}} - \log \left ( \fv{L_1}{T_{(1 2)}^b}{F_2} \right ),
    	\end{align} where we have used that \( \braket{I_2}{R_1} = 1 \). After expanding also \( T_e^r \) and using that \( \ket{F_2} = \ket{R_1} \), we have \begin{align}
    	    \mathbb{E} I_2 (A : B) & = \log \Bigl ( \fv{L_1}{T_{(1 2)}^a}{R_1} \fv{L_1}{T_{(1 2)}^b}{R_1} \nonumber \\
    	    & \hphantom{{} = \log \Bigl (} {} + \lambda_2^r \fv{L_1}{T_{(1 2)}^a}{R_2} \fv{L_2}{T_{(1 2)}^b}{R_1} \Bigr ) \nonumber \\
		    & \hphantom{{} = {}} - \log \left ( \fv{L_1}{T_{(1 2)}^a}{R_1} \right ) \nonumber \\
		    & \hphantom{{} = {}} - \log \left ( \fv{L_1}{T_{(1 2)}^b}{R_1} \right ).
    	\end{align}
    \end{step}

    \begin{step}
        Finally, we can write \( \mathbb{E} I_2 (A : B) \) in the form of Definition~\ref{def:definition}. That is, \begin{align}
            & \mathbb{E} I_2 (A : B) \nonumber \\
    		& \quad = \log \left ( 1 + \lambda_2^r \frac{\fv{L_1}{T_{(1 2)}^a}{R_2} \fv{L_2}{T_{(1 2)}^b}{R_1}}{\fv{L_1}{T_{(1 2)}^a}{R_1} \fv{L_1}{T_{(1 2)}^b}{R_1}} \right ) \\
    		& \quad = \lambda_2^r \frac{\fv{L_1}{T_{(1 2)}^a}{R_2} \fv{L_2}{T_{(1 2)}^b}{R_1}}{\fv{L_1}{T_{(1 2)}^a}{R_1} \fv{L_1}{T_{(1 2)}^b}{R_1}} + O \left ( \lambda_2^{2 r} \right ) \\
    		& \quad \equiv K \exp \left ( - \frac{r}{\xi} \right ) + O \left [ \exp \left ( - \frac{2 r}{\xi} \right ) \right ],
        \end{align} where \begin{align}
    		K = \frac{\fv{L_1}{T_{(1 2)}^a}{R_2} \fv{L_2}{T_{(1 2)}^b}{R_1}}{\fv{L_1}{T_{(1 2)}^a}{R_1} \fv{L_1}{T_{(1 2)}^b}{R_1}}
    	\end{align} and \begin{align}
    		\xi = - \frac{1}{\log (\lambda_2)} = - \left [ \log \left ( \frac{d D^2 - d}{d^2 D^2 - 1} \right ) \right ]^{- 1} = \xone.
	    \end{align}
    \end{step}
    
    \noindent This concludes the proof.
\end{proof}

We are now in a position to discuss the concentration of the purities of \( \varrho_A \), \( \varrho_B \), and \( \varrho_{A B} \) around their averages, which justifies \( \mathbb{E} \log (X) = \log (\mathbb{E}  X) \). By expanding \( \mathbb{E} \tr \left ( \varrho_A^2 \right ) = \fv{L_1}{T_{(1 2)}^a}{R_1} \), one can confirm that \begin{align*}
    \mathbb{E} \tr \left ( \varrho_A^2 \right ) = \left ( \frac{d + 1}{d} \right )^2 \frac{1}{D^2} + O \left ( \frac{1}{d^a} + \frac{1}{D^4} \right ).
\end{align*} Similarly, it holds that \begin{align*}
    \mathbb{E} \tr \left ( \varrho_{A B}^2 \right ) = \left ( \frac{d + 1}{d} \right )^4 \frac{1}{D^4} + O \left ( \frac{1}{d^{a + b}} + \frac{1}{D^6} \right ).
\end{align*} The purities of \( \varrho_A \), \( \varrho_B \), and \( \varrho_{A B} \) are thus close to their minimum values, implying concentration by Markov's inequality and justifying \( \mathbb{E} \log (X) = \log (\mathbb{E}  X) \).

As discussed earlier, the R\'{e}nyi-\( 2 \) mutual information is lacking many of the desirable properties that a sound measure of correlation ought to fulfill. On top, our computation simplifies considerably because we are using the assumption that \( \mathbb{E} \log (X) = \log (\mathbb{E} X) \), which amounts to ignoring statistical fluctuations in the different realizations. In the following section, we will see that \( N (A : B) \) decays exponentially with the same average correlation length \( \xone \). We will furthermore show that \( N (A : B) \) concentrates around its average, providing evidence that fluctuations can be safely ignored in our context.

\subsection{Trace distance and \texorpdfstring{\( 2 \)}{2}-norm} \label{sec:mps_norm}

In this section, we investigate average correlations as quantified by the trace distance \( T (A : B) \). As anticipated in Sec.~\ref{sec:measures}, this is a challenging task. However, as laid out there, the \( 2 \)-norm expression \( N (A : B) \) reliably estimates \( T (A : B) \) for the case of random MPS. Hence, we compute the average of \begin{align}
	N (A : B) & = \left \| \varrho_{A B} - \varrho_A \otimes \varrho_B \right \|_2^2 \\
	& = \tr \left ( \varrho_{A B}^2 \right ) \nonumber \\
	& \hphantom{{} = {}} + \tr \left ( \varrho_A^2 \right ) \tr \left ( \varrho_B^2 \right ) \nonumber \\
	& \hphantom{{} = {}} - 2 \tr \left [ \varrho_{A B} \left (\varrho_A \otimes \varrho_B \right ) \right ]. \label{eq:mps_norm_2}
\end{align} 

Because of its connection to the Hilbert-Schmidt inner product, the average of \( N (A : B) \) can be computed without any simplifying assumptions. Making use of the transfer-matrix techniques introduced above, we prove the following result.

\begin{restatable}{result}{mpsnorm} \label{res:mps_norm}
	The average of \( N (A : B) \) with respect to the random MPS ensemble and subsystems \( A \) and \( B \) as sketched in Fig.~\ref{fig:setup}~(a) decays exponentially as specified in Definition~\ref{def:definition} with the average correlation length \( \xone \) defined in Eq.~\eqref{eq:mps_length}.
\end{restatable}

\begin{proof}[Sketch of proof.]
    The proof follows the same procedure as that of Result~\ref{res:mps_renyi}. Here, we sketch the main steps and refer to Appendix~\ref{app:mps_norm} for more details.

    In Step~\ref{step:mps_norm_1}, we write \( \mathbb{E} N (A : B) \) in terms of expressions of the form of Eq.~\eqref{eq:mps_caffe}. The second summand in Eq.~\eqref{eq:mps_norm_2} requires permutations of \( S_4 \) because \begin{align}
        \mathbb{E} \tr \left ( \varrho_A^2 \right ) \tr \left ( \varrho_B^2 \right ) = \tr \left ( P \mathbb{E} \ketbra{\psi}{\psi}^{\otimes 4} \right )
    \end{align} with \begin{align}
        P & = \left ( P_e^{(d)} \right )^{\otimes c} \otimes \left ( P_{(1 2)}^{(d)} \right )^{\otimes a} \otimes \left ( P_e^{(d)} \right )^{\otimes r} \nonumber \\
        & \hphantom{{} = {}} \otimes \left ( P_{(3 4)}^{(d)} \right )^{\otimes b} \otimes \left ( P_e^{(d)} \right )^{\otimes (f - 1)} \otimes P_e^{(d D)}.
    \end{align} The first summand and the third summand, respectively, require only permutations of \( S_2 \) and \( S_3 \). In Step~\ref{step:mps_norm_2},  we thus write \( \mathbb{E} N (A : B) \) in terms of transfer matrices \( T_\rho \) with \( \rho \in S_4 \).
    
    This means that the average correlation length is determined by the subleading eigenvalue of \( T_e \) with \( e \in S_4 \). Let \( \lambda_1 > \lambda_2 > \cdots \geq 0 \) denote the distinct eigenvalues of \( T_e \). In Step~\ref{step:mps_norm_3}, we find that \begin{align}
    	\lambda_1 = 1 \qquad \text{and} \qquad \lambda_2 = \frac{d D^2 - d}{d^2 D^2 - 1},
    \end{align} just like for \( T_e \) with \( e \in S_2 \). The former is nondegenerate, while the degeneracy of the latter is given by the number of transposition in \( S_4 \), \begin{align}
    	w_2 = \binom{4}{2} = 6.
    \end{align} Thus, the average correlation length for \( N (A : B) \) coincides with that for \( I_2 (A : B) \), as we conclude in Step~\ref{step:mps_norm_4}.
\end{proof}

The above result establishes the exponential decay of the average. However, one is usually interested in knowing if typical instances are expected to have the same exponential decay. This can be easily established by Markov's inequality because \( N (A : B) \) is nonnegative and its average decays to zero as a function of the distance \( r \).

\begin{corollary} \label{cor:mps_norm_concentration}
	For subsystems \( A \) and \( B \) as sketched in Fig.~\ref{fig:setup}~(a), sufficiently large \( r \), and all \( 0 < \varepsilon < 1 \), the random MPS ensemble satisfies \begin{align}
		\Pr \left \{ N (A : B) \geq K \exp \left [ - \frac{(1 - \varepsilon) r}{\xone} \right ] \right \} \leq \exp \left ( - \frac{\varepsilon r}{\xone} \right ),
	\end{align} where \( K \) is constant with respect to \( r \).
\end{corollary}

\begin{proof}
    By Result~\ref{res:mps_norm}, for sufficiently large \( r \), we can bound \( \mathbb{E} N (A : B) \leq K \exp (- r / \xone) \). Because \( N (A : B) \) is nonnegative, by Markov's inequality, we have, for \( \eta > 0 \),
    \begin{align}
        & \Pr \left [ N (A : B) \geq \eta K \exp \left ( - \frac{r}{\xone} \right ) \right ] \nonumber \\
        & \qquad \leq \Pr \bigl [ N (A : B) \geq \eta \mathbb{E} N (A : B) \bigr ] \\
        & \qquad \leq \frac{1}{\eta}.
    \end{align} The result follows with \( \eta = \exp (\varepsilon r / \xone) \).
\end{proof}

The above corollary reflects that it is exponentially unlikely in \( r \) that \( N (A : B) \) decays slower than with the average correlation length \( \xone \). Because we have already established that the average case exhibits an exponential decay with correlation length \( \xone \), the average case is also typical.

By combining the above result with Eq.~\eqref{eq:mps_norm_1}, we can now also bound the correlation length for \( T (A : B) \).

\begin{corollary} \label{cor:mps_trace_concentration}
	For subsystems \( A \) and \( B \) as sketched in Fig.~\ref{fig:setup}~(a), sufficiently large \( r \), and all \( 0 < \varepsilon < 1 \), the random MPS ensemble satisfies \begin{align}
		\Pr \left \{ T (A : B) \geq K \exp \left [ - \frac{(1 - \varepsilon) r}{2 \xone} \right ] \right \} \leq \exp \left ( - \frac{\varepsilon r}{\xone} \right ),
	\end{align} where \( K \) is constant with respect to \( r \).
\end{corollary}

\begin{proof}
     It holds that \begin{align}
         \mathbb{E} T (A : B) & \leq \sqrt{\mathbb{E} [T (A : B)]^2} \\
         & \leq \frac{D^2}{2} \sqrt{\mathbb{E} N (A : B)} \\
         & \leq K \exp \left( - \frac{r}{2 \xone} \right ),
     \end{align} where, in the last line, we have assumed \( r \) to be sufficiently large. The result follows as in the proof of Corollary~\ref{cor:mps_norm_concentration}.
\end{proof}

Thus, with overwhelming probability, correlations as quantified by \( T (A : B) \) decay exponentially with \( \xi \leq 2 \xone \).

\subsection{Von Neumann mutual information} \label{sec:mps_neumann}

\begin{figure}
	\centering
	\includegraphics{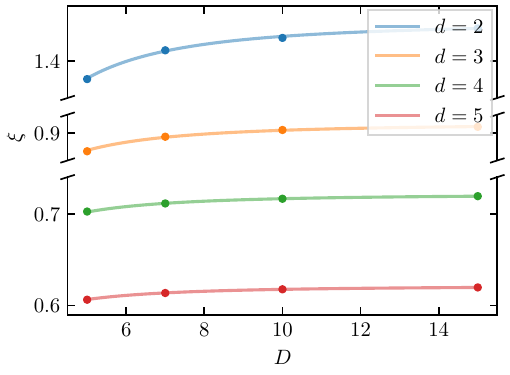}
	\caption{Numerically obtained average correlation length \( \xi \) for different \( d \) and \( D \). The data points are obtained by fitting the average value of \( I (A : B) \) against \( r \in \{ 5, 7, 9, 11, 13, 15 \} \) for \( a = b = 1 \). The sample size of \( 10\,000 \) suffices for the error bars to lie within the plot points. The opaque curves correspond to \( \xone \) [see Eq.~\eqref{eq:mps_length}]}
	\label{fig:numerical_neumann}
\end{figure}

The fact that \( I_2 (A : B) \) and \( N (A : B) \) have the same average correlation length \( \xone \) motivates the question whether other measures of correlation behave similarly. In this section, we will provide compelling evidence that \( \xone \) is indeed the average correlation length also for the von Neumann mutual information \( I (A : B) \).

We start by numerically investigating the behavior of \( I (A : B) \) for random MPS. We have generated MPS according to our measure, computed the average of \( I (A : B) \), and extracted the average correlation length from fits. As Fig.~\ref{fig:numerical_neumann} shows, the numerically obtained average correlation length coincides well with \( \xone \). It should be noted that we have set \( c = 0 \) for our numerical analysis. We discuss in Appendix~\ref{app:mps_neumann_c} why this does not affect the average correlation length. For more details on our numerical analysis, see Appendix~\ref{app:numerical_neumann}.

We now turn to the analytical computation of the average of \begin{align}
	I (A : B) & = \tr \left [ \rho_{A B} \log \left ( \rho_{A B} \right ) \right ] \nonumber \\
	& \hphantom{{} = {}} - \tr \left [ \rho_A \log \left ( \rho_A \right ) \right ] \nonumber \\
	& \hphantom{{} = {}} - \tr \left [ \rho_B \log \left ( \rho_B \right ) \right ].
\end{align}

To be able to make use of the transfer-matrix techniques introduced above, we employ two replica tricks to write \( \mathbb{E} I (A : B) \) in terms of expressions of the form of Eq.~\eqref{eq:mps_caffe}. First, we write \( S (\rho) \) as the limit of \( \alpha \to 1 \) of \( S_\alpha (\rho) \): \begin{align} \label{eq:mps_neumann_1}
	S (\rho) = \lim_{\alpha \to 1} \frac{1}{1 - \alpha} \log [\tr (\varrho^\alpha)]
\end{align} Second, instead of assuming again that \( \mathbb{E} \log (X) = \log (\mathbb{E} X) \), we use \begin{align} \label{eq:mps_neumann_2}
	\mathbb{E} \log (X) = \lim_{v \to 0} \frac{1}{v} \log \left ( \mathbb{E} X^v \right ).
\end{align} We are thus dealing with expressions of the form \( \tr \left ( P \mathbb{E} \ketbra{\psi}{\psi}^{\otimes v \alpha} \right ) \), which require transfer matrices \( T_\rho \) with \( \rho \in S_{v \alpha} \) (see Appendix~\ref{app:mps_neumann}). This means that knowing the spectrum of \( T_e \) with \( e \in S_{v \alpha} \) for all \( v \alpha \geq 2 \) allows us to draw conclusions about the decay of the average of \( I (A : B) \).

While this is in principle a daunting task, we are able to prove several properties of the transfer matrix \( T_e \) with \( e \in S_k \) for any \( k \geq 2 \) (see Propositions~\ref{prop:nonnegative}, \ref{prop:diagonalizable}, and \ref{prop:lambda_1}). 

In Appendix~\ref{app:conjecture}, we furthermore show that \( T_e \) with \( e \in S_k \) has eigenvalue \begin{align} \label{eq:mu_2}
    \mu_2 = \frac{d D^2 - d}{d^2 D^2 - 1}.
\end{align} for any \( k \geq 2 \). Its degeneracy is at least \begin{align}
	v_2 = \binom{k}{2},
\end{align} the number of transpositions in \( S_k \). We conjecture that \( \mu_2 \) is the subleading eigenvalue of \( T_e \) with \( e \in S_k \) for any \( k \geq 2 \) and that it has degeneracy \( v_2 \). We know this conjecture to hold for \( k \in \{ 2, 3, 4 \} \), and we have numerical evidence suggesting so for \( k \in \{ 5, 6, 7 \} \)~\cite{GitHub}. We were not able to prove the statement outright, but in the following we argue that it is the only missing step to show that \( \xone \) is the average correlation length for \( I (A : B) \). Let us state the conjecture formally below.

\begin{conjecture} \label{con:lambda_2}
    Let \( \lambda_1 > \lambda_2 > \cdots \geq 0 \) denote the distinct eigenvalues of \( T_e \) with \( e \in S_k \). Then, \( \lambda_2 = \mu_2 \) with degeneracy \( w_2 = v_2 \) for any \( k \geq 2 \).
\end{conjecture}

If Conjecture~\ref{con:lambda_2} holds, the properties of \( T_e \) with \( e \in S_{v a} \) that are relevant for determining the decay of correlations are independent of \( v \alpha \). In Appendix~\ref{app:mps_neumann}, we argue that the replica limit does not affect this and prove the following result.

\begin{restatable}{result}{mpsneumann} \label{res:mps_neumann}
	If Conjecture~\ref{con:lambda_2} holds, the average of \( I (A : B) \) with respect to the random MPS ensemble and subsystems \( A \) and \( B \) as sketched in Fig.~\ref{fig:setup}~(a) decays exponentially as specified in Definition~\ref{def:definition} with the average correlation length \( \xone \) defined in Eq.~\eqref{eq:mps_length}.
\end{restatable}

We provide a concentration result also for \( I (A : B) \). The statement and its proof are identical to that for \( N (A : B) \). It is exponentially unlikely in \( r \) that \( I (A : B) \) decays slower than with the average correlation length \( \xone \), and the average case is also typical.

\begin{corollary} \label{cor:mps_neumann_concentration}
	If Conjecture~\ref{con:lambda_2} holds, for subsystems \( A \) and \( B \) as sketched in Fig.~\ref{fig:setup}~(a), sufficiently large \( r \), and all \( 0 < \varepsilon < 1 \), the random MPS ensemble satisfies \begin{align}
		\Pr \left \{ I (A : B) \geq K \exp \left [ - \frac{(1 - \varepsilon) r}{\xi} \right ] \right \} \leq \exp \left ( - \frac{\varepsilon r}{\xi} \right ),
	\end{align} where \( K \) is constant with respect to \( r \).
\end{corollary}

Another corollary of Result~\ref{res:mps_neumann} is that \( \xone \) is also the average correlation length for \( I_\alpha (A : B) \) for any integer value of \( \alpha \geq 1 \). We prove also this statement in Appendix~\ref{app:mps_neumann}.

\begin{restatable}{corollary}{mpsalpha} \label{cor:mps_alpha}
	If Conjecture~\ref{con:lambda_2} holds, for any integer value of \( \alpha \geq 1 \), the average of \( I_\alpha (A : B) \) with respect to the random MPS ensemble and subsystems \( A \) and \( B \) as sketched in Fig.~\ref{fig:setup}~(a) decays exponentially as specified in Definition~\ref{def:definition} with the average correlation length \( \xone \) defined in Eq.~\eqref{eq:mps_length}.
\end{restatable}

Finally, let us summarize the reason behind the persistent appearance of the average correlation length \( \xone \). In all of the examined cases, the asymptotic behavior of the correlations was determined by the asymptotic decay of \( T_e^r \). Although the transfer matrix \( T_e \) with \( e \in S_k \) does depend on the number of replicas \( k \), its asymptotic decay does not (given Conjecture~\ref{con:lambda_2}), resulting in a common average correlation length \( \xone \) across measures of correlation with different complexity.

\section{CORRELATIONS IN TWO DIMENSIONS} \label{sec:iso_results}

\begin{figure}
	\centering
	\includegraphics{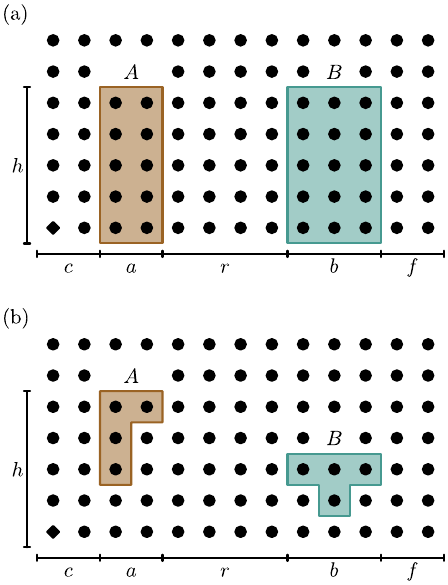}
	\caption{We investigate average correlations in random isoTNS between two subsystems \( A \) and \( B \) as a function of their horizontal distance \( r \). \( A \) and \( B \) respectively stretch across \( a \) and \( b \) consecutive horizontal sites. In addition to indicating the origin of the sequential generation, the diamond also indicates the orthogonality center of the isoTNS. (a) For Results~\ref{res:iso_renyi} and \ref{res:iso_norm}, we consider \( A \) and \( B \) that touch the orthogonality hypersurface and stretch across \( h \) consecutive vertical sites. (b) We will provide an additional result for the \( 2 \)-norm expression \( N (A : B) \) for arbitrary (but fixed) \( A \) and \( B \) that do not need to touch the orthogonality hypersurface.}
	\label{fig:iso_setup}
\end{figure}

In this section, we state and discuss in more detail the results for random isoTNS summarized in Sec.~\ref{sec:iso_summary}. In Sec.~\ref{sec:iso_transfer}, we develop the two-dimensional analogue to the transfer matrices introduced in Sec.~\ref{sec:mps_transfer}, the tool behind our proofs. In Sec.~\ref{sec:iso_renyi}, we compute the average of \( I_2 (A : B) \), and in Sec.~\ref{sec:iso_norm}, we investigate the decays of \( N (A : B) \) and \( T (A : B) \).

We show the exponential decay of the average for each considered measure of correlation and subsystems \( A \) and \( B \) as sketched in Fig.~\ref{fig:iso_setup}~(a) as a function of the horizontal distance \( r \) between \( A \) and \( B \). In particular, we prove that \( I_2 (A : B) \) and \( N (A : B) \) have a common average correlation length that is independent of the sizes of \( A \) and \( B \). Moreover, thanks to the more amenable properties of \( N (A : B) \), we are able to show that average correlations decay exponentially also for subsystems \( A \) and \( B \) that do not need to touch the orthogonality hypersurface [see Fig.~\ref{fig:iso_setup}~(b)]. As in one dimension, we prove our statements in the limit of \( c \to \infty \).

\subsection{Transfer matrices} \label{sec:iso_transfer}

As in one dimension, computing the average of each measure of correlation will involve computing multiple terms of the form \begin{align} \label{eq:iso_caffe}
	\tr \left ( P \mathbb{E} \ketbra{\psi}{\psi}^{\otimes k} \right ),
\end{align} where \( P \) has a similar tensor product structure as Eq.~\eqref{eq:mps_p}, adapted to the two-dimensional setting considered here. The number of required replicas \( k \) again depends on the considered measure of correlation.

We will thus need a two-dimensional analogue to the transfer matrices introduced in Sec.~\ref{sec:mps_transfer}. In contrast to the one-dimensional case, the size of the resulting transfer matrices will also depend on the geometry of the considered subsystems, making their analysis much more challenging. However, the procedure of defining the tensors is similar to that in one dimension. We provide an overview here and refer to Appendix~\ref{app:iso_transfer} for more details.

We define \( V^{(i, j)} = U^{(i, j)} \otimes \overline{U^{(i, j)}} \), where \( U^{(i, j)} \in U \left ( d D^2 \right ) \) is the unitary matrix depicted in Fig.~\ref{fig:sequential_generation}~(b2). By computing the \( k \)-fold twirl [see Eq.~\eqref{eq:twirl_graphical}], we obtain the building block \begin{align}
	\raisebox{-26pt}{\includegraphics{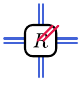}} & = \int \mathrm{d} U^{(i, j)} \, \raisebox{-26pt}{\includegraphics{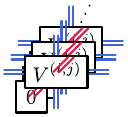}} \\
	& = \raisebox{-37pt}{\includegraphics{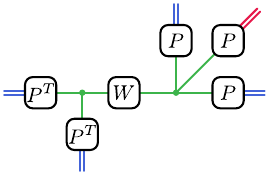}}.
\end{align}

As in one dimension, it is convenient to define a building block for which the contraction of bond (blue) legs is implicit. This results in a tensor with only permutation-valued (green) legs. As we show in Appendix~\ref{app:iso_transfer}, the resulting bulk tensor is given via \begin{align} \label{eq:iso_s}
	\raisebox{-24pt}{\includegraphics{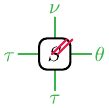}} & = \raisebox{-24pt}{\includegraphics{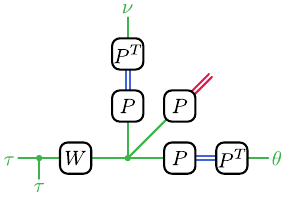}}.
\end{align} For the sake of brevity, the expressions for the boundary tensors are stated in Appendix~\ref{app:iso_transfer}.

Computing expressions of the form of Eq.~\eqref{eq:iso_caffe} corresponds to contracting each tensor \( S \) with \( P_\rho^{(d)} \). The entries of the resulting bulk tensor \( T_\rho \in \mathbb{R}^{k! \times k! \times k! \times k!} \) are given by \begin{align}
	& \vcenter{\hbox{\includegraphics{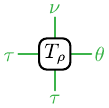}}} \nonumber \\
	& \qquad = \raisebox{-24pt}{\includegraphics{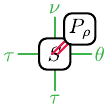}} = \raisebox{-24pt}{\includegraphics{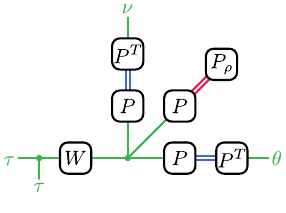}} \\
	& \qquad = \sum_{\sigma \in S_k} \Wg \left ( \sigma \tau^{- 1}, d D^2 \right ) d^{\# \left ( \sigma \rho \right )} D^{\# \left ( \sigma \theta^{- 1} \right )} D^{\# \left ( \sigma \nu^{- 1} \right )}. \label{eq:iso_fame}
\end{align}

In our computations, the site in the top-right corner belongs to neither \( A \) nor \( B \). It is thus acted upon by the trivial permutation \( e \in S_k \). As we show in Appendix~\ref{app:iso_transfer}, the corresponding tensor is given by \( T_e = \ketbra{F_k}{F_k} \).

In Appendix~\ref{app:iso_transfer}, we furthermore show that a property similar to Eq.~\eqref{eq:mps_sequential_generation} also holds for random isoTNS. That is, tensors corresponding to the trivial permutation \( e \in S_k \) on the boundary simplify. For example, \begin{align} \label{eq:iso_sequential_generation}
	\raisebox{-32pt}{\includegraphics{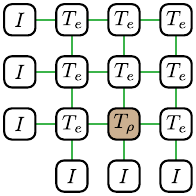}} = \raisebox{-32pt}{\includegraphics{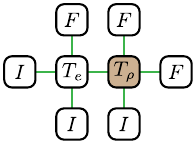}}.
\end{align} As for the one-dimensional case, the proper normalization of \( \mathbb{E} \ketbra{\psi}{\psi}^{\otimes k} \) follows directly from this property.

Instead of thinking in terms of contractions of two-dimensional tensor networks, it will later prove beneficial to think again in terms of multiplications of matrices. To that end, for any height \( h \) of subsystems \( A \) and \( B \) [see. Fig.~\ref{fig:iso_setup}~(a)], we define \begin{align} \label{eq:iso_to_mps}
    \vcenter{\hbox{\includegraphics{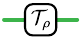}}} = \vcenter{\hbox{\includegraphics{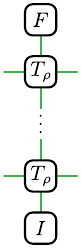}}}
\end{align} as well as \begin{align} \label{eq:iso_vectors}
    \vcenter{\hbox{\includegraphics{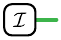}}} = \vcenter{\hbox{\includegraphics{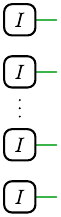}}} \qquad \text{and} \qquad \vcenter{\hbox{\includegraphics{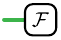}}} = \vcenter{\hbox{\includegraphics{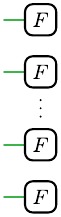}}}.
\end{align}

For subsystems \( A \) and \( B \) as sketched in Fig.~\ref{fig:iso_setup}~(a), we can now write Eq.~\eqref{eq:iso_caffe} in terms of transfer matrices: \begin{align} \label{eq:iso_biscottino}
    \tr \left ( P \mathbb{E} \ketbra{\psi}{\psi}^{\otimes k} \right ) = \fv{\mathcal{I}_k}{\mathcal{T}_e^c \mathcal{T}_\alpha^a \mathcal{T}_e^r \mathcal{T}_\beta^b}{\mathcal{F}_k}
\end{align}

We provide an additional Mathematica package~\cite{GitHub} that defines \( T_\rho \) with \( \rho \in S_k \) for \( k \in \{ 1, \dots, 20 \} \) according to Eq.~\eqref{eq:iso_fame}. Once again, the package relies on the one provided by the authors of Ref.~\cite{Fukuda2019} for evaluating the Weingarten function.

\subsection{Estimating the decay of correlations} \label{sec:iso_correlations}

Equation~\eqref{eq:iso_biscottino} implies that, for subsystems \( A \) and \( B \) as sketched in Fig.~\ref{fig:iso_setup}~(a), the decay of the average of each measure of correlation will again reduce to a statement in terms of transfer matrices. In particular, the decay will be determined by the spectrum of \( \mathcal{T}_e \) with \( e \in S_k \). Notice that, in addition to a dependence on \( k \), the form and properties of \( \mathcal{T}_e \) now depend also on the height \( h \) of subsystems \( A \) and \( B \) [see. Fig.~\ref{fig:iso_setup}~(a)]. However, as we prove in the following sections, its two leading eigenvalues are independent of \( h \) for at least \( k = 2 \) and \( k = 4 \). Crucially, this will allow us to make statements about the decay of the averages of \( I_2 (A : B) \) and \( N (A : B) \) for arbitrary \( h \). 

Note that we could, in principle, investigate vertical separation instead of horizontal separation because \begin{align}
    \vcenter{\hbox{\includegraphics{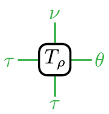}}} = \vcenter{\hbox{\includegraphics{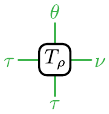}}}.
\end{align} The underlying exchange of indices does not affect the spectrum of the relevant identity transfer matrix and thus neither the average correlation length. This reflects the fact that the sequential generation procedure is symmetric in the horizontal and vertical spatial directions.

\subsection{R\'{e}nyi-\texorpdfstring{\( 2 \)}{2} mutual information} \label{sec:iso_renyi}

In this section, we compute the average of the R\'{e}nyi-\( 2 \) mutual information \( I_2 (A : B) \) for random isoTNS that are generated as sketched in Fig.~\ref{fig:sequential_generation}~(b1). Subsystems \( A \) and \( B \) are defined in Fig.~\ref{fig:iso_setup}~(a).

As in one dimension, we will make the assumption that \( \mathbb{E} \log (X) = \log (\mathbb{E} X) \) (see Sec.~\ref{sec:mps_renyi}). Step~\ref{step:iso_renyi_1} of the proof of Result~\ref{res:iso_renyi} (see Appendix~\ref{app:iso_renyi}) is thus almost identical to Step~\ref{step:mps_renyi_1} of the proof of Result~\ref{res:mps_renyi}.

Also Step~\ref{step:iso_renyi_2} is largely analogous. To see this, let us take \( \mathbb{E} \tr \left (\varrho_A^2 \right ) \) as an example. With \( A \) and \( B \) as defined in Fig.~\ref{fig:iso_setup}~(a) and \( x = (1 2) \), \begin{align}
	\mathbb{E} \tr \left (\varrho_A^2 \right ) & = \vcenter{\hbox{\includegraphics{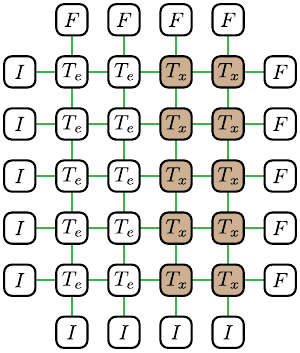}}} \label{eq:iso_renyi_1} \\
	& = \vcenter{\hbox{\includegraphics{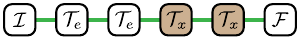}}} \label{eq:iso_renyi_2} \\
	& = \fv{\mathcal{I}_2}{\mathcal{T}_e^c \mathcal{T}_x^a}{\mathcal{F}_2}.
\end{align} The resulting expression for \( \mathbb{E} I_2 (A : B) \), \begin{align}
	\mathbb{E} I_2 (A : B) & = \log \left ( \fv{\mathcal{I}_2}{\mathcal{T}_e^c \mathcal{T}_x^a \mathcal{T}_e^r \mathcal{T}_x^b}{\mathcal{F}_2} \right ) \nonumber \\
	& \hphantom{{} = {}} - \log \bigl ( \fv{\mathcal{I}_2}{\mathcal{T}_e^c \mathcal{T}_x^a}{\mathcal{F}_2} \bigr ) \nonumber \\
	& \hphantom{{} = {}} - \log \left ( \fv{\mathcal{I}_2}{\mathcal{T}_e^{c + a + r} \mathcal{T}_x^b}{\mathcal{F}_2} \right ),
\end{align} resembles Eq.~\eqref{eq:mps_renyi_1} closely.

The remaining technical challenge is the analysis of the spectrum of \( \mathcal{T}_e \) with \( e \in S_2 \). Let \( \lambda_1 > \lambda_2 > \cdots \geq 0 \) denote its distinct eigenvalues. We show in Appendix~\ref{app:iso_spectrum_2} that \begin{align}
	\lambda_1 = 1 \qquad \text{and} \qquad \lambda_2 = \frac{d D^3 - d D}{d^2 D^4 - 1}
\end{align} and that both eigenvalues are nondegenerate. In the proof, that holds for any \( h \), we map the contraction of tensors defining \( \mathcal{T}_e \) with \( e \in S_2 \) to a multiplication of matrices. Using the Weingarten calculus, we show that \( \mathcal{T}_e \) is upper block triangular, a property that simplifies the analysis of its spectrum. The main difficulty is then to prove that the specified \( \lambda_2  \) is indeed the subleading eigenvalue for all \( h \). We do this by exploiting substochasticity.

Following the same reasoning as before, we find that the average of \( I_2 (A : B) \) decays exponentially as specified in Definition~\ref{def:definition}. We prove this result in Appendix~\ref{app:iso_renyi}.

\begin{restatable}{result}{isorenyi} \label{res:iso_renyi}
	The average of \( I_2 (A : B) \) with respect to the random isoTNS ensemble and subsystems \( A \) and \( B \) as sketched in Fig.~\ref{fig:iso_setup}~(a) decays exponentially as specified in Definition~\ref{def:definition} with the average correlation length \( \xtwo \) defined in Eq.~\eqref{eq:iso_length}.
\end{restatable}

\subsection{Trace distance and \texorpdfstring{\( 2 \)}{2}-norm} \label{sec:iso_norm}

In this section, we show the exponential decay of the average of \( N (A : B) \) for random isoTNS. As for the one-dimensional case, we do that to eventually make conclusions about the behavior of the trace distance \( T (A : B) \).

While it is not trivial to compute the average of \( N (A : B) \) with respect to the random isoTNS ensemble and subsystems \( A \) and \( B \) as sketched in Fig.~\ref{fig:iso_setup}~(a), the computation follows along the lines of what we have laid out in Sec.~\ref{sec:mps_norm}. In particular, we find that the decay of the average of \( N (A : B) \) is determined by the spectrum of the transfer matrix \( \mathcal{T}_e \) with \( e \in S_4 \).

Let \( \lambda_1 > \lambda_2 > \cdots \geq 0 \) denote the distinct eigenvalues of \( \mathcal{T}_e \) with \( e \in S_4 \). As we show in Appendix~\ref{app:iso_spectrum_4}, for any \( h \), \begin{align}
    \lambda_1 = 1 \qquad \text{and} \qquad \lambda_2 = \frac{d D^3 - d D}{d^2 D^4 - 1}.
\end{align} As in one dimension, the former is nondegenerate, while the degeneracy of the latter is given by the number of transpositions in \( S_4 \), \begin{align}
    w_2 = \binom{4}{2} = 6.
\end{align} While the analysis of the spectrum is, in principle, similar to that of the spectrum of \( \mathcal{T}_e \) with \( e \in S_2 \), it is considerably more technical. This is largely due to the fact that the matrices whose multiplication defines \( \mathcal{T}_e \) with \( e \in S_4 \) are significantly more complex. Still, we find also \( \mathcal{T}_e \) with \( e \in S_4 \) to be upper block triangular, allowing us to show that the specified \( \lambda_2 \) is indeed the subleading eigenvalue.

Given \( \lambda_2 \) of \( \mathcal{T}_e \) with \( e \in S_4 \) and the arguments developed in Sec.~\ref{sec:mps_norm}, we can state the first result of this section, which we prove in Appendix~\ref{app:iso_norm}.

\begin{restatable}{result}{isonorm} \label{res:iso_norm}
	The average of \( N (A : B) \) with respect to the random isoTNS ensemble and subsystems \( A \) and \( B \) as sketched in Fig.~\ref{fig:iso_setup}~(a) decays exponentially as specified in Definition~\ref{def:definition} with the average correlation length \( \xtwo \) defined in Eq.~\eqref{eq:iso_length}.
\end{restatable}

We now turn to the case of correlations in isoTNS with arbitrary (but fixed) subsystems \( A \) and \( B \) [see Fig.~\ref{fig:iso_setup}~(b)]. The decay of the average of \( N (A : B) \) for arbitrary \( A \) and \( B \) can be bounded by employing the fact that the Schatten \( 2 \)-norm satisfies~\cite{Rastegin2012} \begin{align}
    \left \| \tr_B \left ( X_{A B} \right ) \right \|_2 \leq \sqrt{\dim (B)} \left \| X_{A B} \right \|_2,
\end{align} where \( X_{A B} \) is any bipartite operator and \( \dim (B) \) is the dimension of the Hilbert space that is traced out. This means that \begin{align}
    N (A : B) \leq d^{\left | A_C \right | + \left | B_C \right |} N (A' : B'),
\end{align} where \( A \) is now an arbitrary subsystem, \( A' \) is its (minimal) enclosing rectangle that touches the hypersurface, and \( A_C = A' - A \). \( B' \) and \( B_C \) are defined analogously for \( B \).

Using the statement above, we can bound the decay of the average of \( N (A : B) \) for arbitrary \( A \) and \( B \) as a straightforward corollary of Result~\ref{res:iso_norm}. We stress that we consider regime in which the distance \( r \) between \( A \) and \( B \) grows.

\begin{corollary} \label{cor:iso_norm}
    For arbitrary subsystems \( A \) and \( B \), the average of \( N (A : B) \) with respect to the random isoTNS ensemble decays as
    \begin{align}
        N (A : B) =  O \left [ \exp \left ( - \frac{r}{\xtwo} \right ) \right ],
    \end{align} where the average correlation length \( \xtwo \) is defined in Eq.~\eqref{eq:iso_length}.
\end{corollary}

Finally, we state a concentration result for \( N (A : B) \), which, in combination with Eq.~\eqref{eq:iso_norm_1}, also allows us to draw conclusions about the typical behavior of \( T (A : B) \). As before, we consider arbitrary (but fixed) subsystems \( A \) and \( B \) [see Fig.~\ref{fig:iso_setup}~(b)].

\begin{corollary} \label{cor:iso_norm_concentration}
	For arbitrary subsystems \( A \) and \( B \), sufficiently large \( r \), and all \( 0 < \varepsilon < 1 \), the random isoTNS ensemble satisfies \begin{align}
		\Pr \left \{ N (A : B) \geq K \exp \left [ - \frac{(1 - \varepsilon) r}{\xtwo} \right ] \right \} \leq \exp \left ( - \frac{\varepsilon r}{\xtwo} \right ),
	\end{align} where \( K \) is constant with respect to \( r \) and the average correlation length \( \xtwo \) is defined in Eq.~\eqref{eq:iso_length}.
\end{corollary}

\begin{corollary} \label{cor:iso_trace_concentration}
    For arbitrary subsystems \( A \) and \( B \), sufficiently large \( r \), and all \( 0 < \varepsilon < 1 \), the random isoTNS ensemble satisfies \begin{align}
		\Pr \left \{ T (A : B) \geq K \exp \left [ - \frac{(1 - \varepsilon) r}{2 \xtwo} \right ] \right \} \leq \exp \left ( - \frac{\varepsilon r}{\xtwo} \right ),
	\end{align} where \( K \) is constant with respect to \( r \) and the average correlation length \( \xtwo \) is defined in Eq.~\eqref{eq:iso_length}.
\end{corollary}

The proofs and discussions of these results are identical to their one-dimensional counterparts.

The tools we have developed in this and the previous sections should, in principle, allow us to make statements also about the decay of the average of the von Neumann mutual information \( I (A : B) \) with respect to the random isoTNS ensemble. As in one dimension, we would need to investigate the spectrum of \( \mathcal{T}_e \) with \( e \in S_k \) for all \( k \geq 2 \). However, as the analysis of the spectrum of \( \mathcal{T}_e \) is already quite technical for \( e \in S_4 \) with our methods, we refrain from tackling the spectrum for \( k \geq 5 \) here.

\section{CONCLUSION} \label{sec:conclusion}

We have investigated the average behavior of correlations between two distant subsystems \( A \) and \( B \) for ensembles of random MPS and isoTNS. As measures of correlation, we have considered the R\'{e}nyi-\( \alpha \) mutual information, a measure arising from the Hilbert-Schmidt norm, the trace distance, and the von Neumann mutual information. We have shown that the average of each considered measure exhibits an exponential decay. Our results can equivalently be seen as describing states resulting from quantum circuits with a sequential architecture and Haar random gates.

By leveraging the Weingarten calculus, we have developed a mathematical framework that allows to infer the average correlation length from the subleading eigenvalue of an appropriately defined transfer matrix. We have computed the averages of the R\'{e}nyi-\( \alpha \) mutual information and the measure arising from the Hilbert-Schmidt norm to show the emergence of an average correlation length that depends only on the underlying spatial dimension but not the considered measure. In particular, the average correlation length for random MPS increases weakly with the bond dimension \( D \) and converges rapidly (as \( D \) grows) to \( 1 / \log (d) \), where \( d \) is the physical dimension. In contrast, the average correlation length for random isoTNS decreases with \( D \). The highest average correlation length for random isoTNS is achieved with the lowest nontrivial bond dimension (\( D = 2 \)).

Using elementary concentration results, we have furthermore deduced the typical behavior of the measure arising from the Hilbert-Schmidt norm, which has in turn allowed us to make similar statements about the trace distance. For MPS, we have been able to give strong indications that the universal correlation length applies also to the von Neumann mutual information, and also any R\'{e}nyi-\( \alpha \) mutual information for integer values of \( \alpha \geq 1 \). It would be interesting to prove this behavior rigorously and also investigate its validity for isoTNS. Another possible future direction would be to study average correlations in more general random PEPS as well as other types of quantum circuit architectures. To extend our results from the class of isoTNS to generic PEPS, a protocol to sequentially generate the latter would be needed.

It could also be interesting to investigate the implications of our results on state-preparation schemes and variational algorithms for MPS and particularly isoTNS. Regarding the former, it would be interesting to see if typical sequentially generated states, being short-range correlated, can be prepared by shallower circuits. On the numerical side, we expect our findings to be relevant for the performance of variational algorithms that approximate long-range correlated states with MPS and isoTNS. Our results imply that (even more for isoTNS than for MPS) polynomially decaying states are very special examples of the two variational families.

\section*{ACKNOWLEDGMENTS}

We thank Ignacio Cirac and Philippe Faist for fruitful discussions. F.B. and G.S. acknowledge financial support from the Alexander von Humboldt Foundation.

\bibliography{references}

\appendix

\onecolumngrid

\section{\texorpdfstring{\( k \)}{k}-FOLD TWIRL} \label{app:twirl}

In this Appendix, we present some additional details on the \( k \)-fold twirl. In particular, we go from Eq.~\eqref{eq:twirl_before} to Eq.~\eqref{eq:twirl_after}. To do this, we will need a result of Ref.~\cite{Collins2003}, which appears as Corollary~2.4 in Ref.~\cite{Collins2006}. We state it without proof as Lemma~\ref{lem:collins_2}.

\begin{lemma} \label{lem:collins_2}
	Let \( k \) be a positive integer, and \( \left ( i_1, \dots, i_k \right ) \), \( \left ( j_1, \dots, j_k \right ) \), \( \left ( \ell_1, \dots, \ell_k \right ) \), and \( \left ( m_1, \dots, m_k \right ) \) be \( k \)-tuples of positive integers. Then, \begin{align}
		\int \mathrm{d} U \, U_{i_1 j_i} \cdots U_{i_k j_k} \overline{U_{m_1 \ell_1}} \cdots \overline{U_{m_k \ell_k}} = \sum_{\sigma, \tau \in S_k} \Wg \left ( \sigma^{- 1} \tau, q \right ) \delta_{i_1 m_{\sigma (1)}} \cdots \delta_{i_k m_{\sigma (k)}} \delta_{j_1 \ell_{\tau (1)}} \cdots \delta_{j_k l_{\tau (k)}},
	\end{align} where the integration is with respect to the Haar measure on the unitary group \( U (q) \).
\end{lemma}

By Lemma~\ref{lem:collins_2}, \begin{align}
	\left ( \mathcal{T}_U^{(k)} (X) \right )_{i_1 \cdots i_k m_1 \cdots m_k} & = \sum_{\substack{j_1, \dots, j_k \\ \ell_1, \dots, \ell_k}} \int \mathrm{d} U \, U_{i_1 j_i} \cdots U_{i_k j_k} X_{j_1 \cdots j_k \ell_1 \cdots \ell_k} \overline{U_{m_1 \ell_1}} \cdots \overline{U_{m_k \ell_k}} \\
	& = \sum_{\substack{j_1, \dots, j_k \\ \ell_1, \dots, \ell_k}} \sum_{\sigma, \tau \in S_k} \Wg \left ( \sigma^{- 1} \tau, q \right ) \delta_{i_1 m_{\sigma (1)}} \cdots \delta_{i_k m_{\sigma (k)}} X_{j_1 \cdots j_k \ell_1 \cdots \ell_k} \delta_{j_1 \ell_{\tau (1)}} \cdots \delta_{j_k l_{\tau (k)}} \\
	& = \sum_{j_1, \dots, j_k} \sum_{\sigma, \tau \in S_k} \Wg \left ( \sigma^{- 1} \tau, q \right ) \delta_{i_1 m_{\sigma (1)}} \cdots \delta_{i_k m_{\sigma (k)}} X_{j_1 \cdots j_k j_{\tau^{- 1} (1)} \cdots j_{\tau^{- 1} (k)}} \\
	& = \sum_{\sigma, \tau \in S_k} \Wg \left ( \sigma^{- 1} \tau, q \right ) \left ( P_{\sigma^{- 1}}^{(q)} \right )_{i_1 \cdots i_k m_1 \cdots m_k} \tr \left ( X P_\tau^{(q)} \right ),
\end{align} where, in the final line, we have used that \begin{align}
	\left ( P_{\sigma^{- 1}}^{(q)} \right )_{i_1 \cdots i_k m_1 \cdots m_k} & = \fv{i_1 \cdots i_k}{P_{\sigma^{- 1}}^{(q)}}{m_1 \cdots m_k} \\
	& = \braket{i_1 \cdots i_k}{m_{\sigma (1)} \cdots m_{\sigma (k)}} \\
	& = \delta_{i_1 m_{\sigma (1)}} \cdots \delta_{i_k m_{\sigma (k)}}
\end{align} and that \begin{align}
	\tr \left ( X P_\tau^{(q)} \right ) & = \sum_{j_1, \dots, j_k} \fv{j_1 \cdots j_k}{X P_\tau^{(q)}}{j_1 \cdots j_k} \\
	& = \sum_{j_1, \dots, j_k} \fv{j_1 \cdots j_k}{X}{j_{\tau^{- 1} (1)} \cdots j_{\tau^{- 1} (k)}} \\
	& = X_{j_1 \cdots j_k j_{\tau^{- 1} (1)} \cdots j_{\tau^{- 1} (k)}}.
\end{align}

Thus, \begin{align}
	\mathcal{T}_U^{(k)} (X) & = \sum_{\sigma, \tau \in S_k} \Wg \left ( \sigma^{- 1} \tau, q \right ) P_{\sigma^{- 1}}^{(q)} \tr \left ( X P_\tau^{(q)} \right ) \\
	& = \sum_{\sigma, \tau \in S_k} \Wg \left ( \sigma \tau^{- 1}, q \right ) P_\sigma^{(q)} \tr \left ( X P_{\tau^{- 1}}^{(q)} \right ) \\
	& = \sum_{\sigma, \tau \in S_k} \Wg \left ( \sigma \tau^{- 1}, q \right ) P_\sigma^{(q)} \tr \left [ X \left ( P_\tau^{(q)} \right )^T \right ],
\end{align} which coincides with Eq.~\eqref{eq:twirl_after}.

In addition, let us confirm that \( \mathcal{T}_U^{(k)} \left ( P_\rho^{(q)} \right ) = P_\rho^{(q)} \)~\cite{Roberts2017, Brandao2021}. Indeed, \begin{align}
	\mathcal{T}_U^{(k)} \left ( P_\rho^{(q)} \right ) & = \sum_{\sigma, \tau \in S_k} \Wg \left ( \sigma \tau^{- 1}, q \right ) P_\sigma^{(q)} \tr \left [ P_\rho^{(q)} \left ( P_\tau^{(q)} \right )^T \right ] \\
	& = \sum_{\sigma, \tau \in S_k} \Wg \left ( \sigma \tau^{- 1}, q \right ) P_\sigma^{(q)} q^{\# \left (\rho \tau^{- 1} \right )} \\
	& = \sum_{\sigma \in S_k} \delta_{\rho \sigma} P_\sigma^{(q)} \\
	& = P_\rho^{(q)},
\end{align} where, in the third line, we have used the definition of the Weingarten function.

\section{PROOFS OF PROPOSITIONS~\ref{prop:nonnegative} AND \ref{prop:diagonalizable}} \label{app:x_and_y}

\nonnegative*

\diagonalizable*

To prove Propositions~\ref{prop:nonnegative} and \ref{prop:diagonalizable}, we will define matrices \( X \in \mathbb{R}^{k! \times k!} \) and \( Y \in \mathbb{R}^{k! \times k!} \) so that \( C_k = W X Y \). We will then discuss some properties of those three matrices. The proofs themselves will boil down to similarity.

We define the diagonal matrix \( X \in \mathbb{R}^{k! \times k!} \) via \begin{align}
    X_{\sigma \sigma} = d^{\# (\sigma)}.
\end{align} It is evident that \( X \) is positive definite. 
 
We define the Gram matrix \( Y \in \mathbb{R}^{k! \times k!} \) via \begin{align}
    Y_{\sigma \theta} = \tr \left [ P_\sigma^{(D)} \left ( P_\theta^{(D)} \right )^T \right ] = D^{\# \left ( \sigma \theta^{- 1} \right )}.
\end{align} \( Y \) is positive semidefinite because it is a Gram matrix.

The Weingarten matrix \( \Wg \left ( \sigma \tau^{- 1}, q \right ) = \left ( G^{- 1} \right )_{\sigma \tau} \) is positive definite because the Gram matrix \( G \) [see Eq.~\eqref{eq:gram}] is positive definite~\footnote{The Gram matrix \( G \) is positive definite only if \( k \leq d D \) because the set \( \left \{ P_\sigma^{(d D)} \right \} \) is linearly independent only if \( k \leq d D \)~\cite{Collins2021}. As noted in Footnote~\cite{Note1}, the Weingarten matrix \( W = G^{- 1} \) thus exists only if \( k \leq d D \).}. 

We will need the fact that \( Y W \) is positive semidefinite. \( Y W \) has nonnegative eigenvalues because it is a product of a positive semidefinite (\( Y \)) and a positive definite matrix (\( W \)) (see Corollary~7.6.2 of Ref.~\cite{Horn2012}). Furthermore, \( Y W \) is symmetric: \begin{align}
	\fv{\tau}{Y W}{\theta} & = \sum_{\sigma \in S_k} \Wg \left ( \sigma \tau^{- 1}, d D \right ) D^{\# \left ( \sigma \theta^{- 1} \right )} \\
	& = \sum_{\pi \in S_k} \Wg \left ( \pi \theta \tau^{- 1}, d D \right ) D^{\# (\pi)} \\
	& = \sum_{\pi \in S_k} \Wg \left ( \tau \theta^{- 1} \pi^{- 1}, d D \right ) D^{\# (\pi)} \\
	& = \sum_{\pi \in S_k} \Wg \left ( \theta^{- 1} \pi^{- 1} \tau, d D \right ) D^{\# (\pi)} \\
	& = \sum_{\varphi \in S_k} \Wg \left ( \theta^{- 1} \varphi, d D \right ) D^{\# \left ( \tau \varphi^{- 1} \right )} \\
	& = \sum_{\varphi \in S_k} \Wg \left ( \varphi \theta^{- 1}, d D \right ) D^{\# \left ( \varphi \tau^{- 1} \right )} \\
	& = \fv{\theta}{Y W}{\tau}
\end{align} In the third line, we have used that \( \Wg (\alpha, q) = \Wg \left ( \alpha^{- 1}, q \right ) \), and, in the fourth line, we have used that \( \Wg (\alpha, q) = \Wg \left ( \beta \alpha \beta^{- 1}, q \right ) \). Both identities are a result of the Weingarten function being sensitive only to the conjugacy class of a given permutation.

\begin{proof}[Proof of Proposition~\ref{prop:nonnegative}]
    \( C_k = W X Y \) is similar to \( X Y W \). A product of a positive definite (\( X \)) and a positive semidefinite matrix (\( Y W \)), \( X Y W \) has nonnegative eigenvalues. The statement follows by similarity.
\end{proof}

\begin{proof}[Proof of Proposition~\ref{prop:diagonalizable}]
    \( C_k = W X Y \) is similar to \( X Y W \), which is similar to \( X^{1 / 2} Y W X^{1 / 2} \). Because \( Y W \) is symmetric, so is \( X^{1 / 2} Y W X^{1 / 2} \), which makes the latter diagonalizable. The statement follows.
\end{proof}

\section{PROOF OF PROPOSITION~\ref{prop:lambda_1}} \label{app:lambda_1}

\lambdaone*

\begin{proof}
    For what follows, it is convenient to define the transfer matrix \( \Sigma_e \) with \( e \in S_k \) via \begin{align} \label{eq:keks}
    	\raisebox{-7pt}{\includegraphics{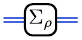}} = \raisebox{-7pt}{\includegraphics{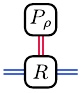}} = \raisebox{-7pt}{\includegraphics{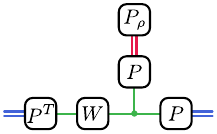}}.
    \end{align}
    
    Our strategy will be to prove the claimed spectral property for \( \Sigma_{e} \). This is enough because, for any two operators \( X \) and \( Y \) for which \( X Y \) and \( Y X \) are well defined, the sets of eigenvalues of \( X Y \) and \( Y X \) are the same (up to zeros and the multiplicity of the nonzero eigenvalues). By Eqs.~\eqref{eq:gelato} and \eqref{eq:keks}, \( T_\rho \) and \( \Sigma_\rho \) are related in exactly this way.
    
    As \( \Sigma_e \) arises from the contraction of the quantum channel underlying \( R \) [see Eq.~\eqref{eq:mps_r}] with the identity permutation, it can be understood as a generalization of the \( k \)-fold twirling operator. In particular, it involves an ''environment'' \( E \) of dimension \(\left ( \mathbb C^d \right )^{\otimes k} \) that is eventually traced out. As a superoperator (that is, without using the operator-vector correspondence), it reads [see Eq.~\eqref{eq:twirl_before}] \begin{align} \label{eq:kuchen}
    	\Sigma_e (\cdot) = \int \mathrm{d} U \, \tr_E \left \{ U^{\otimes k} \left [ \bigotimes_{l = 1}^k  \ketbra{0}{0}_{E_l} \otimes (\cdot)_{S_l} \right ] \left ( U^\dagger \right )^{\otimes k} \right \},
    \end{align} where the integration is with respect to the Haar measure on the unitary group \( U (d D) \), \( E = \bigotimes_{l = 1}^k E_l \) corresponds to \( \left ( \mathbb{C}^d \right )^{\otimes k} \), and \( S = \bigotimes_{l = 1}^k S_l \) corresponds to \( \left ( \mathbb{C}^D \right )^{\otimes k} \). It is apparent that Eq.~\eqref{eq:kuchen} represents a convex combination of quantum channels (note the Stinespring dilation form), and thus the resulting operator is also a valid quantum channel. This implies that \( 1 \) is an eigenvalue of \( \Sigma_e \) and that there is no eigenvalue of greater modulus~\cite{Wolf2012}.
    
    We now prove that no other eigenvalue of the same modulus exists. To that end, we will show that \( \Sigma_e \) is a primitive channel~\cite{Wolf2012}, which implies said property. A quantum channel is primitive if and only if the spanning space formed by products of its Kraus operators, \begin{align} \label{eq:kraus_span} 
    	\mathcal{K}_m = \spn \left ( \left \{ \prod_{k = 1}^m K_{i_k} \right \} \right ),
    \end{align} is equal to the full matrix algebra for some integer \( m \), that is, \( \mathcal{K}_m = \mathcal M_{D^k} (\mathbb{C}) \)~\cite{Sanz2010}. We show that this condition is satisfied for \( \Sigma_e \) if \( d \geq 2 \).
    
    Indeed, the Haar integral in Eq.~\eqref{eq:kuchen}, together with the partial trace, can be understood as a (redundant) Kraus decomposition. Precisely, we can take \begin{align} \label{eq:kraus_trace}
        \left \{ K_i \right \}_i = \left \{ \tr_E \left [ \left ( \bigotimes_{l = 1}^k \ketbra{+}{\psi_l}_{E_l} \otimes I_{S_l} \right ) U^{\otimes k} \right ] \right \}_{U, \psi},
    \end{align} where \( U \in U (d D) \), \( \ket{\psi_l} \in \{ \ket{0}, \dots, \ket{d - 1} \} \) is the computational basis of the \( l \)th replica of the environment, and \( \ket{+} = \sum_{j = 0}^{d - 1} \ket{j} / \sqrt{d} \). The latter is a choice (instead of \( \ket{0} \)) made for later convenience. It remains to show that there exists an integer \( m = m (k, d, D) \) such that \( \mathcal{K}_m = \mathcal{M}_{D^k} (\mathbb{C}) \) if \( d \geq 2 \). First of all, note that the above fails for \( d = 1 \) (that is, if the environment \( E \) is trivial). This is because \( \spn \left ( \left \{ U^{\otimes k} \right \}_U \right ) \) coincides with the symmetric subspace over the \( k \) subsystems~\cite{Watrous2018}. However, \( d = 2 \) is already enough to span the full matrix algebra.
    
    An explicit construction to show this fact amounts to taking \( U \) to be a controlled unitary gate, where the control system is \( E \). In particular, consider \begin{align}
        \ket{\psi_l} = \ket{\delta_{l r}} \qquad \text{and} \qquad U = \ketbra{0}{0}_E \otimes I_S + \sum_{j = 1}^{d - 1} \ketbra{j}{j}_{E} \otimes V_S,
    \end{align} where \( r \in \{ 1, \dots, k \} \), and \( V_S \) is unitary. This results in Kraus operators [see Eq.~\eqref{eq:kraus_trace}] of the form \begin{align}
        I_{S_1} \otimes \cdots \otimes I_{S_{r - 1}} \otimes V \otimes I_{S_{r + 1}} \cdots \otimes I_{S_k}.
    \end{align} Taking (finite) products, as Eq.~\eqref{eq:kraus_span} dictates, is enough to build a basis of the vector space \( \mathcal{M}_{D^k} (\mathbb{C}) \). Note that this construction requires two control levels (that is, \( d \geq 2 \)).
    
    This concludes the proof.
\end{proof}

\section{FURTHER PROPERTIES OF THE \texorpdfstring{TRANSFER MATRIX \( T_e \)}{IDENTITY TRANSFER MATRIX}} \label{app:conjecture}

In this Appendix, we state and prove statements about the structure of the transfer matrix \( T_e \) with \( e \in S_k \) that will motivate Conjecture~\ref{con:lambda_2}. We prove the statements in Appendices~\ref{app:conjugacy_class}, \ref{app:block_triangular}, and \ref{app:basis_change}.

Moving forward, we denote \( T_e \) with \( e \in S_k \) by \( C_k \). Each entry \begin{align} \label{eq:stucafe}
    \fv{\tau}{C_k}{\theta} = \sum_{\sigma \in S_k} \Wg \left ( \sigma \tau^{- 1}, d D \right ) d^{\# (\sigma)} D^{\# \left ( \sigma \theta^{- 1} \right )}
\end{align} of \( C_k \in \mathbb{R}^{k! \times k!} \) is a sum of \( k! \) terms. While this may sound daunting, \( C_k \) exhibits a structure that reduces its complexity.

As formalized by Proposition~\ref{prop:conjugacy_class}, the entries \( \fv{\tau}{C_k}{\theta} \) of \( C_k \) exhibit a dependence on the conjugacy class of their indices. In particular, if we know the entries of the column given by \( \theta \in S_k \), we also know the entries of the columns given by permutations in the same conjugacy class as \( \theta \).

\begin{restatable}{proposition}{conjugacyclass} \label{prop:conjugacy_class}
	For any \( \pi \in S_k \), \begin{align}
		\fv{\tau}{C_k}{\theta} = \fv{\pi \tau \pi^{- 1}}{C_k}{\pi \theta \pi^{- 1}}.
	\end{align}
\end{restatable}

Certain entries of \( C_k \) vanish, while others are given by the entries of \( C_{k - 1} \). Proposition~\ref{prop:block_triangular} captures these two statements.

\begin{restatable}{proposition}{blocktriangular} \label{prop:block_triangular}
	For all \( \theta \in S_k \) with \( \theta (k) = k \), \begin{align}
		\fv{\tau}{C_k}{\theta} = \delta_{k, \tau (k)} \fv{\tau^\downarrow}{C_{k - 1}}{\theta^\downarrow},
	\end{align} where \( \rho^\downarrow \in S_{k - 1} \) is the restriction of \( \rho \in S_k \) with \( \rho (k) = k \) to the permutation on \( \{ 1, \dots, k - 1 \} \).
\end{restatable}

To understand the strength of Propositions~\ref{prop:conjugacy_class} and \ref{prop:block_triangular}, let us have a look at \( C_2 \) and \( C_3 \), the two most simple transfer matrices. With bases \( S_2 = \{ e, (1 2) \} \) and \( S_3 = \{ e, (1 2), (1 3), (2 3), (1 2 3), (1 3 2) \} \), one finds that \begin{align} \label{eq:conjecture_1}
	C_2 = \begin{pmatrix}
		1 & \alpha \\ 0 & \beta
	\end{pmatrix} \qquad \text{and} \qquad C_3 = \begin{pmatrix}
		1 & \alpha & \alpha & \alpha & \gamma & \gamma \\ 0 & \beta & 0 & 0 & \delta & \delta \\ 0 & 0 & \beta & 0 & \delta & \delta \\ 0 & 0 & 0 & \beta & \delta & \delta \\ 0 & 0 & 0 & 0 & \varepsilon & \zeta \\ 0 & 0 & 0 & 0 & \zeta & \varepsilon
	\end{pmatrix},
\end{align} where each Greek letter corresponds to some function of \( d \) and \( D \) whose exact form is not relevant here. The entries in the first four columns of \( C_3 \) are fully determined by those of \( C_2 \). The entries in the last two columns of \( C_3 \) do not arise from those of \( C_2 \), but the sixth column is a permutation of the fifth.

Note that we are deliberately choosing a basis \( S_k = \left \{ s_1, \dots, s_{k!} \right \} \) that makes the special structure of \( C_k \) more apparent. In particular, we sort permutations so that those with \( i \) fixed points come before those with \( i - 1 \) fixed points. We group permutations that have common fixed points and then those that are in the the same conjugacy class. Given this basis, Propositions~\ref{prop:conjugacy_class} and \ref{prop:block_triangular} imply that \( C_k \) is block triangular with \( k \) diagonal blocks.

We denote by \( C_k^{(1)} \) the diagonal block with \( \tau = \theta = e \in S_k \) and by \( C_k^{(i)} \) with \( 2 \leq i \leq k \) the diagonal block corresponding to \( \tau, \theta \in S_k \) with \( k - i \) fixed points. The spectrum of \( C_k \) is then given by \begin{align}
	\lambda \left ( C_k \right ) = \lambda \left ( C_k^{(1)} \right ) \cup \lambda \left ( C_k^{(2)} \right ) \cup \cdots \cup \lambda \left ( C_k^{(k)} \right ).
\end{align} It is apparent that \( C_k^{(1)} \) has a single, nondegenerate eigenvalue \begin{align}
	\mu_1 = \fv{e}{C_k}{e} = 1.
\end{align} \( C_k^{(2)} \) has a single, degenerate eigenvalue \begin{align}
	\mu_2 = \fv{(1 2)}{C_k}{(1 2)} = \frac{d D^2 - d}{d^2 D^2 - 1},
\end{align} which corresponds to the expression of \( \beta \) in Eq.~\eqref{eq:conjecture_1}. The degeneracy of \( \mu_2 \) is given by the size of the block, which is in turn given by the number of transpositions in \( S_k \), \begin{align}
	v_2 = \binom{k}{2}.
\end{align}

By Proposition~\ref{prop:lambda_1}, \( 1 \) is the leading eigenvalue of \( C_k \) and nondegenerate. We conjecture that \( \mu_2 \) is the subleading eigenvalue of \( C_k \) and that it has degeneracy \( v_2 \). The statement of Conjecture~\ref{con:lambda_2} holds for \( k \in \{ 2, 3, 4 \} \), and we have numerical evidence suggesting so for \( k \in \{ 5, 6, 7 \} \)~\cite{GitHub}.

\begin{repconjecture}{con:lambda_2}
    Let \( \lambda_1 > \lambda_2 > \cdots \geq 0 \) denote the distinct eigenvalues of \( C_k \). Then, \( \lambda_2 = \mu_2 \) with degeneracy \( w_2 = v_2 \) for any \( k \geq 2 \).
\end{repconjecture}

As formalized by Proposition~\ref{prop:basis_change}, the statements above hold for any transfer matrix \( T_\rho \) with \( \rho \in S_k \) because \( T_\rho \) is similar to \( C_k \).

\begin{restatable}{proposition}{basischange} \label{prop:basis_change}
	For any \( \rho \in S_k \), \begin{align}
		T_\rho = Q_\rho^T C_k Q_\rho \qquad \text{with} \qquad Q_\rho = \sum_{\pi \in S_k} \ketbra{\rho \pi}{\pi}.
	\end{align}
\end{restatable}

\section{PROOF OF PROPOSITION~\ref{prop:conjugacy_class}} \label{app:conjugacy_class}

\conjugacyclass*

\begin{proof}
	It holds that \begin{align}
		\fv{\tau}{C_k}{\theta} & = \sum_{\sigma \in S_k} \Wg \left ( \sigma \tau^{- 1}, d D \right ) d^{\# (\sigma)} D^{\# \left ( \sigma \theta^{- 1} \right )} \\
		& = \sum_{\varphi \in S_k} \Wg \left ( \pi^{- 1} \varphi \pi \tau^{- 1}, d D \right ) d^{\# \left ( \pi^{- 1} \varphi \pi \right )} D^{\# \left ( \pi^{- 1} \varphi \pi \theta^{- 1} \right )} \\
		& = \sum_{\varphi \in S_k} \Wg \left ( \varphi \pi \tau^{- 1} \pi^{- 1}, d D \right ) d^{\# \left ( \pi^{- 1} \varphi \pi \right )} D^{\# \left ( \varphi \pi \theta^{- 1} \pi^{- 1} \right )} \\
		& = \sum_{\varphi \in S_k} \Wg \left [ \varphi \left ( \pi \tau \pi^{- 1} \right )^{- 1}, d D \right ] d^{\# (\varphi)} D^{\# \left [ \varphi \left ( \pi \theta \pi^{- 1} \right )^{- 1} \right ]} \\
		& = \fv{\pi \tau \pi^{- 1}}{C_k}{\pi \theta \pi^{- 1}},
	\end{align} where, in the second line, we have used that the conjugation map is an isomorphism. In the third line, we have used that \( \Wg (\alpha, q) = \Wg \left ( \beta \alpha \beta^{- 1}, q \right ) \) and that \( \# (\alpha) = \# \left ( \beta \alpha \beta^{- 1} \right ) \). Both identities are a result of the two functions being sensitive only to the conjugacy class of a given permutation.
\end{proof}

\section{PROOF OF PROPOSITION~\ref{prop:block_triangular}} \label{app:block_triangular}

\blocktriangular*

To prove Proposition~\ref{prop:block_triangular}, we will need a number of ingredients. The main one will be Proposition~2.2 of Ref.~\cite{Collins2017}, which we state without proof as Lemma~\ref{lem:collins_1}.

\begin{lemma} \label{lem:collins_1}
	For any \( \rho \in S_k \), \begin{align}
		\sum_{i = 1}^{k - 1} \Wg \bigl ( (i k) \pi, q \bigr ) + q \Wg (\pi, q) = \delta_{k, \pi (k)} \Wg \left ( \pi^\downarrow, q \right ),
	\end{align} where \( (i k) \) denotes the transposition of elements \( i \) and \( k \).
\end{lemma}

To use Lemma~\ref{lem:collins_1}, we must do two things. First, we must split the sum in Eq.~\eqref{eq:stucafe} into certain sums of \( k \) terms. We employ Lemma~\ref{lem:grouping} to achieve this.

\begin{lemma} \label{lem:grouping}
	For any \( \alpha \in S_k \), there exists a \( \beta \in S_k \) with \( \beta (k) = k \) such that either \( \alpha = \beta \) or \( \alpha = (i k) \beta \) with \( i \in \{ 1, \dots, k - 1 \} \).
\end{lemma}

\begin{proof}
	If \( \alpha (k) = k \), then \( \alpha = \beta \). Otherwise, \( k \) is in exactly one of the disjoint cycles of \( \alpha \). Without loss of generality, say \( \alpha = (k i a_3 a_4 \cdots) (b_1 b_2 \cdots) \cdots \). Then, \( \alpha = (i k) (i a_3 a_4 \cdots) (b_1 b_2 \cdots) \cdots \equiv (i k) \beta \).
\end{proof}

Second, we must ensure that summands in Eq.~\eqref{eq:stucafe} with \( \Wg (\pi, d D) \) have a factor \( d D \) while those with \( \Wg \bigl ( (i k) \pi, d D \bigr ) \) do not. We achieve this with Lemma~\ref{lem:theta} whose proof employs Lemma~4 of Ref.~\cite{Akers2022}, which we state without proof as Lemma~\ref{lem:french}. The lemma also appears as Lemma~1 in Ref.~\cite{Biane1997}.

\begin{lemma} \label{lem:french}
	Let \( \alpha \in S_k \) be a transposition, and \( \beta \in S_k \) such that \( \# (\beta) = u \). If the elements exchanged by \( \alpha \) are not in the same cycle of \( \beta \), then \( \# (\alpha \beta) = u - 1 \).
\end{lemma}

\begin{lemma} \label{lem:theta}
	Let \( \varphi, \theta \in S_k \) with \( \varphi (k) = k \) and \( \theta (k) = k \). Then, \begin{enumerate}
		\item \( \# \left ( \varphi \theta^{- 1} \right ) = \# \left ( (i k) \varphi \theta^{- 1} \right ) + 1 \) for all \( i \in \{ 1, \dots, k - 1 \} \).
		\item \( \# \left ( \varphi \theta^{- 1} \right ) - \# (\varphi) = \# \left ( (i k) \varphi \theta^{- 1} \right ) - \# \bigl ( (i k) \varphi \bigr ) \) for all \( i \in \{ 1, \dots, k - 1 \} \).
	\end{enumerate}
\end{lemma}

\begin{proof}[Proof of 1]
	Let \( \varphi, \theta \in S_k \) with \( \varphi (k) = k \) and \( \theta (k) = k \). Then, \( \left ( \varphi \theta^{- 1} \right ) (k) = k \). That is, \( k \) is in a cycle by itself and thus not in the same cycle of \( \varphi \theta^{- 1} \) as \( i \). The statement follows with Lemma~\ref{lem:french}.
\end{proof}

\begin{proof}[Proof of 2]
	Let \( \varphi, \theta \in S_k \) with \( \varphi (k) = k \) and \( \theta (k) = k \). Then, \( i \) and \( k \) are not in the same cycle of neither \( \varphi \theta^{- 1} \) nor \( \varphi \). By Lemma~\ref{lem:french}, \( \# \left ( \varphi \theta^{- 1} \right ) = \# \left ( (i k) \varphi \theta^{- 1} \right ) + 1 \) and \( \# (\varphi) = \# \bigl ( (i k) \varphi \bigr ) + 1 \). Thus, \( \# \left ( \varphi \theta^{- 1} \right ) - \# (\varphi) = \# \left ( (i k) \varphi \theta^{- 1} \right ) + 1 - \# \bigl ( (i k) \varphi \bigr ) - 1 = \left ( (i k) \varphi \theta^{- 1} \right ) - \# \bigl ( (i k) \varphi \bigr ) \).
\end{proof}

With that, we can prove Proposition~\ref{prop:block_triangular}.

\begin{proof}[Proof of Proposition~\ref{prop:block_triangular}]
	By Lemma~\ref{lem:grouping}, \begin{align}
		\fv{\tau}{C_k}{\theta} & = \sum_{\sigma \in S_k} \Wg \left ( \sigma \tau^{- 1}, d D \right ) d^{\# (\sigma)} D^{\# \left ( \sigma \theta^{- 1} \right )} \\
		& = \sum_{\substack{\varphi \in S_k \\ \varphi (k) = k}} \left \{ \sum_{i = 1}^{k - 1} \Wg \left ( (i k) \varphi \tau^{- 1}, d D \right ) d^{\# \bigl ( (i k) \varphi \bigr )} D^{\# \left ( (i k) \varphi \theta^{- 1} \right )} + \Wg \left ( \varphi \tau^{- 1}, d D \right ) d^{\# (\varphi)} D^{\# \left ( \varphi \theta^{- 1} \right )} \right \}.
	\end{align} 	
	
	By Lemma~\ref{lem:theta}, for \( \theta \in S_k \) with \( \theta (k) = k \), \begin{align}
		\fv{\tau}{C_k}{\theta} & = \sum_{\substack{\varphi \in S_k \\ \varphi (k) = k}} \frac{1}{d^{\# \left ( \varphi \theta^{- 1} \right ) - \# (\varphi)}} \left \{ \sum_{i = 1}^{k - 1} \Wg \left ( (i k) \varphi \tau^{- 1}, d D \right ) (d D)^{\# \left ( (i k) \varphi \theta^{- 1} \right )} + \Wg \left ( \varphi \tau^{- 1}, d D \right ) (d D)^{\# \left ( \varphi \theta^{- 1} \right )} \right \} \\
		& = \sum_{\substack{\varphi \in S_k \\ \varphi (k) = k}} \frac{(d D)^{\# \left ( \varphi \theta^{- 1} \right ) - 1}}{d^{\# \left ( \varphi \theta^{- 1} \right ) - \# (\varphi)}} \left \{ \sum_{i = 1}^{k - 1} \Wg \left ( (i k) \varphi \tau^{- 1}, d D \right ) + d D \Wg \left ( \varphi \tau^{- 1}, d D \right ) \right \} \\
		& = \sum_{\substack{\varphi \in S_k \\ \varphi (k) = k}} d^{\# (\varphi) - 1} D^{\# \left ( \varphi \theta^{- 1} \right ) - 1} \left \{ \sum_{i = 1}^{k - 1} \Wg \left ( (i k) \varphi \tau^{- 1}, d D \right ) + d D \Wg \left ( \varphi \tau^{- 1}, d D \right ) \right \} \\
		& = \sum_{\substack{\varphi \in S_k \\ \varphi (k) = k}} d^{\# \left ( \varphi^\downarrow \right )} D^{\# \left [ \left ( \varphi \theta^{- 1} \right ]^\downarrow \right )} \left \{ \sum_{i = 1}^{k - 1} \Wg \left ( (i k) \varphi \tau^{- 1}, d D \right ) + d D \Wg \left ( \varphi \tau^{- 1}, d D \right ) \right \},
	\end{align} where, in the final line, we have used that \( \# \left ( \alpha^\downarrow \right ) = \# (\alpha) - 1 \).
	
	By Lemma~\ref{lem:collins_1}, \begin{align}
		\fv{\tau}{C_k}{\theta} & = \sum_{\substack{\varphi \in S_k \\ \varphi (k) = k}} \delta_{k, \left ( \varphi \tau^{- 1} \right ) (k)} \Wg \left [ \left ( \varphi \tau^{- 1} \right )^\downarrow, d D \right ] d^{\# \left ( \varphi^\downarrow \right )} D^{\# \left [ \left ( \varphi \theta^{- 1} \right )^\downarrow \right ]} \\
		& = \delta_{k, \tau (k)} \sum_{\substack{\varphi \in S_k \\ \varphi (k) = k}} \Wg \left [ \left ( \varphi \tau^{- 1} \right )^\downarrow, d D \right ] d^{\# \left ( \varphi^\downarrow \right )} D^{\# \left [ \left ( \varphi \theta^{- 1} \right )^\downarrow \right ]} \\
		& = \delta_{k, \tau (k)} \fv{\tau^\downarrow}{C_{k - 1}}{\theta^\downarrow},
	\end{align} where, in the second line, we have used that \( \left ( \varphi \tau^{- 1} \right ) (k) = k \) if and only if \( \tau (k) = k \) because \( \varphi (k) = k \).
	
	This concludes the proof.
\end{proof}

\section{PROOF OF PROPOSITION~\ref{prop:basis_change}} \label{app:basis_change}

\basischange*

\begin{proof}
	It holds that \begin{align}
		\fv{\tau}{T_\rho}{\theta} & = \sum_{\sigma \in S_k} \Wg \left ( \sigma \tau^{- 1}, d D \right ) d^{\# (\sigma \rho)} D^{\# \left ( \sigma \theta^{- 1} \right )} \\
		& = \sum_{\sigma \in S_k} \Wg \left ( \rho \sigma \tau^{- 1} \rho^{- 1}, d D \right ) d^{\# (\rho \sigma)} D^{\# \left ( \rho \sigma \theta^{- 1} \rho^{- 1} \right )} \\
		& = \sum_{\sigma \in S_k} \Wg \left [ \rho \sigma (\rho \tau)^{- 1}, d D \right ] d^{\# (\rho \sigma)} D^{\# \left [ \rho \sigma (\rho \theta)^{- 1} \right ]} \\
		& = \sum_{\pi \in S_k} \Wg \left [ \pi (\rho \tau)^{- 1}, d D \right ] d^{\# (\pi)} D^{\# \left [ \pi (\rho \theta)^{- 1} \right ]} \\
		& = \fv{\rho \tau}{C_k}{\rho \theta},
	\end{align} where, in the second line, we have used that \( \Wg (\alpha, q) = \Wg \left ( \beta \alpha \beta^{- 1}, q \right ) \) and that \( \# (\alpha) = \# \left ( \beta \alpha \beta^{- 1} \right ) \). Both identities are a result of the two functions being sensitive only to the conjugacy class of a given permutation. In the fourth line, we have used that the left-multiplication map is an isomorphism.
\end{proof}

\section{PROOF OF RESULT~\ref{res:mps_norm}} \label{app:mps_norm}

\mpsnorm*

\begin{proof}
    We split the proof into four steps, following the structure of the proof of Result~\ref{res:mps_renyi}.

    \begin{step} \label{step:mps_norm_1}
        We rewrite \( \mathbb{E} N (A : B) \) in terms of expressions of the form of Eq.~\eqref{eq:mps_caffe}. With the Hilbert-Schmidt inner product, \begin{align}
            \mathbb{E} N (A : B) = \mathbb{E} \tr \left ( \varrho_{A B}^2 \right ) + \mathbb{E} \tr \left ( \varrho_A^2 \right ) \tr \left ( \varrho_B^2 \right ) - 2 \mathbb{E} \tr \left [ \varrho_{A B} \left ( \varrho_A \otimes \varrho_B \right ) \right ].
        \end{align} It is then easy to confirm that \begin{align} \label{eq:mps_norm_5}
    		\mathbb{E} N (A : B) & = \tr \left \{ \left [ \left ( P_e^{(d)} \right )^{\otimes c} \otimes \left ( P_{(1 2)}^{(d)} \right )^{\otimes a} \otimes \left ( P_e^{(d)} \right )^{\otimes r} \otimes \left ( P_{(1 2)}^{(d)} \right )^{\otimes b} \otimes \left ( P_e^{(d)} \right )^{\otimes (f - 1)} \otimes P_e^{(d D)} \right ] \mathbb{E} \ketbra{\psi}{\psi}^{\otimes 4} \right \} \nonumber \\
    		& \phantom{{} = {}} + \tr \left \{ \left [ \left ( P_e^{(d)} \right )^{\otimes c} \otimes \left ( P_{(1 2)}^{(d)} \right )^{\otimes a} \otimes \left ( P_e^{(d)} \right )^{\otimes r} \otimes \left ( P_{(3 4)}^{(d)} \right )^{\otimes b} \otimes \left ( P_e^{(d)} \right )^{\otimes (f - 1)} \otimes P_e^{(d D)} \right ] \mathbb{E} \ketbra{\psi}{\psi}^{\otimes 4} \right \} \nonumber \\
    		& \phantom{{} = {}} - 2 \tr \left \{ \left [ \left ( P_e^{(d)} \right )^{\otimes c} \otimes \left ( P_{(1 2)}^{(d)} \right )^{\otimes a} \otimes \left ( P_e^{(d)} \right )^{\otimes r} \otimes \left ( P_{(1 3)}^{(d)} \right )^{\otimes b} \otimes \left ( P_e^{(d)} \right )^{\otimes (f - 1)} \otimes P_e^{(d D)} \right ] \mathbb{E} \ketbra{\psi}{\psi}^{\otimes 4} \right \}.
    	\end{align}
    \end{step}
    
    \begin{step} \label{step:mps_norm_2}
        We express \( \mathbb{E} N (A : B) \) in terms of the transfer matrices defined in Sec.~\ref{sec:mps_transfer}. Given the previous step, it is easy to confirm that \begin{align}
    		\mathbb{E} N (A : B) & = \fv{I_4}{T_e^c T_{(1 2)}^a T_e^r T_{(1 2)}^b}{F_4} + \fv{I_4}{T_e^c T_{(3 4)}^a T_e^r T_{(1 2)}^b}{F_4} - 2 \fv{I_4}{T_e^c T_{(1 2)}^a T_e^r T_{(1 3)}^b}{F_4} \\
    		& = \fv{I_4}{T_e^c T_{(1 2)}^a T_e^r \left ( T_{(1 2)}^b + T_{(3 4)}^b - 2 T_{(1 3)}^b \right )}{F_4} \\
        	& \equiv \fv{I_4}{T_e^c \mathcal{A} T_e^r \mathcal{B}}{F_4},
    	\end{align} where we have defined \begin{align}
    		\mathcal{A} = T_{(1 2)}^a \qquad \text{and} \qquad \mathcal{B} = T_{(1 2)}^b + T_{(3 4)}^b - 2 T_{(1 3)}^b.
    	\end{align}
    \end{step}
    
    \begin{step} \label{step:mps_norm_3}
        We expand \( \mathbb{E} N (A : B) \) in terms of the spectrum of \( T_e \) with \( e \in S_4 \) [see Eq.~\eqref{eq:expansion}]. Let \( \lambda_1 > \lambda_2 > \cdots \geq 0 \) denote the distinct eigenvalues of \( T_e \). It holds~\cite{GitHub} that \begin{align}
        	\lambda_1 = 1 \qquad \text{and} \qquad \lambda_2 = \frac{d D^2 - d}{d^2 D^2 - 1}.
        \end{align} The former is nondegenerate, while the degeneracy of the latter is given by the number of transposition in \( S_4 \)~\cite{GitHub}, \begin{align}
        	w_2 = \binom{4}{2} = 6.
        \end{align} Expanding \( T_e^c \) and taking the limit of \( c \to \infty \) yields \begin{align}
    	    \mathbb{E} N (A : B) & = \fv{L_1}{\mathcal{A} T_e^r \mathcal{B}}{F_4},
    	\end{align} where we have used that \( \braket{I_4}{R_1} = 1 \). After expanding \( T_e^r \) and using that \( \ket{F_4} = \ket{R_1} \), we have \begin{align}
    	    \mathbb{E} N (A : B) & = \fv{L_1}{\mathcal{A}}{R_1} \fv{L_1}{\mathcal{B}}{R_1} + \lambda_2^r \sum_{\mu = 1}^{w_2} \fv{L_1}{\mathcal{A}}{R_2^{(\mu)}} \fv{L_2^{(\mu)}}{\mathcal{B}}{R_1} + O \left ( \lambda_3^{r} \right ) \\
    	    & = \lambda_2^r \sum_{\mu = 1}^{w_2} \fv{L_1}{\mathcal{A}}{R_2^{(\mu)}} \fv{L_2^{(\mu)}}{\mathcal{B}}{R_1} + O \left ( \lambda_3^{r} \right ),
    	\end{align} where, in the second line, we have used that \begin{align} \label{eq:mps_temp}
    	    \fv{L_1}{\mathcal{B}}{R_1} = 0,
    	\end{align} which we prove in the following.
    	
    	In particular, we prove that \( \fv{L_1}{T_t^b}{R_1} \) does not depend on the two elements the transposition \( t \in S_4 \) acts upon. We need some understanding of \( T_t^b \) as well as the eigenvectors \( \bra{L_1} \) and \( \ket{R_1} \) of \( T_e \). For the former, we make use of the fact that \( T_\rho^b \) with \( \rho \in S_4 \) is similar to to \( T_e^b \) with \( e \in S_4 \), \begin{align}
    		T_\rho^b & = \sum_{\pi, \varphi \in S_4} \ketbra{\pi}{\rho \pi} C_4^b \ketbra{\rho \varphi}{\varphi}.
   	    \end{align} With \begin{align}
        	\alpha = \frac{d^2 D - D}{d^2 D^2 - 1}, \qquad \beta = \frac{d D^2 - d}{d^2 D^2 - 1}, \qquad f (u) = \sum_{i = 0}^{u - 1} \alpha \beta^i, \qquad \text{and} \qquad g (u) = \beta^u,
        \end{align} we have \begin{align}
            \Bigl ( T_e \ket{s_i} \Bigr )_{1 \leq i \leq 7} = \begin{pmatrix}
        		1 & \alpha & \alpha & \alpha & \alpha & \alpha & \alpha \\
        		0 & \beta & 0 & 0 & 0 & 0 & 0 \\
        		0 & 0 & \beta & 0 & 0 & 0 & 0 \\
        		0 & 0 & 0 & \beta & 0 & 0 & 0 \\
        		0 & 0 & 0 & 0 & \beta & 0 & 0 \\
        		0 & 0 & 0 & 0 & 0 & \beta & 0 \\
        		0 & 0 & 0 & 0 & 0 & 0 & \beta \\
        		0 & 0 & 0 & 0 & 0 & 0 & 0 \\
        		0 & 0 & 0 & 0 & 0 & 0 & 0 \\
        		0 & 0 & 0 & 0 & 0 & 0 & 0 \\
        		0 & 0 & 0 & 0 & 0 & 0 & 0 \\
        		0 & 0 & 0 & 0 & 0 & 0 & 0 \\
        		0 & 0 & 0 & 0 & 0 & 0 & 0 \\
        		0 & 0 & 0 & 0 & 0 & 0 & 0 \\
        		0 & 0 & 0 & 0 & 0 & 0 & 0 \\
        		0 & 0 & 0 & 0 & 0 & 0 & 0 \\
        		0 & 0 & 0 & 0 & 0 & 0 & 0 \\
        		0 & 0 & 0 & 0 & 0 & 0 & 0 \\
        		0 & 0 & 0 & 0 & 0 & 0 & 0 \\
        		0 & 0 & 0 & 0 & 0 & 0 & 0 \\
        		0 & 0 & 0 & 0 & 0 & 0 & 0 \\
        		0 & 0 & 0 & 0 & 0 & 0 & 0 \\
        		0 & 0 & 0 & 0 & 0 & 0 & 0 \\
        		0 & 0 & 0 & 0 & 0 & 0 & 0
        	\end{pmatrix} \qquad \text{and} \qquad \Bigl ( T_e^b \ket{s_i} \Bigr )_{1 \leq i \leq 7} = \begin{pmatrix}
        		1 & f (b) & f (b) & f (b) & f (b) & f (b) & f (b) \\
        		0 & g (b) & 0 & 0 & 0 & 0 & 0 \\
        		0 & 0 & g (b) & 0 & 0 & 0 & 0 \\
        		0 & 0 & 0 & g (b) & 0 & 0 & 0 \\
        		0 & 0 & 0 & 0 & g (b) & 0 & 0 \\
        		0 & 0 & 0 & 0 & 0 & g (b) & 0 \\
        		0 & 0 & 0 & 0 & 0 & 0 & g (b) \\
        		0 & 0 & 0 & 0 & 0 & 0 & 0 \\
        		0 & 0 & 0 & 0 & 0 & 0 & 0 \\
        		0 & 0 & 0 & 0 & 0 & 0 & 0 \\
        		0 & 0 & 0 & 0 & 0 & 0 & 0 \\
        		0 & 0 & 0 & 0 & 0 & 0 & 0 \\
        		0 & 0 & 0 & 0 & 0 & 0 & 0 \\
        		0 & 0 & 0 & 0 & 0 & 0 & 0 \\
        		0 & 0 & 0 & 0 & 0 & 0 & 0 \\
        		0 & 0 & 0 & 0 & 0 & 0 & 0 \\
        		0 & 0 & 0 & 0 & 0 & 0 & 0 \\
        		0 & 0 & 0 & 0 & 0 & 0 & 0 \\
        		0 & 0 & 0 & 0 & 0 & 0 & 0 \\
        		0 & 0 & 0 & 0 & 0 & 0 & 0 \\
        		0 & 0 & 0 & 0 & 0 & 0 & 0 \\
        		0 & 0 & 0 & 0 & 0 & 0 & 0 \\
        		0 & 0 & 0 & 0 & 0 & 0 & 0 \\
        		0 & 0 & 0 & 0 & 0 & 0 & 0
        	\end{pmatrix}.
        \end{align} For getting some understanding of \( \bra{L_1} \), we make use of the fact that \( \braket{L_i^{(\mu)}}{R_j^{(\nu)}} = \delta_{i j} \delta_{\mu \nu} \). It is easy to confirm that \begin{align}
            \ket{R_1} = \ket{s_1} \qquad \text{and} \qquad \ket{R_2^{(\mu)}} = \frac{\alpha}{\beta - 1} \ket{s_1} + \ket{s_{\mu + 1}},
        \end{align} which implies that \begin{align}
    		\Bigl ( \braket{L_1}{s_i} \Bigr )_{1 \leq i \leq 7} = \left ( 1 \quad \frac{\alpha}{\beta - 1} \quad \frac{\alpha}{\beta - 1} \quad \frac{\alpha}{\beta - 1} \quad \frac{\alpha}{\beta - 1} \quad \frac{\alpha}{\beta - 1} \quad \frac{\alpha}{\beta - 1} \right ).
    	\end{align} For any transposition \( t \in S_4 \), it thus holds that \begin{align}
    		\fv{L_1}{T_t^b}{R_1} & = \sum_{\pi, \varphi \in S_4} \braket{L_1}{\pi} \fv{t \pi}{C_4^b}{t \varphi} \braket{\varphi}{R_1} \\
    		& = \sum_{\pi, \varphi \in S_4} \braket{L_1}{\pi} \fv{t \pi}{C_4^b}{t \varphi} \delta_{s_1 \varphi} \\
    		& = \sum_{\pi \in S_4} \braket{L_1}{\pi} \fv{t \pi}{C_4^b}{t} \\
    		& = \sum_{\pi \in S_4} f (b) \braket{L_1}{\pi} \braket{t \pi}{s_1} + \sum_{\pi \in S_4} g (b) \braket{L_1}{\pi} \braket{t \pi}{t} \\
    		& = \sum_{\pi \in S_4} f (b) \braket{L_1}{\pi} \delta_{t \pi} + \sum_{\pi \in S_4} g (b) \braket{L_1}{\pi} \delta_{s_1 \pi} \\
    		& = f (b) \braket{L_1}{t} + g (b) \braket{L_1}{s_1} \\
    		& = \frac{\alpha}{\beta - 1} f (b) + g (b),
    	\end{align} which is independent of the two elements the transposition \( t \in S_4 \) acts upon. Thus, \begin{align}
    	    \fv{L_1}{\mathcal{B}}{R_1} = \fv{L_1}{\left ( T_{(1 2)}^b + T_{(3 4)}^b - 2 T_{(1 3)}^b \right )}{R_1} = 0.
    	\end{align}
    \end{step}

    \begin{step} \label{step:mps_norm_4}
        Finally, we can write \( \mathbb{E} N (A : B) \) in the form of Definition~\ref{def:definition}. That is, \begin{align}
            \mathbb{E} N (A : B) \equiv K \exp \left ( - \frac{r}{\xi} \right ) + O \left [ \exp \left ( - \frac{r}{\chi} \right ) \right ],
        \end{align} where \begin{align}
    		K = \sum_{\mu = 1}^{w_2} \fv{L_1}{\mathcal{A}}{R_2^{(\mu)}} \fv{L_2^{(\mu)}}{\mathcal{B}}{R_1}
    	\end{align} and \begin{align}
    		\xi = - \frac{1}{\log (\lambda_2)} = - \left [ \log \left ( \frac{d D^2 - d}{d^2 D^2 - 1} \right ) \right ]^{- 1} = \xone > \chi.
        \end{align}
    \end{step}
    
    \noindent This concludes the proof.
\end{proof}

\section{PROOF OF RESULT~\ref{res:mps_neumann} AND COROLLARY~\ref{cor:mps_alpha}} \label{app:mps_neumann}

\mpsneumann*

\mpsalpha*

In this Appendix, we prove Result~\ref{res:mps_neumann} and Corollary~\ref{cor:mps_alpha}. The proof of the latter will follow directly from the proof of the former.

\begin{proof}[Proof of Result~\ref{res:mps_neumann}]
    We split the proof into four steps, following the structure of the proof of Result~\ref{res:mps_renyi}. Because \( I (A : B) \) and \( I_2 (A : B) \) are related, the steps are overall very similar. As is usual in the context of the replica trick \cite{Castellani2005}, we will interchange the order of some limits without rigorous justification.
    
    \begin{step} \label{step:mps_neumann_1}
        We rewrite \( \mathbb{E} I (A : B) \) in terms of expressions of the form of Eq.~\eqref{eq:mps_caffe}. To that end, we make use of Eqs.~\eqref{eq:mps_neumann_1} and \eqref{eq:mps_neumann_2}. With those, \begin{align}
        	\mathbb{E} I (A : B) = \lim_{\alpha \to 1} \lim_{v \to 0} \frac{1}{v \alpha - v} \left \{ \log \bigl [ \mathbb{E} \tr \left ( \varrho_{A B}^\alpha \right )^v \bigr ] - \log \bigl [ \mathbb{E} \tr \left ( \varrho_A^\alpha \right )^v \bigr ] - \log \bigl [ \mathbb{E} \tr \left ( \varrho_B^\alpha \right )^v \bigr ] \right \}.
        \end{align} \( \mathbb{E} \tr \left ( \varrho_A^\alpha \right )^v \), \( \mathbb{E} \tr \left ( \varrho_A^\alpha \right )^v \), and \( \mathbb{E} \tr \left ( \varrho_{A B}^\alpha \right )^v \) can be written in the desired form. Let us define \( x \in S_\alpha \) to be the cyclic permutation so that \( x (i) = i + 1 \) modulo \( \alpha \) and \begin{align}
        	x_w = \bigl ( \alpha (w - 1) + 1, \alpha (w - 1) + 2, \dots, \alpha w \bigr ) \in S_{v \alpha}
        \end{align} with \( w \in \{ 1, \dots, v \} \). Then, for example, \begin{align}
            \mathbb{E} \tr \left ( \varrho_{A B}^\alpha \right )^v & = \tr \left \{ \left [ \left ( P_e^{(d)} \right )^{\otimes c} \otimes \left ( P_x^{(d)} \right )^{\otimes a} \otimes \left ( P_e^{(d)} \right )^{\otimes r} \otimes \left ( P_x^{(d)} \right )^{\otimes b} \otimes \left ( P_e^{(d)} \right )^{\otimes (f - 1)} \otimes P_e^{(d D)} \right ] \mathbb{E} \ketbra{\psi}{\psi}^{\otimes \alpha} \right \}^v \\
        	& = \tr \left \{ \left [ \left ( P_e^{(d)} \right )^{\otimes c} \otimes \left ( P_{x_1 \cdots x_v}^{(d)} \right )^{\otimes a} \otimes \left ( P_e^{(d)} \right )^{\otimes r} \otimes \left ( P_{x_1 \cdots x_v}^{(d)} \right )^{\otimes b} \otimes \left ( P_e^{(d)} \right )^{\otimes (f - 1)} \otimes P_e^{(d D)} \right ] \mathbb{E} \ketbra{\psi}{\psi}^{\otimes v \alpha} \right \}.
        \end{align}
    \end{step}
    
    \begin{step} \label{step:mps_neumann_2}
        We express \( \mathbb{E} I (A : B) \) in terms of the transfer matrices defined in Sec.~\ref{sec:mps_transfer}. Given the previous step, it is easy to confirm that \begin{align}
        	\mathbb{E} I (A : B) & = \lim_{\alpha \to 1} \lim_{v \to 0} \frac{1}{v \alpha - v} \bigl [ \log \left ( \fv{I_{v \alpha}}{T_e^c T_{x_1 \cdots x_v}^a T_e^r T_{x_1 \cdots x_v}^b}{F_{v \alpha}} \right ) - \log \left ( \fv{I_{v \alpha}}{T_e^c T_{x_1 \cdots x_v}^a}{F_{v \alpha}} \right ) \nonumber \\
        	& \hphantom{{} = \lim_{\alpha \to 1} \lim_{v \to 0} \frac{1}{v \alpha - v} \bigl [} - \log \left ( \fv{I_{v \alpha}}{T_e^{c + a + r} T_{x_1 \cdots x_v}^b}{F_{v \alpha}} \right ) \bigr ] \\
        	& \equiv \lim_{\alpha \to 1} \lim_{v \to 0} \frac{1}{v \alpha - v} \left [ \log \left ( \fv{I_{v \alpha}}{T_e^c \mathcal{A} T_e^r \mathcal{B}}{F_{v \alpha}} \right ) - \log \left ( \fv{I_{v \alpha}}{T_e^c \mathcal{A}}{F_{v \alpha}} \right ) - \log \left ( \fv{I_{v \alpha}}{T_e^{c + a + r} \mathcal{B}}{F_{v \alpha}} \right ) \right ], \label{eq:mps_neumann_3}
        \end{align} where we have defined \begin{align}
        	\mathcal{A} = T_{x_1 \cdots x_v}^a \qquad \text{and} \qquad \mathcal{B} = T_{x_1 \cdots x_v}^b.
        \end{align}
    \end{step}
    
    \begin{step} \label{step:mps_neumann_3}
        We expand \( \mathbb{E} I (A : B) \) in terms of the spectrum of \( T_e \) with \( e \in S_{v \alpha} \) [see Eq.~\eqref{eq:expansion}]. At this point, we assume Conjecture~\ref{con:lambda_2} to hold. That is, for any \( v \alpha \geq 2 \), we assume that \begin{align}
            \lambda_2 = \frac{d D^2 - d}{d^2 D^2 - 1}.
        \end{align} with degeneracy \begin{align}
        	w_2 = \binom{k}{2}.
        \end{align} Expanding \( T_e^c \) and taking the limit of \( c \to \infty \) yields \begin{align}
    	    \mathbb{E} I (A : B) & = \lim_{\alpha \to 1} \lim_{v \to 0} \frac{1}{v \alpha - v} \left [ \log \left ( \fv{L_1}{\mathcal{A} T_e^r \mathcal{B}}{F_{v \alpha}} \right ) - \log \left ( \fv{L_1}{\mathcal{A}}{F_{v \alpha}} \right ) - \log \left ( \fv{L_1}{\mathcal{B}}{F_{v \alpha}} \right ) \right ],
    	\end{align} where we have used that \( \braket{I_{v \alpha}}{R_1} = 1 \). After expanding \( T_e^r \) and using that \( \ket{F_{v \alpha}} = \ket{R_1} \), we have \begin{align}
    	    \mathbb{E} I (A : B) & = \lim_{\alpha \to 1} \lim_{v \to 0} \frac{1}{v \alpha - v} \Biggl \{ \log \left [ \fv{L_1}{\mathcal{A}}{R_1} \fv{L_1}{\mathcal{B}}{R_1} + \lambda_2^r \sum_{\mu = 1}^{w_2} \fv{L_1}{\mathcal{A}}{R_2^{(\mu)}} \fv{L_2^{(\mu)}}{\mathcal{B}}{R_1} + O \left ( \lambda_3^r \right ) \right ] \nonumber \\
    	    & \hphantom{{} = \lim_{\alpha \to 1} \lim_{v \to 0} \frac{1}{v \alpha - v} \Biggl \{} - \log \left ( \fv{L_1}{\mathcal{A}}{R_1} \right ) - \log \left ( \fv{L_1}{\mathcal{B}}{R_1} \right ) \Biggr \}.
    	\end{align}
    \end{step}
    
    \begin{step} \label{step:mps_neumann_4}
        Finally, we can write \( \mathbb{E} I (A : B) \) in the form of Definition~\ref{def:definition}. With \( \Lambda = \max \left ( \left \{ \lambda_2^{2 r}, \lambda_3^r \right \} \right ) \), \begin{align}
        	\mathbb{E} I (A : B) & = \lim_{\alpha \to 1} \lim_{v \to 0} \frac{1}{v \alpha - v} \log \left [ 1 + \lambda_2^r \sum_{\mu = 1}^{w_2} \frac{\fv{L_1}{\mathcal{A}}{R_2^{(\mu)}} \fv{L_2^{(\mu)}}{\mathcal{B}}{R_1}}{\fv{L_1}{\mathcal{A}}{R_1} \fv{L_1}{\mathcal{B}}{R_1}} + O \left ( \lambda_3^r \right ) \right ] \\
        	& = \lim_{\alpha \to 1} \lim_{v \to 0} \frac{1}{v \alpha - v} \left [ \lambda_2^r \sum_{\mu = 1}^{w_2} \frac{\fv{L_1}{\mathcal{A}}{R_2^{(\mu)}} \fv{L_2^{(\mu)}}{\mathcal{B}}{R_1}}{\fv{L_1}{\mathcal{A}}{R_1} \fv{L_1}{\mathcal{B}}{R_1}} + O (\Lambda) \right ] \\
        	& \equiv \lim_{\alpha \to 1} \lim_{v \to 0} \frac{1}{v \alpha - v} \left \{ \widetilde{K} (v \alpha) \exp \left ( - \frac{r}{\xi} \right ) + O \left [ \exp \left ( - \frac{r}{\chi} \right ) \right ] \right \},
        \end{align} where \begin{align}
        	\widetilde{K} (v \alpha) = \sum_{\mu = 1}^{w_2} \frac{\fv{L_1}{\mathcal{A}}{R_2^{(\mu)}} \fv{L_2^{(\mu)}}{\mathcal{B}}{R_1}}{\fv{L_1}{\mathcal{A}}{R_1} \fv{L_1}{\mathcal{B}}{R_1}}
        \end{align} and \begin{align}
        	\xi = - \frac{1}{\log (\lambda_2)} = - \left [ \log \left ( \frac{d D^2 - d}{d^2 D^2 - 1} \right ) \right ]^{- 1} = \xone > \chi.
        \end{align} As \( \xi \) is independent of \( v \alpha \), it cannot be affected by the replica limits. \( \widetilde{K} (v \alpha) \) will converge to some \( K \) that is guaranteed to be independent of \( r \).
    \end{step}
    
    \noindent This concludes the proof.
\end{proof}

\begin{proof}[Proof of Corollary~\ref{cor:mps_alpha}]
    Using Eq.~\eqref{eq:mps_neumann_2}, it holds that \begin{align}
        \mathbb{E} I_\alpha (A : B) = \lim_{v \to 0} \frac{1}{v \alpha - v} \left \{ \log \bigl [ \mathbb{E} \tr \left ( \varrho_{A B}^\alpha \right )^v \bigr ] - \log \bigl [ \mathbb{E} \tr \left ( \varrho_A^\alpha \right )^v \bigr ] - \log \bigl [ \mathbb{E} \tr \left ( \varrho_B^\alpha \right )^v \bigr ] \right \}.
    \end{align} Thus, the proof is identical to that of Result~\ref{res:mps_neumann} without the limit of \( \alpha \to 1 \). The statement follows.
\end{proof}

\section{RESULT~\ref{res:mps_neumann} WITH \texorpdfstring{\( c = 0 \)}{ZERO INITIAL SITES}} \label{app:mps_neumann_c}

In this Appendix, we prove a version of Result~\ref{res:mps_neumann} with \( c = 0 \). Steps~1 and 2 of this proof are identical to Steps~\ref{step:mps_neumann_1} and \ref{step:mps_neumann_2} of the proof of Result~\ref{res:mps_neumann} with \( c = 0 \). That is, at the end of Step~2, we have \begin{align}
    \mathbb{E} I (A : B) & = \lim_{\alpha \to 1} \lim_{v \to 0} \frac{1}{v \alpha - v} \left [ \log \left ( \fv{I_{v \alpha}}{\mathcal{A} C_{v \alpha}^{r} \mathcal{B}}{F_{v \alpha}} \right ) - \log \left ( \fv{I_{v \alpha}}{\mathcal{A}}{F_{v \alpha}} \right ) - \log \left ( \fv{I_{v \alpha}}{C_{v \alpha}^{a + r} \mathcal{B}}{F_{v \alpha}} \right ) \right ].
\end{align} We start the proof at Step~3.

\setcounter{step}{2}

\begin{step}
    We expand \( \mathbb{E} I (A : B) \) in terms of the spectrum of \( T_e \) with \( e \in S_{v \alpha} \) [see Eq.~\eqref{eq:expansion}]. We assume Conjecture~\ref{con:lambda_2} to hold. Expanding \( T_e^r \) and using that \( \braket{I_{v \alpha}}{R_1} = 1 \) yields \begin{align}
	    \mathbb{E} I (A : B) & = \lim_{\alpha \to 1} \lim_{v \to 0} \frac{1}{v \alpha - v} \Biggl \{ \log \left [ \fv{I_{v \alpha}}{\mathcal{A}}{R_1} \fv{L_1}{\mathcal{B}}{F_{v \alpha}} + \lambda_2^r \sum_{\mu = 1}^{w_2} \fv{I_{v \alpha}}{\mathcal{A}}{R_2^{(\mu)}} \fv{L_2^{(\mu)}}{\mathcal{B}}{F_{v \alpha}} + O \left ( \lambda_3^r \right ) \right ] \nonumber \\
    	& \hphantom{{} = \lim_{\alpha \to 1} \lim_{v \to 0} \frac{1}{v \alpha - v} \Biggl \{} - \log \left ( \fv{I_{v \alpha}}{\mathcal{A}}{F_{v \alpha}} \right ) \nonumber \\
    	& \hphantom{{} = \lim_{\alpha \to 1} \lim_{v \to 0} \frac{1}{v \alpha - v} \Biggl \{} \left . {} - \log \left [ \fv{L_1}{\mathcal{B}}{F_{v \alpha}} + \lambda_2^{a + r} \sum_{\mu = 1}^{w_2} \braket{I_{v \alpha}}{R_2^{(\mu)}} \fv{L_2^{(\mu)}}{\mathcal{B}}{F_{v \alpha}} + O \left ( \lambda_3^r \right ) \right ] \right \}.
	\end{align}
\end{step}

\begin{step}
    We write \( \mathbb{E} I (A : B) \) in the form of Definition~\ref{def:definition}. With \( \Lambda = \max \left ( \left \{ \lambda_2^{2 r}, \lambda_3^r \right \} \right ) \), \begin{align}
    	\mathbb{E} I (A : B) & = \lim_{\alpha \to 1} \lim_{v \to 0} \frac{1}{v \alpha - v} \Biggl \{ \log \left [ 1 + \lambda_2^r \sum_{\mu = 1}^{w_2} \frac{\fv{I_{v \alpha}}{\mathcal{A}}{R_2^{(\mu)}} \fv{L_2^{(\mu)}}{\mathcal{B}}{F_{v \alpha}}}{\fv{I_{v \alpha}}{\mathcal{A}}{R_1} \fv{L_1}{\mathcal{B}}{F_{v \alpha}}} + O \left ( \lambda_3^r \right ) \right ] \nonumber \\
    	& \hphantom{{} = \lim_{\alpha \to 1} \lim_{v \to 0} \frac{1}{v \alpha - v} \Biggl \{} - \log \left [ 1 + \lambda_2^{a + r} \sum_{\mu = 1}^{w_2} \braket{I_{v \alpha}}{R_2^{(\mu)}} \frac{\fv{L_2^{(\mu)}}{\mathcal{B}}{F_{v \alpha}}}{\fv{L_1}{\mathcal{B}}{F_{v \alpha}}} + O \left ( \lambda_3^r \right ) \right ] \Biggr \} \\
    	& = \lim_{\alpha \to 1} \lim_{v \to 0} \frac{1}{v \alpha - v} \left [ \lambda_2^r \sum_{\mu = 1}^{w_2} \frac{\fv{I_{v \alpha}}{\mathcal{A}}{R_2^{(\mu)}} \fv{L_2^{(\mu)}}{\mathcal{B}}{F_{v \alpha}}}{\fv{I_{v \alpha}}{\mathcal{A}}{R_1} \fv{L_1}{\mathcal{B}}{F_{v \alpha}}} + \lambda_2^{a + r} \sum_{\mu = 1}^{w_2} \braket{I_{v \alpha}}{R_2^{(\mu)}} \frac{\fv{L_2^{(\mu)}}{\mathcal{B}}{F_{v \alpha}}}{\fv{L_1}{\mathcal{B}}{F_{v \alpha}}} + O (\Lambda) \right ] \\
    	& = \lim_{\alpha \to 1} \lim_{v \to 0} \frac{1}{v \alpha - v} \left \{ \lambda_2^r \sum_{\mu = 1}^{w_2} \left [ \frac{\fv{I_{v \alpha}}{\mathcal{A}}{R_2^{(\mu)}} \fv{L_2^{(\mu)}}{\mathcal{B}}{F_{v \alpha}}}{\fv{I_{v \alpha}}{\mathcal{A}}{R_1} \fv{L_1}{\mathcal{B}}{F_{v \alpha}}} + \lambda_2^a \braket{I_{v \alpha}}{R_2^{(\mu)}} \frac{\fv{L_2^{(\mu)}}{\mathcal{B}}{F_{v \alpha}}}{\fv{L_1}{\mathcal{B}}{F_{v \alpha}}} \right ] + O (\Lambda) \right \} \\
    	& \equiv \lim_{\alpha \to 1} \lim_{v \to 0} \frac{1}{v \alpha - v} \left \{ \widetilde{K'} (v \alpha) \exp \left ( - \frac{r}{\xi} \right ) + O \left [ \exp \left ( - \frac{r}{\chi} \right ) \right ] \right \},
    \end{align} where \begin{align}
    	\widetilde{K'} (v \alpha) = \sum_{\mu = 1}^{w_2} \left [ \frac{\fv{I_{v \alpha}}{\mathcal{A}}{R_2^{(\mu)}} \fv{L_2^{(\mu)}}{\mathcal{B}}{F_{v \alpha}}}{\fv{I_{v \alpha}}{\mathcal{A}}{R_1} \fv{L_1}{\mathcal{B}}{F_{v \alpha}}} + \lambda_2^a \braket{I_{v \alpha}}{R_2^{(\mu)}} \frac{\fv{L_2^{(\mu)}}{\mathcal{B}}{F_{v \alpha}}}{\fv{L_1}{\mathcal{B}}{F_{v \alpha}}} \right ]
    \end{align} and \begin{align}
    	\xi = - \frac{1}{\log (\lambda_2)} = - \left [ \log \left ( \frac{d D^2 - d}{d^2 D^2 - 1} \right ) \right ]^{- 1} \xone > \chi.
    \end{align} Again, \( \widetilde{K'} (v \alpha) \) will converge to some \( K' \) that is guaranteed to be independent of \( r \). While \( K' \) is different from \( K \) in general, the correlation length \( \xi \) is independent of \( c \).
\end{step}

\section{NUMERICAL ANALYSIS} \label{app:numerical_neumann}

In this Appendix, we briefly review our numerical analysis of the von Neumann mutual information \( I (A : B) \) in Sec.~\ref{sec:mps_neumann}. We fix \( d \) and \( D \), and we set \( a = b = 1 \) and \( r = 5 \). (i) We generate \( a + r + b + 1 \) Haar-random unitary matrices of \( U (d D) \) to define \( \ket{\psi} \). This definition makes the assumption that there are no sites before subsystem \( A \) (that is, \( c = 0 \)). We discuss in Appendix~\ref{app:mps_neumann_c} why this does not affect the average correlation length. By setting \( f = 1 \), we furthermore use the fact that the sites after subsystem \( B \) do not play a role as a result of the sequential generation [see Eq.~\eqref{eq:mps_sequential_generation}]. (ii) We compute \( I (A : B) \) with respect to \( \ket{\psi} \). (iii) We repeat steps (i) and (ii) \( 10\,000 \) times to compute the average of \( I (A : B) \). (iv) We repeat steps (i) through (iii) for \( r \in \{ 7, 9, 11, 13, 15 \} \), plot the averages of \( I (A : B) \) against \( r \), and fit the data to extract the average correlation length. (v) To obtain Fig.~\ref{fig:numerical_neumann}, we repeat steps (i) through (iv) for different \( d \) and \( D \).

\section{TRANSFER MATRICES IN TWO DIMENSIONS} \label{app:iso_transfer}

This Appendix expands on Sec.~\ref{sec:iso_transfer}. We will state the definitions of the boundary tensors and prove Eq.~\eqref{eq:iso_sequential_generation}.

From the main text, recall that we define \( V^{(i, j)} = U^{(i, j)} \otimes \overline{U^{(i, j)}} \). By computing the \( k \)-fold twirl, we obtain the building block \begin{align}
	\raisebox{-37pt}{\includegraphics{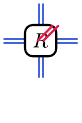}} = \int \mathrm{d} U^{(i, j)} \, \raisebox{-37pt}{\includegraphics{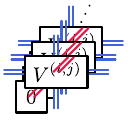}} = \raisebox{-37pt}{\includegraphics{iso_building_block_rhs}}.
\end{align}

With that, we have \begin{align}
	\mathbb{E} \ketbra{\psi}{\psi}^{\otimes k} = \raisebox{-83pt}{\includegraphics{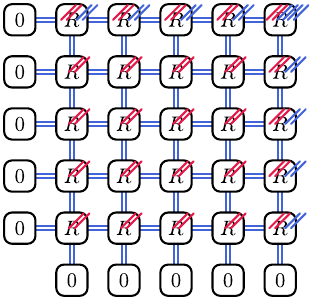}} = \raisebox{-83pt}{\includegraphics{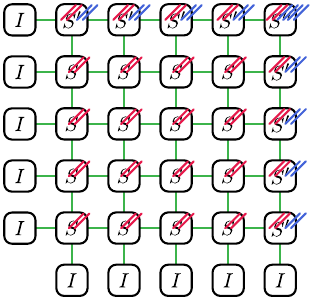}},
\end{align} where we have cut permutation-valued (green) legs instead of bond (blues) ones in the second step. The tensor \( S \) is stated in Eq.~\eqref{eq:iso_s}. \( S' \) and \( S'' \) reflect the different boundary conditions.

After contracting \( S' \) with \( P_\rho^{(d)} \), we obtain \begin{align}
	\raisebox{-24pt}{\includegraphics{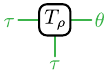}} & = \raisebox{-24pt}{\includegraphics{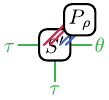}} = \raisebox{-24pt}{\includegraphics{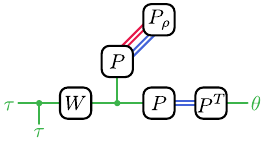}} \\
	& = \sum_{\sigma \in S_k} \Wg \left ( \sigma \tau^{- 1}, d D^2 \right ) (d D)^{\# \left ( \sigma \rho \right )} D^{\# \left ( \sigma \theta^{- 1} \right )}
\end{align} at the top boundary and \begin{align}
	\raisebox{-24pt}{\includegraphics{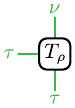}} & = \raisebox{-24pt}{\includegraphics{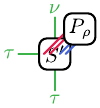}} = \raisebox{-24pt}{\includegraphics{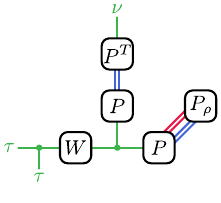}} \\
	& = \sum_{\sigma \in S_k} \Wg \left ( \sigma \tau^{- 1}, d D^2 \right ) (d D)^{\# \left ( \sigma \rho \right )} D^{\# \left ( \sigma \nu^{- 1} \right )}
\end{align} at the right boundary.

We will always contract \( S'' \) with \( P_e^{(d)} \). The tensor in the top-right corner thus plays the same role as the final vector \( \ket{F_k} = e_1 \in \mathbb{R}^{k!} \) does in one dimension. In fact, \( T_e = \ketbra{F_k}{F_k} \), \begin{align}
	\raisebox{-24pt}{\includegraphics{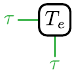}} & = \raisebox{-24pt}{\includegraphics{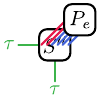}} = \raisebox{-24pt}{\includegraphics{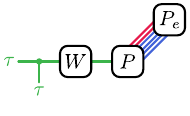}} \equiv \raisebox{-24pt}{\includegraphics{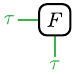}}\\
	& = \sum_{\sigma \in S_k} \Wg \left ( \sigma \tau^{- 1}, d D^2 \right ) \left (d D^2 \right )^{\# \left ( \sigma e \right )} = \delta_{e \tau},
\end{align} where we have defined \begin{align}
    \raisebox{-44pt}{\includegraphics{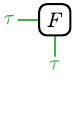}} = \raisebox{-44pt}{\includegraphics{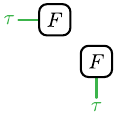}}.
\end{align}

If a tensor corresponding to \( e \in S_k \) is contracted with \( \ket{F_k} \), it factorizes. For tensors at the top boundary, we have \begin{align}
	\raisebox{-24pt}{\includegraphics{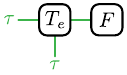}} & = \raisebox{-24pt}{\includegraphics{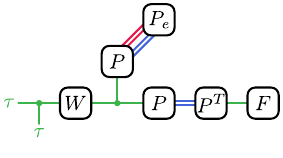}} = \raisebox{-24pt}{\includegraphics{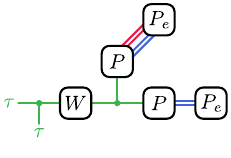}} = \raisebox{-24pt}{\includegraphics{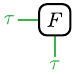}} \\
	& = \sum_{\sigma \in S_k} \Wg \left ( \sigma \tau^{- 1}, d D^2 \right ) \left ( d D^2 \right )^{\# \left ( \sigma e \right )} = \delta_{e \tau},
\end{align} for those at the right boundary, we have \begin{align}
	\raisebox{-24pt}{\includegraphics{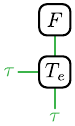}} & = \raisebox{-24pt}{\includegraphics{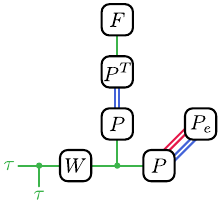}} = \raisebox{-24pt}{\includegraphics{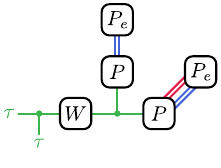}} = \raisebox{-24pt}{\includegraphics{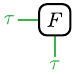}} \\
	& = \sum_{\sigma \in S_k} \Wg \left ( \sigma \tau^{- 1}, d D^2 \right ) \left ( d D^2 \right )^{\# \left ( \sigma e \right )} = \delta_{e \tau},
\end{align} and for those in the bulk, we have \begin{align}
	\raisebox{-24pt}{\includegraphics{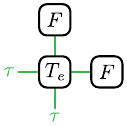}} & = \raisebox{-24pt}{\includegraphics{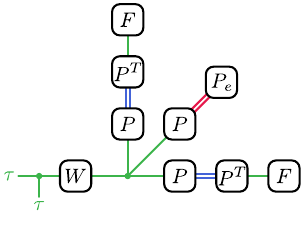}} = \raisebox{-24pt}{\includegraphics{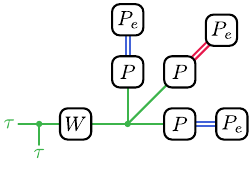}} = \raisebox{-24pt}{\includegraphics{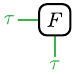}} \\
	& = \sum_{\sigma \in S_k} \Wg \left ( \sigma \tau^{- 1}, d D^2 \right ) \left ( d D^2 \right )^{\# \left ( \sigma e \right )} = \delta_{e \tau}.
\end{align}

The identities above lead to Eq.~\eqref{eq:iso_sequential_generation}: \begin{align}
    \raisebox{-32pt}{\includegraphics{iso_sequential_generation_lhs}} & = \raisebox{-32pt}{\includegraphics{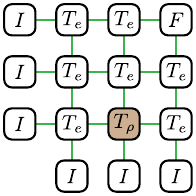}} = \raisebox{-32pt}{\includegraphics{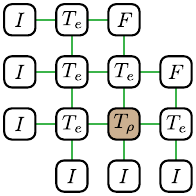}} = \raisebox{-32pt}{\includegraphics{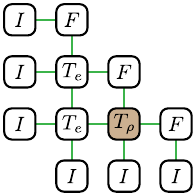}} \\
    & = \raisebox{-32pt}{\includegraphics{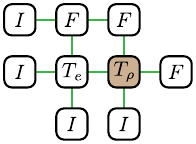}} = \raisebox{-32pt}{\includegraphics{iso_sequential_generation_rhs}}.
\end{align}

\section{SPECTRUM OF THE \texorpdfstring{TRANSFER MATRIX \( \mathcal{T}_e \) with \( e \in S_2 \)}{CURLY 2-COPY IDENTITY TRANSFER MATRIX}} \label{app:iso_spectrum_2}

In this Appendix, we state and prove two lemmas concerning the spectrum of \( \mathcal{T}_e \) with \( e \in S_2 \).

Let us start with some preliminaries. We define \begin{align}
	\ket{0} = \begin{pmatrix}
		1 \\ 0
	\end{pmatrix}, \qquad \ket{1} = \begin{pmatrix}
		0 \\ 1
	\end{pmatrix}, \qquad \text{and} \qquad \ket{+} = \begin{pmatrix}
		1 \\ 1
	\end{pmatrix},
\end{align} and map the contraction of tensors defining \( \mathcal{T}_e \) with \( e \in S_2 \) to a multiplication of matrices: \begin{align} \label{eq:fame}
	\vcenter{\hbox{\includegraphics{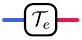}}} = \vcenter{\hbox{\includegraphics{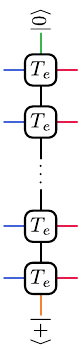}}} = \vcenter{\hbox{\includegraphics{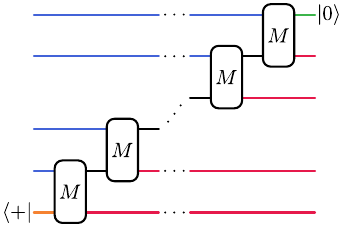}}},
\end{align} where \begin{align}
	\vcenter{\hbox{\includegraphics{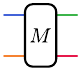}}} = \begin{pmatrix}
		1 & \alpha & \alpha & \gamma \\ 0 & 0 & 0 & 0 \\ 0 & 0 & 0 & 0 \\ 0 & \beta & \beta & \delta
	\end{pmatrix}
\end{align} with \begin{align}
    \alpha = \frac{d^2 D^3 - D}{d^2 D^4 - 1}, \qquad \beta = \frac{d D^3 - d D}{d^2 D^4 - 1}, \qquad \gamma = \frac{d^2 D^2 - D^2}{d^2 D^4 - 1}, \qquad \text{and} \qquad \delta = \frac{d D^2 - d}{d^2 D^4 - 1}.
\end{align}

Note that \begin{align}
    & \vcenter{\hbox{\includegraphics{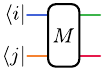}}} = 0 \qquad \text{if} \qquad i \neq j,
\end{align} and \begin{align}
	N = \vcenter{\hbox{\includegraphics{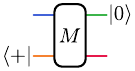}}} = \begin{pmatrix}
		1 & \alpha \\ 0 & \beta
	\end{pmatrix}
\end{align} is equal to \( T_e \) with \( e \in S_2 \) for \( d \to d D \). With that, it easy to check that \begin{align} \label{eq:iso_spectrum_2_5}
    \vcenter{\hbox{\includegraphics{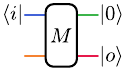}}} = \vcenter{\hbox{\includegraphics{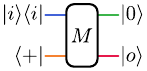}}} = \fv{o}{N}{i} \bra{i}.
\end{align}

We also introduce an analytical notation. We define \begin{align}
    M_j & = I^{\otimes (j - 1)} \otimes M \otimes I^{\otimes (h - j)}.
\end{align} \( \mathcal{T}_e \) with \( e \in S_2 \) is then given by \begin{align}
    \mathcal{T}_e = \left ( I^{\otimes h} \otimes \bra{0} \right ) \bigl ( M_h \cdots M_1 \bigr ) \left ( \ket{+} \otimes I ^{\otimes h} \right ).
\end{align} With \( i_1, \dots, i_h \in \{ 0, 1 \} \) and \( o_1, \dots, o_h \in \{ 0, 1 \} \), the entries of \( \mathcal{T}_e \) with \( e \in S_2 \) are given by \begin{align}
    \bigl ( \bra{o_1, \dots, o_h} \otimes \bra{0} \bigr ) \bigl ( M_h \cdots M_1 \bigr ) \bigl ( \ket{+} \otimes \ket{i_1, \dots, i_h} \bigr ).
\end{align}

Our first lemma states that \( \mathcal{T}_e \) with \( e \in S_2 \) is block triangular, where our definition of blocks arises from the indexing of rows and columns in base \( 2 \). In particular, with \( 2 \leq j \leq h \), the \( j \)th diagonal block of \( \mathcal{T}_e \), which we denote by \( \mathcal{T}_e^{(j)} \in \mathbb{R}^{2^{j - 1} \times 2^{j - 1}} \), has fixed indices \( i_h = \cdots i_{j + 1} = o_h = \cdots = o_{j + 1} = 0 \) and \( i_j = o_j = 1 \). Using Eq.~\eqref{eq:iso_spectrum_2_1}, it is given by \begin{align}
    \mathcal{T}_e^{(j)} = \raisebox{-44pt}{\includegraphics{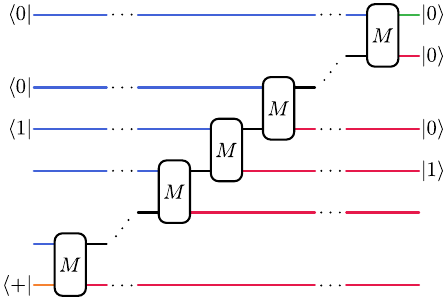}} = \raisebox{-44pt}{\includegraphics{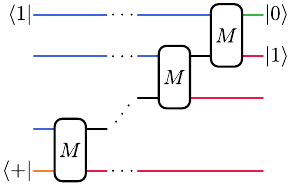}}.
\end{align} The first diagonal block, which we denote by \( \mathcal{T}_e^{(1)} \in \mathbb{R}^{2 \times 2} \), is given by \begin{align}
    \mathcal{T}_e^{(1)} = \raisebox{-15pt}{\includegraphics{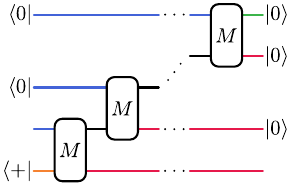}} = \raisebox{-15pt}{\includegraphics{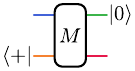}} = \begin{pmatrix} 1 & \alpha \\ 0 & \beta \end{pmatrix}.
\end{align}

In particular, we will prove that a block of \( \mathcal{T}_e \) is zero if its defining row digit \( o_j \) is higher than its defining column digit \( i_j \), which implies that \( \mathcal{T}_e \) is upper block triangular. As the proof relies exclusively on Eq.~\eqref{eq:iso_spectrum_2_5}, \( \mathcal{T}_e \) inherits its upper block triangularity from the upper triangularity of \( N \).

\begin{lemma} \label{lem:iso_spectrum_2_1}
    \( \mathcal{T}_e \) with \( e \in S_2 \) is upper block triangular because \begin{align}
        \left ( I^{\otimes (j - 1)} \otimes \bra{0} \otimes \bra{0}^{\otimes (h - j)} \otimes \bra{0} \right ) \bigl ( M_h \cdots M_1 \bigr ) \left ( \ket{+} \otimes I^{\otimes (j - 1)} \otimes \ket{1} \otimes \ket{0}^{\otimes (h - j)} \right ) = 0.
    \end{align}
\end{lemma}

\begin{proof}
    From \begin{align} \label{eq:iso_spectrum_2_1}
        \vcenter{\hbox{\includegraphics{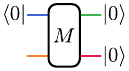}}} = \ket{0}
    \end{align} and \begin{align}
        \vcenter{\hbox{\includegraphics{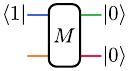}}} = 0,
    \end{align} it follows that \begin{align}
        \raisebox{-44pt}{\includegraphics{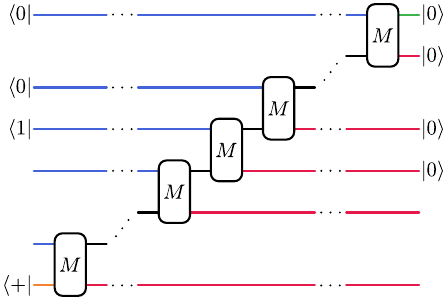}} = \raisebox{-44pt}{\includegraphics{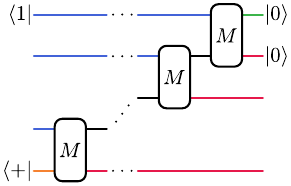}} = 0,
    \end{align} which concludes the proof.
\end{proof}

In our second lemma, we utilize the block triangularity of \( \mathcal{T}_e \) with \( e \in S_2 \) to make a direct statement about its two leading eigenvalues.

\begin{lemma} \label{lem:iso_spectrum_2_2}
    Let \( | \lambda_1 | > | \lambda_2 | > \cdots \geq 0 \) denote the distinct eigenvalues of \( \mathcal{T}_e \) with \( e \in S_2 \). Then, for any \( h \), \( \lambda_1 = 1 \) and \( \lambda_2 = \beta \). Furthermore, \( \lambda_1 \) and \( \lambda_2 \) are nondegenerate.
\end{lemma}

\begin{proof}
    The spectrum of \( \mathcal{T}_e \) with \( e \in S_2 \) is given by the union of the spectra of its diagonal blocks. It is evident that the first block \( \mathcal{T}_e^{(1)} \) has eigenvalues \( 1 \) and \( \beta \). In the following, we show that any other diagonal block \( \mathcal{T}_e^{(j)} \) with \( 2 \leq j \leq h \) can be written as a product of \( \beta \) and a strictly substochastic matrix, implying that its eigenvalues are strictly smaller than \( \beta \). We structure the proof in steps.
    
    \begin{step}
        We show that any diagonal block \( \mathcal{T}_e^{(j)} \) with \( 2 \leq j \leq h \) can be written as \( \beta \) times a matrix. From \begin{align}
            \vcenter{\hbox{\includegraphics{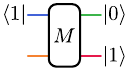}}} = \beta \ket{1},
        \end{align} it follows that \begin{align}
            \mathcal{T}_e^{(j)} = \raisebox{-33pt}{\includegraphics{iso_jth_block_mid}} = \beta \raisebox{-33pt}{\includegraphics{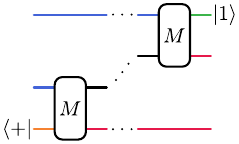}}.
        \end{align}
    \end{step}
    
    \begin{step}
        We argue that the matrix \begin{align} \label{eq:iso_spectrum_2_2}
            \raisebox{-33pt}{\includegraphics{iso_jth_block_rhs}}
        \end{align} is strictly column substochastic. It holds that \begin{align} \label{eq:iso_spectrum_2_4}
            \vcenter{\hbox{\includegraphics{iso_matrix}}} = \begin{pmatrix}
        		1 & \alpha & \alpha & \gamma \\ 0 & 0 & 0 & 0 \\ 0 & 0 & 0 & 0 \\ 0 & \beta & \beta & \delta
        	\end{pmatrix} 
        \end{align} is column substochastic. While the first column of \( M \) evidently sums to \( 1 \), the sums of the other columns are strictly bounded by \( 1 \). As a result, \( M_{j - 1} \cdots M_1 \) is column substochastic. The boundary condition \( \ket{1} \) does not affect this because it specifies a subset of columns of the matrix \( M_{j - 1} \cdots M_1 \). In fact, it imposes strict substochasticity because this subset does not include the only column summing to \( 1 \). Also the boundary condition \( \bra{+} \) does not affect the substochasticity. The boundary condition means that Eq.~\eqref{eq:iso_spectrum_2_2} is a sum of matrices. Each of these matrices comprises a disjoint subset of rows of the matrix \( M_{j - 1} \cdots M_1 \). Because \( M_{j - 1} \cdots M_1 \), the matrix comprising the whole set of rows, is column substochastic, so is the sum of the matrices comprising the disjoint subsets.
    \end{step}
    
    \begin{step}
        \( 1 \) and \( \beta \) are the only eigenvalues of \( \mathcal{T}_e^{(1)} \). They are nondegenerate. Because the matrix in Eq.~\eqref{eq:iso_spectrum_2_4} is strictly column substochastic, the eigenvalues of any diagonal block \( \mathcal{T}_e^{(j)} \) with \( 2 \leq j \leq h \) are strictly smaller than \( \beta \).
    \end{step}
    
    \noindent The statement follows.
\end{proof}

As a preparation for Appendix~\ref{app:iso_spectrum_4}, we provide the proofs of Lemmas~\ref{lem:iso_spectrum_2_1} and \ref{lem:iso_spectrum_2_2} in analytical notation.

\begin{proof}[Proof of Lemma~\ref{lem:iso_spectrum_2_1} in analytical notation]
    From \begin{align} \label{eq:iso_spectrum_2_3}
        \bigl ( \bra{0} \otimes \bra{0} \bigr ) M \bigl ( I \otimes \ket{0} \bigr ) = \ket{0}
    \end{align} and \begin{align}
        \bigl ( \bra{0} \otimes \bra{0} \bigr ) M \bigl ( I \otimes \ket{1} \bigr ) = 0,
    \end{align} it follows that \begin{align}
        & \left ( I^{\otimes (j - 1)} \otimes \bra{0} \otimes \bra{0}^{\otimes (h - j)} \otimes \bra{0} \right ) \bigl ( M_h \cdots M_1 \bigr ) \left ( \ket{+} \otimes I^{\otimes (j - 1)} \otimes \ket{1} \otimes \ket{0}^{\otimes (h - j)} \right ) \nonumber \\
        & \qquad = \left ( I^{\otimes (j - 1)} \otimes \bra{0} \otimes \bra{0} \right ) \bigl ( M_j \cdots M_1 \bigr ) \left ( \ket{+} \otimes I^{\otimes (j - 1)} \otimes \ket{1} \right ) \\
        & \qquad = 0,
    \end{align} which concludes the proof.
\end{proof}

\begin{proof}[Proof of Lemma~\ref{lem:iso_spectrum_2_2} in analytical notation]
    As in the version with graphical notation, we structure the proof in steps, without repeating the details.
    
    \setcounter{step}{0}
    
    \begin{step}
        From \begin{align}
            \bigl ( \bra{1} \otimes \bra{0} \bigr ) M \bigl ( I \otimes \ket{0} \bigr ) = \beta \ket{1},
        \end{align} it follows that \begin{align}
            \mathcal{T}_e^{(j)} & = \left ( I^{\otimes (j - 1)} \otimes \bra{1} \otimes \bra{0} \right ) \bigl ( M_j \cdots M_1 \bigr ) \left ( \ket{+} \otimes I^{\otimes (j - 1)} \otimes \ket{1} \right ) \\
            & = \beta \left ( I^{\otimes (j - 1)} \otimes \bra{1} \right ) \bigl ( M_{j - 1} \cdots M_1 \bigr ) \left ( \ket{+} \otimes I^{\otimes (j - 1)} \right ).
        \end{align}
    \end{step}
    
    \begin{step}
        It holds that \begin{align} \label{eq:iso_spectrum_2_6}
            \left ( I^{\otimes (j - 1)} \otimes \bra{1} \right ) \bigl ( M_{j - 1} \cdots M_1 \bigr ) \left ( \ket{+} \otimes I^{\otimes (j - 1)} \right )
        \end{align} is strictly column substochastic.
    \end{step}
    
    \begin{step}
        \( 1 \) and \( \beta \) are the only eigenvalues of \( \mathcal{T}_e^{(1)} \). They are nondegenerate. Because the matrix in Eq.~\eqref{eq:iso_spectrum_2_6} is strictly column substochastic, the eigenvalues of any diagonal block \( \mathcal{T}_e^{(j)} \) with \( 2 \leq j \leq h \) are strictly smaller than \( \beta \).
    \end{step}
    
    \noindent The statement follows.
\end{proof}

\section{SPECTRUM OF THE \texorpdfstring{TRANSFER MATRIX \( \mathcal{T}_e \) with \( e \in S_4 \)}{CURLY 4-COPY IDENTITY TRANSFER MATRIX}} \label{app:iso_spectrum_4}

In this Appendix, we state and prove two lemmas concerning the spectrum of \( \mathcal{T}_e \) with \( e \in S_4 \). Using the same notation as in Appendix~\ref{app:iso_spectrum_2}, we will draw on results from that Appendix.

\setcounter{MaxMatrixCols}{24}

\( \mathcal{T}_e \) with \( e \in S_4 \) is defined by Eq.~\eqref{eq:fame}, now with \( M \in \mathbb{R}^{576 \times 576} \). As in Appendix~\ref{app:iso_spectrum_2}, it holds that \begin{align} \label{eq:iso_spectrum_4_1}
    \bigl ( \bra{o} \otimes \bra{0} \bigr ) M \bigl ( I \otimes \ket{i} \bigr ) = \fv{o}{N}{i} \bra{i},
\end{align} where \begin{align} \label{eq:iso_spectrum_4_2}
    N = \bigl ( \bra{+} \otimes I \bigr ) M \bigl ( I \otimes \ket{0} \bigr ) = \begin{pmatrix}
		1 & \alpha & \alpha & \alpha & \alpha & \alpha & \alpha & \gamma & \gamma & \gamma & \gamma & \gamma & \gamma & \gamma & \gamma & \eta & \eta & \eta & \eta & \eta & \eta & \rho & \rho & \rho \\
		0 & \beta & 0 & 0 & 0 & 0 & 0 & \delta & \delta & \delta & \delta & 0 & 0 & 0 & 0 & \theta & \theta & \iota & \theta & \iota & \theta & \sigma & \tau & \tau \\
		0 & 0 & \beta & 0 & 0 & 0 & 0 & \delta & \delta & 0 & 0 & \delta & \delta & 0 & 0 & \iota & \theta & \theta & \theta & \theta & \iota & \tau & \sigma & \tau \\
		0 & 0 & 0 & \beta & 0 & 0 & 0 & 0 & 0 & \delta & \delta & \delta & \delta & 0 & 0 & \theta & \iota & \theta & \iota & \theta & \theta & \tau & \tau & \sigma \\
		0 & 0 & 0 & 0 & \beta & 0 & 0 & \delta & \delta & 0 & 0 & 0 & 0 & \delta & \delta & \theta & \iota & \theta & \iota & \theta & \theta & \tau & \tau & \sigma \\
		0 & 0 & 0 & 0 & 0 & \beta & 0 & 0 & 0 & \delta & \delta & 0 & 0 & \delta & \delta & \iota & \theta & \theta & \theta & \theta & \iota & \tau & \sigma & \tau \\
		0 & 0 & 0 & 0 & 0 & 0 & \beta & 0 & 0 & 0 & 0 & \delta & \delta & \delta & \delta & \theta & \theta & \iota & \theta & \iota & \theta & \sigma & \tau & \tau \\
		0 & 0 & 0 & 0 & 0 & 0 & 0 & \varepsilon & \zeta & 0 & 0 & 0 & 0 & 0 & 0 & \kappa & \kappa & \lambda & \lambda & \kappa & \lambda & \upsilon & \upsilon & \upsilon \\
		0 & 0 & 0 & 0 & 0 & 0 & 0 & \zeta & \varepsilon & 0 & 0 & 0 & 0 & 0 & 0 & \lambda & \lambda & \kappa & \kappa & \lambda & \kappa & \upsilon & \upsilon & \upsilon \\
		0 & 0 & 0 & 0 & 0 & 0 & 0 & 0 & 0 & \varepsilon & \zeta & 0 & 0 & 0 & 0 & \kappa & \kappa & \kappa & \lambda & \lambda & \lambda & \upsilon & \upsilon & \upsilon \\
		0 & 0 & 0 & 0 & 0 & 0 & 0 & 0 & 0 & \zeta & \varepsilon & 0 & 0 & 0 & 0 & \lambda & \lambda & \lambda & \kappa & \kappa & \kappa & \upsilon & \upsilon & \upsilon \\
		0 & 0 & 0 & 0 & 0 & 0 & 0 & 0 & 0 & 0 & 0 & \varepsilon & \zeta & 0 & 0 & \kappa & \lambda & \kappa & \kappa & \lambda & \lambda & \upsilon & \upsilon & \upsilon \\
		0 & 0 & 0 & 0 & 0 & 0 & 0 & 0 & 0 & 0 & 0 & \zeta & \varepsilon & 0 & 0 & \lambda & \kappa & \lambda & \lambda & \kappa & \kappa & \upsilon & \upsilon & \upsilon \\
		0 & 0 & 0 & 0 & 0 & 0 & 0 & 0 & 0 & 0 & 0 & 0 & 0 & \varepsilon & \zeta & \kappa & \lambda & \lambda & \kappa & \kappa & \lambda & \upsilon & \upsilon & \upsilon \\
		0 & 0 & 0 & 0 & 0 & 0 & 0 & 0 & 0 & 0 & 0 & 0 & 0 & \zeta & \varepsilon & \lambda & \kappa & \kappa & \lambda & \lambda & \kappa & \upsilon & \upsilon & \upsilon \\
		0 & 0 & 0 & 0 & 0 & 0 & 0 & 0 & 0 & 0 & 0 & 0 & 0 & 0 & 0 & \mu & \nu & \nu & \nu & \nu & \xi & \tau & \varphi & \tau \\
		0 & 0 & 0 & 0 & 0 & 0 & 0 & 0 & 0 & 0 & 0 & 0 & 0 & 0 & 0 & \nu & \mu & \nu & \xi & \nu & \nu & \tau & \tau & \varphi \\
		0 & 0 & 0 & 0 & 0 & 0 & 0 & 0 & 0 & 0 & 0 & 0 & 0 & 0 & 0 & \nu & \nu & \mu & \nu & \xi & \nu & \varphi & \tau & \tau \\
		0 & 0 & 0 & 0 & 0 & 0 & 0 & 0 & 0 & 0 & 0 & 0 & 0 & 0 & 0 & \nu & \xi & \nu & \mu & \nu & \nu & \tau & \tau & \varphi \\
		0 & 0 & 0 & 0 & 0 & 0 & 0 & 0 & 0 & 0 & 0 & 0 & 0 & 0 & 0 & \nu & \nu & \xi & \nu & \mu & \nu & \varphi & \tau & \tau \\
		0 & 0 & 0 & 0 & 0 & 0 & 0 & 0 & 0 & 0 & 0 & 0 & 0 & 0 & 0 & \xi & \nu & \nu & \nu & \nu & \mu & \tau & \varphi & \tau \\
		0 & 0 & 0 & 0 & 0 & 0 & 0 & 0 & 0 & 0 & 0 & 0 & 0 & 0 & 0 & o & o & \pi & o & \pi & o & \chi & \psi & \psi \\
		0 & 0 & 0 & 0 & 0 & 0 & 0 & 0 & 0 & 0 & 0 & 0 & 0 & 0 & 0 & \pi & o & o & o & o & \pi & \psi & \chi & \psi \\
		0 & 0 & 0 & 0 & 0 & 0 & 0 & 0 & 0 & 0 & 0 & 0 & 0 & 0 & 0 & o & \pi & o & \pi & o & o & \psi & \psi & \chi
	\end{pmatrix}
\end{align} is equal to \( T_e \) with \( e \in S_4 \) for \( d \to d D \).

As in Appendix~\ref{app:iso_spectrum_2}, our first lemma states that \( \mathcal{T}_e \) with \( e \in S_4 \) is block triangular. The definition of blocks now arises from the indexing of rows and columns in base \( 24 \). In particular, any diagonal block (but the first) is defined by \( i_h = \cdots = i_{j + 1} = o_h = \cdots = o_{j + 1} = 0 \) and \( i_j = o_j \neq 0 \). There are four classes of diagonal blocks: \begin{itemize}
    \item The first diagonal block is in its own class. It is given by \begin{align}
        \left ( I \otimes \bra{0}^{\otimes (h - 1)} \otimes \bra{0} \right ) \bigl ( M_h \cdots M_1 \bigr ) \left ( \ket{+} \otimes I \otimes \ket{0}^{\otimes (h - 1)} \right ) = \bigl ( I \otimes \bra{0} \bigr ) M \bigl ( \ket{+} \otimes I \bigr ) = N.
    \end{align}
    \item The second class of diagonal blocks corresponds to transpositions. \( i_j \) and \( o_j \) correspond to the same transposition. There are six sub-blocks in this class because there are six different transpositions in \( S_4 \).
    \item The third class of diagonal blocks corresponds to permutations with a single fixed point. \( i_j \) and \( o_j \) correspond to any of the two permutations with the same single fixed point. There are four sub-blocks in this class because there are four different choices of a single fixed point.
    \item The fourth class of diagonal blocks corresponds to permutations with no fixed point. \( i_j \) and \( o_j \) correspond to any of the nine permutations with no fixed point.
\end{itemize}

In particular, we will prove that a block of \( \mathcal{T}_e \) is zero if its defining row digit \( o_j \) is higher than its defining column digit \( i_j \) stand in a certain relation to each other. As the proof relies exclusively on Eq.~\eqref{eq:iso_spectrum_4_1}, \( \mathcal{T}_e \) again inherits its upper block triangularity from the upper triangularity of \( N \).

\begin{lemma} \label{lem:iso_spectrum_4_1}
    \( \mathcal{T}_e \) with \( e \in S_4 \) is upper block triangular.
\end{lemma}

\begin{proof}
    From Eqs.~\eqref{eq:iso_spectrum_4_1}, it follows that \begin{align}
        \left ( I^{\otimes (j - 1)} \otimes \bra{o_j} \otimes \bra{0}^{\otimes (h - j)} \otimes \bra{0} \right ) \bigl ( M_h \cdots M_1 \bigr ) \left ( \ket{+} \otimes I^{\otimes (j - 1)} \otimes \ket{i_j} \otimes \ket{0}^{\otimes (h - j)} \right ) = 0
    \end{align} if \begin{itemize}
        \item \( i_j \) corresponds to the trivial permutation and \( o_j \) does not,
        \item \( i_j \) corresponds to a transposition and \( o_j \) corresponds to a different transposition or a permutation with one or no fixed point,
        \item \( i_j \) corresponds to a permutation with a single fixed point and \( o_j \) corresponds to a permutation with a different single fixed point or no fixed point,
    \end{itemize} which concludes the proof.
\end{proof}

In our second lemma, we utilize the block triangularity of \( \mathcal{T}_e \) with \( e \in S_4 \) to make direct a statement about its two leading eigenvalues.

\begin{lemma} \label{lem:iso_spectrum_4_2}
	Let \( | \lambda_1 | > | \lambda_2 | > \cdots \geq 0 \) denote the distinct eigenvalues of \( \mathcal{T}_e \) with \( e \in S_4 \). Then, for any \( h \), \( \lambda_1 = 1 \) and \( \lambda_2 = \beta \). Furthermore, \( \lambda_1 \) is nondegenerate, and \( \lambda_2 \) has a degeneracy of six.
\end{lemma}

\begin{proof}
    As is the case for \( \mathcal{T}_e \) with \( e \in S_2 \), the spectrum of \( \mathcal{T}_e \) with \( e \in S_4 \) is given by the union of the spectra of its diagonal blocks. The two leading eigenvalues of the first diagonal block \( N \) are \( 1 \) and \( \beta \). In the following, we show that any other diagonal block can be written as a product of \( \beta \) and a matrix whose spectral radius is strictly bounded by \( 1 \), implying that its eigenvalues are strictly smaller than \( \beta \). We again structure the proof in steps.
    
    \begin{step}
        We show that any diagonal block but the first can be written as a product of beta and a matrix. From Eq.~\eqref{eq:iso_spectrum_4_1}, it follows that, \begin{itemize}
            \item if \( i_j \) and \( o_j \) correspond to the same transposition, \begin{align}
                & \left ( I^{\otimes (j - 1)} \otimes \bra{o_j} \otimes \bra{0}^{\otimes (h - j)} \otimes \bra{0} \right ) \bigl ( M_h \cdots M_1 \bigr ) \left ( \ket{+} \otimes I^{\otimes (j - 1)} \otimes \ket{i_j} \otimes \ket{0}^{\otimes (h - j)} \right ) \nonumber \\
                & \qquad = \left ( I^{\otimes (j - 1)} \otimes \bra{o_j} \otimes \bra{0} \right ) \bigl ( M_j \cdots M_1 \bigr ) \left ( \ket{+} \otimes I^{\otimes (j - 1)} \otimes \ket{i_j} \right ) \\
                & \qquad = \beta \left ( I^{\otimes (j - 1)} \otimes \bra{o_j} \right ) \bigl ( M_{j - 1} \cdots M_1 \bigr ) \left ( \ket{+} \otimes I^{\otimes (j - 1)} \right ),
            \end{align}
            \item if \( i_j \) and \( o_j \) correspond to any two permutations with the same single fixed point, \begin{align}
                & \left ( I^{\otimes (j - 1)} \otimes \bra{o_j} \otimes \bra{0}^{\otimes (h - j)} \otimes \bra{0} \right ) \bigl ( M_h \cdots M_1 \bigr ) \left ( \ket{+} \otimes I^{\otimes (j - 1)} \otimes \ket{i_j} \otimes \ket{0}^{\otimes (h - j)} \right ) \nonumber \\
                & \qquad = \left ( I^{\otimes (j - 1)} \otimes \bra{o_j} \otimes \bra{0} \right ) \bigl ( M_j \cdots M_1 \bigr ) \left ( \ket{+} \otimes I^{\otimes (j - 1)} \otimes \ket{i_j} \right ) \\
                & \qquad = p \left ( I^{\otimes (j - 1)} \otimes \bra{o_j} \right ) \bigl ( M_{j - 1} \cdots M_1 \bigr ) \left ( \ket{+} \otimes I^{\otimes (j - 1)} \right ) \\
                & \qquad < \frac{\beta}{2} \left ( I^{\otimes (j - 1)} \otimes \bra{o_j} \right ) \bigl ( M_{j - 1} \cdots M_1 \bigr ) \left ( \ket{+} \otimes I^{\otimes (j - 1)} \right ),
            \end{align} where \( \{ \varepsilon, \zeta \} \ni p < \beta / 2 \)~\cite{GitHub} depends on \( i_j \) and \( o_j \),
            \item if \( i_j \) and \( o_j \) correspond to any permutation with no fixed point, \begin{align}
                & \left ( I^{\otimes (j - 1)} \otimes \bra{o_j} \otimes \bra{0}^{\otimes (h - j)} \otimes \bra{0} \right ) \bigl ( M_h \cdots M_1 \bigr ) \left ( \ket{+} \otimes I^{\otimes (j - 1)} \otimes \ket{i_j} \otimes \ket{0}^{\otimes (h - j)} \right ) \nonumber \\
                & \qquad = \left ( I^{\otimes (j - 1)} \otimes \bra{o_j} \otimes \bra{0} \right ) \bigl ( M_j \cdots M_1 \bigr ) \left ( \ket{+} \otimes I^{\otimes (j - 1)} \otimes \ket{i_j} \right ) \\
                & \qquad = p \left ( I^{\otimes (j - 1)} \otimes \bra{o_j} \right ) \bigl ( M_{j - 1} \cdots M_1 \bigr ) \left ( \ket{+} \otimes I^{\otimes (j - 1)} \right ) \\
                & \qquad < \frac{\beta}{9} \left ( I^{\otimes (j - 1)} \otimes \bra{o_j} \right ) \bigl ( M_{j - 1} \cdots M_1 \bigr ) \left ( \ket{+} \otimes I^{\otimes (j - 1)} \right ),
            \end{align} where \( \{ \mu, \nu, \xi, o, \pi, \tau, \varphi, \chi, \psi \} \ni p < \beta / 9 \)~\cite{GitHub} depends on \( i_j \) and \( o_j \).
        \end{itemize}
    \end{step}
    
    \begin{step}
        We now argue that the spectral radius of the matrix \begin{align} \label{eq:iso_spectrum_4_3}
            \left ( I^{\otimes (j - 1)} \otimes \bra{o_j} \right ) \bigl ( M_{j - 1} \cdots M_1 \bigr ) \left ( \ket{+} \otimes I^{\otimes (j - 1)} \right )
        \end{align} is strictly bounded by \( 1 \) for \( 2 \leq j \leq h \). It holds that the spectral radius of \( M \)  is bounded by \( 1 \). While the first column of \( | M | \) sums to \( 1 \), the sums of the other columns are strictly bounded by \( 1 \)~\cite{GitHub}. As a result, the spectral radius of \( M_{j - 1} \cdots M_1 \) is bounded by \( 1 \). The boundary condition \( \ket{o} \) does not affect this because it specifies a subset of columns of the matrix \( M_{j - 1} \cdots M_1 \). In fact, it imposes a strict bound because this subset does not include the only column summing to \( 1 \). Also the boundary condition \( \bra{+} \) does not affect the bound on the spectral radius. The boundary condition means that Eq.~\eqref{eq:iso_spectrum_4_3} is a sum of matrices. Each of these matrices comprises a disjoint subset of rows of the matrix \( M_{j - 1} \cdots M_1 \). Because the spectral radius of \( M_{j - 1} \cdots M_1 \), the matrix comprising the whole set of rows, is bounded by \( 1 \), so is the sum of the matrices comprising the disjoint subsets.
    \end{step}
    
    \begin{step}
        \( 1 \) and \( \beta \) are the two leading eigenvalues of the first diagonal block \( N \). They are nondegenerate. Because the spectral radius of the matrix in Eq.~\eqref{eq:iso_spectrum_4_3} is strictly bounded by \( 1 \), the eigenvalues of any other diagonal block are strictly smaller than \( \beta \).
    \end{step}
    
    \noindent The statement follows.
\end{proof}

\section{PROOF OF RESULT~\ref{res:iso_renyi}} \label{app:iso_renyi}

\isorenyi*

\begin{proof}
    We split the proof into four steps, following the structure of the proof of Result~\ref{res:mps_renyi}. The steps are overall very similar to those of that proof.
    
    \begin{step} \label{step:iso_renyi_1}
        We rewrite \( \mathbb{E} I_2 (A : B) \) in terms of expressions of the form of Eq.~\eqref{eq:iso_caffe}. As in one dimension, we make the assumption that \( \mathbb{E} \log (X) = \log (\mathbb{E} X) \). Then, \begin{align}
            \mathbb{E} I_2 (A : B) = \log \left [ \mathbb{E} \tr \left ( \varrho_{A B}^2 \right ) \right ] - \log \left [ \mathbb{E} \tr \left ( \varrho_A^2 \right ) \right ] - \log \left [ \mathbb{E} tr \left ( \varrho_B^2 \right ) \right ].
        \end{align} \( \mathbb{E} \tr \left ( \varrho_A^2 \right ) \), \( \mathbb{E} \tr \left ( \varrho_B^2 \right ) \), and \( \mathbb{E} \tr \left ( \varrho_{A B}^2 \right ) \) can be written in the desired form [see Eqs.~\eqref{eq:mps_renyi_2} and \eqref{eq:mps_renyi_3}].
    \end{step}
    
    \begin{step} \label{step:iso_renyi_2}
        We express \( \mathbb{E} I_2 (A : B) \) in terms of the transfer tensors defined in Sec.~\ref{sec:iso_transfer} and use Eq.~\eqref{eq:iso_to_mps} to map contractions of two-dimensional tensor networks to multiplications of matrices. The latter is enabled by our definition of subsystems \( A \) and \( B \) [see Fig.~\ref{fig:iso_setup}~(a)]. We have done this for \( \mathbb{E} \tr \left ( \varrho_A^2 \right ) \) in graphical notation in Sec.~\ref{sec:iso_renyi} [see Eqs.~\eqref{eq:iso_renyi_1} and \eqref{eq:iso_renyi_2}]. It is easy to confirm that \begin{align}
            \mathbb{E} I_2 (A : B) = \log \left ( \fv{\mathcal{I}_2}{\mathcal{T}_e^c \mathcal{T}_{(1 2)}^a \mathcal{T}_e^r \mathcal{T}_{(1 2)}^b}{\mathcal{F}_2} \right ) - \log \left ( \fv{\mathcal{I}_2}{\mathcal{T}_e^c \mathcal{T}_{(1 2)}^a}{\mathcal{F}_2} \right ) - \log \left ( \fv{\mathcal{I}_2}{\mathcal{T}_e^{c + a + r} \mathcal{T}_{(1 2)}^b}{\mathcal{F}_2} \right ).
        \end{align}
    \end{step}
    
    \begin{step} \label{step:iso_renyi_3}
        We expand \( \mathbb{E} I_2 (A : B) \) in terms of the spectrum of \( \mathcal{T}_e \) with \( e \in S_2 \), which we consider in Appendix~\ref{app:iso_spectrum_2}. Because we know \( \lambda_1 \) and \( \lambda_2 \) as well as their algebraic and geometric multiplicities, we do not need \( \mathcal{T}_e \) to be diagonalizable. Expanding \( \mathcal{T}_e^c \) and taking the limit of \( c \to \infty \) yields \begin{align}
    	    \mathbb{E} I_2 (A : B) & = \log \left ( \fv{L_1}{\mathcal{T}_{(1 2)}^a \mathcal{T}_e^r \mathcal{T}_{(1 2)}^b}{\mathcal{F}_2} \right ) - \log \left ( \fv{L_1}{\mathcal{T}_{(1 2)}^a}{\mathcal{F}_2} \right ) - \log \left ( \fv{L_1}{\mathcal{T}_{(1 2)}^b}{\mathcal{F}_2} \right ),
    	\end{align} where we have used that \( \braket{\mathcal{I}_2}{R_1} = 1 \). After expanding \( \mathcal{T}_e^r \) and using that \( \ket{\mathcal{F}_2} = \ket{R_1} \), we have \begin{align}
    		I_2 (A : B) & = \log \left [ \fv{L_1}{\mathcal{T}_{(1 2)}^a}{R_1} \fv{L_1}{\mathcal{T}_{(1 2)}^b}{R_1} + \lambda_2^r \fv{L_1}{\mathcal{T}_{(1 2)}^a}{R_2} \fv{L_2}{\mathcal{T}_{(1 2)}^b}{R_1} + O \left ( r^{v - 1} \lambda_3^r \right ) \right ] \nonumber \\
    		& \hphantom{{} = {}} - \log \left ( \fv{L_1}{\mathcal{T}_{(1 2)}^a}{R_1} \right ) - \log \left ( \fv{L_1}{\mathcal{T}_{(1 2)}^b}{R_1} \right ),
    	\end{align} where \( v \) denotes the size of the largest Jordan block with respect to \( \lambda_3 \).
    \end{step}
    
    \begin{step} \label{step:iso_renyi_4}
        Finally, we can write \( \mathbb{E} I_2 (A : B) \) in the form of Definition~\ref{def:definition}. With \( \Lambda = \max \left ( \left \{ \lambda_2^{2 r}, r^{v - 1} \lambda_3^r \right \} \right ) \), \begin{align}
    		\mathbb{E} I_2 (A : B) & = \log \left [ 1 + \lambda_2^r \frac{\fv{L_1}{\mathcal{T}_{(1 2)}^a}{R_2} \fv{L_2}{\mathcal{T}_{(1 2)}^b}{R_1}}{\fv{L_1}{\mathcal{T}_{(1 2)}^a}{R_1} \fv{L_1}{\mathcal{T}_{(1 2)}^b}{R_1}} + O \left ( r^{v - 1} \lambda_3^r \right ) \right ] \\
    		& = \lambda_2^r \frac{\fv{L_1}{\mathcal{T}_{(1 2)}^a}{R_2} \fv{L_2}{\mathcal{T}_{(1 2)}^b}{R_1}}{\fv{L_1}{\mathcal{T}_{(1 2)}^a}{R_1} \fv{L_1}{\mathcal{T}_{(1 2)}^b}{R_1}} + O (\Lambda) \\
    		& \equiv K \exp \left ( - \frac{r}{\xi} \right ) + O \left [ \exp \left ( - \frac{r}{\chi} \right ) \right ],
    	\end{align} where \begin{align}
    		K = \frac{\fv{L_1}{\mathcal{T}_{(1 2)}^a}{R_2} \fv{L_2}{\mathcal{T}_{(1 2)}^b}{R_1}}{\fv{L_1}{\mathcal{T}_{(1 2)}^a}{R_1} \fv{L_1}{\mathcal{T}_{(1 2)}^b}{R_1}}
    	\end{align} and \begin{align}
    		\xi = - \frac{1}{\log (\lambda_2)} = - \left [ \log \left ( \frac{d D^3 - d D}{d^2 D^4 - 1} \right ) \right ]^{- 1} = \xtwo > \chi.
    	\end{align}
    \end{step}
	
	\noindent This concludes the proof.
\end{proof}

\section{PROOF OF RESULT~\ref{res:iso_norm}} \label{app:iso_norm}

\isonorm*

\begin{proof}
    We split the proof into four steps, following the structure of the proof of Result~\ref{res:mps_renyi}. The steps are overall very similar to those of the proof of Result~\ref{res:mps_norm} (see Appendix~\ref{app:mps_norm}).

    \begin{step} \label{step:iso_norm_1}
        We rewrite \( \mathbb{E} N (A : B) \) in terms of expressions of the form of Eq.~\eqref{eq:iso_caffe}. As in one dimension, with the Hilbert-Schmidt inner product, \begin{align}
            \mathbb{E} N (A : B) = \mathbb{E} \tr \left ( \varrho_{A B}^2 \right ) + \mathbb{E} \tr \left ( \varrho_A^2 \right ) \tr \left ( \varrho_B^2 \right ) - 2 \mathbb{E} \tr \left [ \varrho_{A B} \left ( \varrho_A \otimes \varrho_B \right ) \right ].
        \end{align} The right-hand side can be written in the desired form [see Eq.~\eqref{eq:mps_norm_5}].
    \end{step}
    
    \begin{step} \label{step:iso_norm_2}
        We express \( \mathbb{E} N (A : B) \) in terms of the transfer tensors defined in Sec.~\ref{sec:iso_transfer} and use Eq.~\eqref{eq:iso_to_mps} to map contractions of two-dimensional tensor networks to multiplications of matrices. The latter is enabled by our definition of subsystems \( A \) and \( B \) [see Fig.~\ref{fig:iso_setup}~(a)]. It is easy to confirm that \begin{align}
    		\mathbb{E} N (A : B) & = \fv{\mathcal{I}_4}{\mathcal{T}_e^c \mathcal{T}_{(1 2)}^a \mathcal{T}_e^r \mathcal{T}_{(1 2)}^b}{\mathcal{F}_4} + \fv{\mathcal{I}_4}{\mathcal{T}_e^c \mathcal{T}_{(3 4)}^a \mathcal{T}_e^r \mathcal{T}_{(1 2)}^b}{\mathcal{F}_4} - 2 \fv{\mathcal{I}_4}{\mathcal{T}_e^c \mathcal{T}_{(1 2)}^a \mathcal{T}_e^r \mathcal{T}_{(1 3)}^b}{\mathcal{F}_4} \\
    		& = \fv{\mathcal{I}_4}{\mathcal{T}_e^c \mathcal{T}_{(1 2)}^a \mathcal{T}_e^r \left ( \mathcal{T}_{(1 2)}^b + \mathcal{T}_{(3 4)}^b - 2 \mathcal{T}_{(1 3)}^b \right )}{\mathcal{F}_4} \\
        	& \equiv \fv{\mathcal{I}_4}{\mathcal{T}_e^c \mathcal{A} \mathcal{T}_e^r \mathcal{B}}{\mathcal{F}_4},
    	\end{align} where we have defined \begin{align}
    		\mathcal{A} = \mathcal{T}_{(1 2)}^a \qquad \text{and} \qquad \mathcal{B} = \mathcal{T}_{(1 2)}^b + \mathcal{T}_{(3 4)}^b - 2 \mathcal{T}_{(1 3)}^b.
    	\end{align}
    \end{step}
    
    \begin{step} \label{step:iso_norm_3}
        We expand \( \mathbb{E} N (A : B) \) in terms of the spectrum of \( \mathcal{T}_e \) with \( e \in S_4 \), which we consider in Appendix~\ref{app:iso_spectrum_4}. Because we know \( \lambda_1 \) and \( \lambda_2 \) as well as their algebraic and geometric multiplicities, we do not need \( \mathcal{T}_e \) to be diagonalizable. Expanding \( \mathcal{T}_e^c \) and taking the limit of \( c \to \infty \) yields \begin{align}
    	    \mathbb{E} N (A : B) & = \fv{L_1}{\mathcal{A} \mathcal{T}_e^r \mathcal{B}}{\mathcal{F}_4},
    	\end{align} where we have used that \( \braket{\mathcal{I}_4}{R_1} = 1 \). After expanding \( T_e^r \) and using that \( \ket{\mathcal{F}_4} = \ket{R_1} \), we have \begin{align}
    	    \mathbb{E} N (A : B) & = \fv{L_1}{\mathcal{A}}{R_1} \fv{L_1}{\mathcal{B}}{R_1} + \lambda_2^r \sum_{\mu = 1}^{w_2} \fv{L_1}{\mathcal{A}}{R_2^{(\mu)}} \fv{L_2^{(\mu)}}{\mathcal{B}}{R_1} + O \left ( \lambda_3^{r} \right ) \\
    	    & = \lambda_2^r \sum_{\mu = 1}^{w_2} \fv{L_1}{\mathcal{A}}{R_2^{(\mu)}} \fv{L_2^{(\mu)}}{\mathcal{B}}{R_1} + O \left ( \lambda_3^{r} \right ),
    	\end{align} where, in the second line, we have used that \begin{align} \label{eq:iso_temp}
    	    \fv{L_1}{\mathcal{B}}{R_1} = 0,
    	\end{align} which we prove in the following.
    	
    	As in the proof of Result~\ref{res:mps_norm} (see Appendix~\ref{app:mps_norm}), we prove that \( \fv{L_1}{\mathcal{T}_t^b}{R_1} \) does not depend on the two elements the transposition \( t \in S_4 \) acts upon. In fact, the proof follows from the proof of Eq.~\eqref{eq:mps_temp} of that Appendix. We just need two additional considerations. First, the proof of Lemma~\ref{lem:iso_spectrum_4_1} is not specific to the trivial permutation \( e \in S_4 \). In particular, \( \mathcal{T}_\rho \) exhibits an upper block triangular structure for any \( \rho \in S_4 \). The first diagonal block is given by \( T_\rho \) with \( \rho \in S_4 \), \begin{align}
    	    \Bigl ( \fv{s_i}{\mathcal{T}_\rho}{s_j} \Bigr )_{\substack{1 \leq i \leq 24 \\ 1 \leq j \leq 24}} = T_\rho,
    	\end{align} which implies that \begin{align}
    	    \Bigl ( \fv{s_i}{\mathcal{T}_\rho^b}{s_j} \Bigr )_{\substack{1 \leq i \leq 24 \\ 1 \leq j \leq 24}} = T_\rho^b.
    	\end{align} Second, the eigenvectors \( \bra{L_1} \) and \( \ket{R_1} \) of \( \mathcal{T}_e \) with \( e \in S_4 \) arise from those of \( T_e \) with \( e \in S_4 \). That is, \begin{align}
    	    \ket{R}_1 = s_1 \qquad \text{and} \qquad \Bigl ( \braket{L_1}{s_i} \Bigr )_{1 \leq i \leq 7} = \left ( 1 \quad \frac{\alpha}{\beta - 1} \quad \frac{\alpha}{\beta - 1} \quad \frac{\alpha}{\beta - 1} \quad \frac{\alpha}{\beta - 1} \quad \frac{\alpha}{\beta - 1} \quad \frac{\alpha}{\beta - 1} \right ).
    	\end{align} \( \fv{L_1}{\mathcal{T}_t^b}{R_1} \) thus does not depend on the two elements the transposition \( t \in S_4 \) acts upon because \( \fv{L_1}{T_t^b}{R_1} \) does not. This implies that \begin{align}
    	    \fv{L_1}{\mathcal{B}}{R_1} = \fv{L_1}{\left ( \mathcal{T}_{(1 2)}^b + \mathcal{T}_{(3 4)}^b - 2 \mathcal{T}_{(1 3)}^b \right )}{R_1} = 0.
    	\end{align}
    \end{step}
    
    \begin{step} \label{step:iso_norm_4}
        Finally, we can write \( \mathbb{E} N (A : B) \) in the form of Definition~\ref{def:definition}. That is, \begin{align}
            \mathbb{E} N (A : B) \equiv K \exp \left ( - \frac{r}{\xi} \right ) + O \left [ \exp \left ( - \frac{r}{\chi} \right ) \right ],
        \end{align} where \begin{align}
    		K = \sum_{\mu = 1}^{w_2} \fv{L_1}{\mathcal{A}}{R_2^{(\mu)}} \fv{L_2^{(\mu)}}{\mathcal{B}}{R_1}
    	\end{align} and \begin{align}
    		\xi = - \frac{1}{\log (\lambda_2)} = - \left [ \log \left ( \frac{d D^3 - d D}{d^2 D^4 - 1} \right ) \right ]^{- 1} = \xtwo > \chi.
        \end{align}
    \end{step}
    
    \noindent This concludes the proof.
\end{proof}

\end{document}